\documentclass[pra,reprint,superscriptaddress,nofootinbib,longbibliography]{revtex4-1}
\usepackage{amsthm}
\usepackage{amssymb}
\usepackage{amsmath}
\usepackage{color}
\usepackage{graphicx}
\usepackage{hyperref}
\usepackage{comment}
\usepackage{appendix}
\newtheorem{thm}{Theorem}
\newtheorem{lem}{Lemma}

\theoremstyle{definition}
\newtheorem{assumption}{Assumption}
\newcommand{\ket}[1]{|#1\rangle}
\newcommand{\bra}[1]{\langle#1|}
\newcommand{\eqnref}[1]{Eq.~(\ref{#1})}
\newcommand{\Uf}{U_{\mathrm{f}}}
\newcommand{\U}{\mathrm{U}}
\newcommand{\WUf}{\widetilde{U}_{\mathrm{f}}}
\newcommand{\cH}{\mathcal{H}}

\makeatletter
\newcommand{\neutralize}[1]{\expandafter\let\csname c@#1\endcsname\count@}
\makeatother

\newenvironment{thmbis}[1]
  {\renewcommand{\thethm}{\ref*{#1}$'$}%
   \neutralize{thm}\phantomsection
   \begin{thm}}
  {\end{thm}}

\begin{document}
\title{Pre-thermal Phases of Matter Protected by Time-Translation Symmetry 
%Pre-thermal Time Crystals and Floquet topological phases without disorder
}

\author{Dominic V. Else}
\affiliation{Physics Department, University of California,  Santa Barbara, California 93106, USA}

\author{Bela Bauer}
\affiliation{Station Q, Microsoft Research, Santa Barbara, California 93106-6105, USA}

\author{Chetan Nayak}
\affiliation{Station Q, Microsoft Research, Santa Barbara, California 93106-6105, USA}
\affiliation{Physics Department, University of California,  Santa Barbara, California 93106, USA}

\begin{abstract}
    In a periodically driven (Floquet) system, there is the possibility for new
    phases of matter,
    not present in stationary systems, protected by discrete time-translation
    symmetry. This includes topological phases protected in part by time-translation symmetry, as well
    as phases distinguished by the spontaneous breaking of this symmetry, dubbed ``Floquet time crystals''.
	 We show that such phases of matter can exist in the pre-thermal regime of periodically-driven systems, which
	 exists generically for sufficiently large drive frequency, thereby eliminating the need for
    integrability or strong quenched disorder that limited previous constructions.
    We prove a theorem that states that such a pre-thermal regime persists until
    times that are nearly exponentially-long in the ratio of certain couplings
    to the drive frequency. By similar techniques, we can also construct stationary systems which spontaneously break \emph{continuous} time-translation symmetry.
    %After this thermalization time, the system will
    %generically heat to the trivial infinite temperature state. However, during the time interval
    %prior to that, phases protected by time-translation symmetry
%can exist even without disorder.
    We argue furthermore that for driven systems coupled to a cold bath, the pre-thermal regime could potentially persist to infinite time.
\end{abstract}

\maketitle
\section{Introduction}

Much of condensed matter physics revolves around determining which distinct \emph{phases of matter}
can exist as equilibrium states of physical systems. Within a phase, the properties of the system
vary continuously as external parameters are varied, while different phases
are separated by phase transitions, at which the properties change abruptly.
An extremely rich set of observed phases can be
characterized by symmetry. The best known example is \emph{spontaneous symmetry-breaking},
as a result of which the equilibrium state of the system is less symmetrical than the Hamiltonian.
More recently, a set of uniquely quantum phases---\emph{symmetry-protected topological} (SPT) phases~\cite{Gu2009,Pollmann2010,Pollmann2012,Fidkowski2010,Chen2010,Chen2011,Schuch2011a,Fidkowski2011, Chen2011b,Chen2013,Levin2012,Vishwanath2013,Wang2014, Kapustin2014,Gu2014,Else2014,Burnell2014,Wang2015,Cheng2015},
including topological insulators~\cite{Hasan10,Qi11},
and \emph{symmetry-enriched topological} (SET) phases~\cite{Maciejko2010,Essin2013,Lu2013,Mesaros2013,Hung2013,Barkeshli2014,Cheng2015a}---has been discovered. These phases, while symmetric, manifest the symmetry in subtly
anomalous ways, and are distinct only as long as the symmetry is preserved. We can collectively refer to these three classes of phases
as \emph{symmetry-protected phases of matter}.

Thus far, the concept of symmetry-protected phases of matter has
not been as succesful in describing systems away from equilibrium. Recently, however, it was realized that
certain periodically-driven ``Floquet'' systems can exhibit distinct phases,
akin to those of equilibrium systems~\cite{Khemani15b}.
In this paper, we show that there is, in fact, a very general set of non-equilibrium conditions
under which such phases can arise, due to a remarkable phenomenon called "pre-thermalization".
In Floquet systems, pre-thermalization occurs when a time-dependent change of basis removes
all but a small residual time-dependence from the Hamiltonian, and thus allows the properties of the system
to be mapped approximately onto those of a system in thermal equilibrium.
The residual time-dependence is nearly exponentially-small
in a large parameter $\alpha$ of the original Hamiltonian of the system.
One can then talk about a ``pre-thermal regime'' in which the system reaches a thermal equilibrium state with
respect to the approximate effective time-independent Hamiltonian that results from neglecting the small residual time dependence.
In this regime, the system can exhibit phases and phase transitions analogous to those seen in thermal equilibrium, such as
symmetry-protected phases. Nevertheless, in the original non-rotating frame, the system remains
\emph{very far from thermal equilibrium}
with respect to the instantaneous Hamiltonian at any given time.
After the characteristic time $t_*$, which is nearly exponentially-long in
the large parameter $\alpha$, other physics (related the residual time-dependence) takes
over.

In this paper, we show that pre-thermal systems
can also exhibit phases of matter that cannot exist in thermal equilibrium.
These novel phases can also be understood as symmetry-protected phases
but of a variety that cannot occur in thermal equilibrium: these phases are protected by
\emph{discrete time-translation symmetry}. While these include topological phases protected
by time-translation symmetry~\cite{vonKeyserlingk16a,Else16a,Potter16,Roy16}, perhaps the most dramatic of these are
``time crystals'' that spontaneously break time-translation symmetry.
The idea of time crystals that spontaneously break \emph{continuous} time-translation symmetry was first proposed by Wilczek and Shapere \cite{Wilczek12,Shapere12},
but finding a satisfactory \emph{equilibrium}
model has proven difficult and some no-go theorems exist \cite{Li13,Bruno13a,Bruno13b,Bruno2013,Nozieres13,Volovik2013,Watanabe15}.
In this paper, we construct pre-thermal ``Floquet time crystals'', which spontaneously break
the \emph{discrete time-translation symmetry} of periodically-driven systems~\cite{Else2016b}
\footnote{For an alternative view of such systems that focuses on other symmetries of the discrete time-translation
operator, see Refs. \onlinecite{Khemani15b,vonKeyserlingk16b,vonKeyserlingk2016a}.}.
Floquet time crystals are the focus of
this paper, but as a by-product of our analysis, we also find pre-thermal -- i.e. non-equilibrium --
time crystals that spontaneously break \emph{continuous} time-translation
symmetry. We also construct SPT and SET phases protected by discrete time-translation
symmetry.

Periodically-driven systems have long been considered an unlikely place to find
interesting phases of matter and phase transitions
since generic driven closed systems will heat up to infinite
temperature \cite{DAlessio2014,Lazarides2014,Ponte2015a}.
It has been known that the heating problem can be avoided \cite{Abanin2014,Ponte15a,Ponte15b,Lazarides15,Iadecola15}
if the system is integrable or if the system has sufficiently strong quenched disorder that it undergoes
many-body localization (MBL) \onlinecite{Basko06a,Basko06b,Oganesyan07,Znidaric08,Pal10,Bardarson12,Bauer13,Serbyn13a,Serbyn13b,Huse14}.
However, integrability relies on fine-tuning, and MBL requires the system to be completely decoupled from the environment \cite{Levi2015,Fischer2015,Nandkishore2014,Gopalakrishnan2014,Johri2015,Nandkishore2015b,Li2015,Nandkishore2016,Hyatt2016}.
Furthermore, the disorder must be sufficiently strong,
which may be difficult to realize in an experiment but does not constitute fine-tuning.

The central result of this paper is therefore to show that pre-thermalization makes it possible for
non-equilibrium phases protected by time-translation symmetry
to occur in more generic non-equilibrium systems without the need for fine-tuning, strong disorder,
or complete decoupling from the environment.
Remarkably, these non-equilibrium phases and phase transitions, which have
have no direct analogues in thermal equilbrium,
have a mathematical formulation that is identical to that of equilibrium phases,
though with a different physical interpretation.
Since MBL is not a requirement, it is conceivable that pre-thermal
time-translation protected phases could survive the presence of coupling
to an environment. In fact, we will discuss a plausible scenario by which these phases can actually be \emph{stabilized} by coupling to a sufficiently cold thermal bath, such that the system remains in the pre-thermal regime even at infinite time.

The structure of the paper will be as follows. In Section
\ref{prethermalization_results}, we state our
main technical result. In Section \ref{sec_floquet_time_crystal}, we apply this to construct prethermal
Floquet time crystals which spontaneously break discrete time-translation
symmetry. In Section \ref{sec_continuous_time_crystal}, we show that a \emph{continuous} time-translation
symmetry can also also be spontaneously broken in the pre-thermal regime for a
system with a time-independent Hamiltonian. In Section \ref{sec_sxt}, we outline how
our methods can also be applied to construct SPT and SET phases protected by
time-translation symmetry. In Section \ref{sec_open}, we discuss what we expect to
happen for non-isolated systems coupled to a cold thermal bath.
Finally, we discuss implications and interpretations in Section \ref{sec_discussion}.

\section{Pre-Thermalization Results}
\label{prethermalization_results}

The simplest incarnation of pre-thermalization occurs in periodically-driven systems when the
driving frequency $\nu$ is much larger than all of the local energy scales of the instantaneous Hamiltonian
\cite{Abanin2015,Abanin2015a,Abanin2015b,Kuwahara2015,Mori2016}
(see also Refs.~\onlinecite{bukov2015,canovi2016,bukov2016} for numerical results).
The key technical result of our paper will be a theorem generalizing these results
to other regimes in which the driving frequency is not greater than \emph{all} the local scales of the Hamiltonian, but there is nevertheless some separation of energy scales. This will allow us to show that time-translation protected phases can exist in the pre-thermal regime. More precisely, in the models that we construct,
one local coupling strength is large and the others are small; the drive frequency is large compared to the
small couplings, and the parameter $\alpha$ is the ratio of the drive frequency to the largest of
the small local couplings. The term in the
Hamiltonian with large coupling must take a special form, essentially that of a symmetry generator,
that allows it to avoid heating the system.

Accordingly, we will consider a time-dependent Hamiltonian of the form $H(t) = H_0(t) + V(t)$, where
$H_0(t)$ and $V(t)$ are periodic with period $T$. We assume that $\lambda T \ll 1$, where $\lambda$ is the local energy scale of $V$. We further assume that $H_0(t)$ has the property that it generates a trivial time evolution over $N$ time cycles: $U_0(NT,0) = U_0(T,0)^N = 1$, where
\begin{equation}
U_0(t_2,t_1) = \mathcal{T} \exp\left(-i\int_{t_1}^{t_2} H_0(t) \right) dt, \quad \mbox{$\mathcal{T} = $ time-ordering.}
\end{equation}
We claim that such a time evolution will exhibit pre-thermalizing behavior for $\lambda T \ll 1/N$ even if
the local energy scale of ${H_0}(t)$ is comparable to $1/T$. In other words, such a system
exhibits pre-thermalizing behavior when the frequency is large compared some of the couplings
(those in $V(t)$) but not others (those in ${H_0}(t)$), as promised in the introduction.

An easy way to see that this claim is true is to work in the interaction picture
(treating $V$ as the ``interaction''). Then we see that the time evolution of
the total Hamiltonian $H(t)$ over $N$ time cycles is given by
\begin{equation}
\label{UNT}
U(NT,0) = \mathcal{T} \exp\left(-i\int_{0}^{NT} V^{\mathrm{int}}(t) dt\right),
\end{equation}
where $V^{\mathrm{int}}(t) = U_0(0,t)^{\dagger} V(t) U_0(0,t)$ is the representation of $V(t)$ in the interaction picture, and $U_0(0,NT) = 1$ ensures that the time evolution operator \eqnref{UNT} is the same in the interaction and Schr\"odinger pictures. If we rescale time as $t \to t/\lambda$, then \eqnref{UNT} describes a system being driven at the large  frequency $\nu  = 1/(\lambda NT)$ by a drive of local strength 1, which by the results of Refs.~\onlinecite{Abanin2015,Abanin2015a,Abanin2015b,Kuwahara2015,Mori2016}
will exhibit pre-thermalizing behavior for $\nu \gg 1$.

On the other hand, since the above argument for pre-thermalization required coarse-graining the time period from $T$ to $NT$, it prevents us from identifying phases of matter, such as time crystals or Floquet SPT phases, that are protected by time translation symmetry. The problem is that the time-translation symmetry by $T$ is what allows different phases of matter to be sharply distinguished. This symmetry is still present, of course (because the coarse-graining is a feature of our description of the system, not the system itself), but it is no longer manifest. Therefore, it is not at all transparent how to understand the different phases of matter in this picture.

In order to proceed further, we will need a new approach. In this paper, we develop a \emph{new} formalism that analyzes $U(T,0)$ itself rather than $U(NT,0)$, allowing the effects of time-translation symmetry to be seen in a transparent way. Our central tool is a theorem that we will prove, substantially generalizing those of Abanin et al.\cite{Abanin2015a}.
A more precise version of our theorem will be given momentarily, and the proof will be given in
Appendix \ref{sec:proofs}; the theorem essentially states that
there exists a time-independent local unitary rotation $\mathcal{U}$ such that
$U_\text{f} \approx \WUf = \mathcal{U}^{\dagger}(X e^{-iDT})\,\mathcal{U}$, where $X = U_0(T,0)$ is the time evolution of $H_0$ over one time cycle, and $D$ is a quasi-local Hamiltonian that commutes with $X$. The dynamics at stroboscopic times are well-approximated by $\WUf$ for times $t \ll t_{*}$, where $t_{*} = e^{O(1/(\lambda T[\log(1/\lambda T)]^3))}$.
This result combines ideas in Ref.~\onlinecite{Abanin2015a}
about (1) the high-frequency limit of driven systems and (2) approximate symmetries
in systems with a large separation of scales. Recall that,
in the high-frequency limit of a driven system, the Floquet operator can be approximated by
the evolution (at stroboscopic times) due a time-independent Hamiltonian,
$\Uf \approx \exp(-iTH_{\mathrm{eff}})$. Meanwhile, in a static system with a large separation of scales,
$H = -u L + {D_0}$, where $u$ is much larger than the couplings in $D_0$ but $[L,{D_0}]\neq 0$,
Ref.~\onlinecite{Abanin2015a} shows that there is a unitary transformation
$\mathcal{U}$ such that $\mathcal{U} H \mathcal{U}^\dagger \approx -u L + D$ where $[L,D]=0$,
i.e. the system has an approximate symmetry generated by $\mathcal{U}^\dagger L\, \mathcal{U}$.
Our theorem states that, after a time-independent local unitary change of basis, a periodic Hamiltonian $H(t) = H_0(t) + V(t)$, with ${H_0}(t)$ satisfying the
condition given above, can be approximated, as far as the evolution at stroboscopic times is concerned,
by a binary drive that is composed of two components: (1) the action of ${H_0}(t)$ over one
cycle, namely $U_0(T,0)$ and (2) a static Hamiltonian that is invariant under the symmetry generated by ${U_0}(T,0)$. 

These results might seem surprising, because they imply that the evolution over one time period commutes with a symmetry $X = U_0(T,0)$ [or $\mathcal{U} X \mathcal{U}^{\dagger}$ in the original basis], despite the fact that the microscopic time-dependent Hamiltonian $H(t)$ had no such symmetry. We interpret this ``hidden'' symmetry as a shadow of the discrete time-translation symmetry. (For example, the evolution over $N$ time periods also commutes with $\mathcal{U} X \mathcal{U}^{\dagger}$, but if we add weak $NT$-periodic perturbations to break the discrete time-translation symmetry then this is no longer the case.) Thus, our theorem is precisely allowing us to get a handle on the implications of discrete time-translation symmetry. Compare Ref.~\cite{vonKeyserlingk2016a}, where a similar ``hidden'' symmetry was constructed for many-body-localized Floquet time crystals.

The preceding paragraphs summarize the physical meaning of our theorem.
A more precise statement of the theorem, although it is a bit more opaque physically,
is useful because it makes the underlying assumptions manifest. The statement of the theorem
makes use of an operator norm $\| O \|_{n}$ that measures the average over one Floquet cycle
of the size of the local terms whose sum makes up a Hamiltonian; the subscript $n$
parametrizes the extent to which the norm suppresses the weight of operators with larger spatial support.
An explicit definition of the norm is given in Appendix \ref{sec:proofs}. The theorem states the following.

\begin{thm}
\label{mainthm}Consider a periodically-driven system with Floquet operator:
\begin{equation}
U_\text{f} = \,\mathcal{T}\!\exp\!\left(-i\int_{0}^{T} H(t) dt\right)
\end{equation}
where $H(t) = H_0(t) + V(t)$, and $X \equiv U_0(0,T)$ satisfies $X^N=1$ for some
integer $N$. We assume that $H_0(t)$ can be written as a sum $H_0(t) = \sum_i
h_i(t)$ of terms acting only on single sites $i$.
Define $\lambda \equiv \| V \|_1$.
Assume that
\begin{equation}
    \lambda T \leq \frac{\gamma \kappa_1^2}{N+3}, \quad \gamma \approx 0.14.
\end{equation}
Then there exists a (time-independent) unitary $\mathcal{U}$ such that
\begin{equation}
\label{eqn:rotated-Floquet}
\mathcal{U}\, U_\text{f} \, \mathcal{U}^{\dagger} =
X \,\mathcal{T}\!\exp\!\left(-i\int_{0}^{T} [D + E+ V(t)] dt\right)
\end{equation}
where $D$ is local and $[D,X]=0$; $D, E$ are independent of time; and
\begin{align}
\| V \|_{n_*} &\leq \lambda  \left(\frac{1}{2}\right)^{n_*}\\
\| E \|_{n_*} &\leq \lambda  \left(\frac{1}{2}\right)^{n_*}
\end{align}
The exponent $n_*$ is given by
\begin{equation}
n_* = \frac{\lambda_0/\lambda}{[1+\log(\lambda_0/\lambda)]^3}, \quad \lambda_0 = 
\frac{(\kappa_1)^2}{72(N+3)(N+4) T}
\end{equation}
Furthermore,
\begin{equation}
\label{D_first_order}
\| D - \overline{V} \|_{n_*} \leq \mu (\lambda^2/\lambda_0), \quad \mu \approx 2.9,
\end{equation}
where
\begin{align}
\overline{V} &= \frac{1}{NT}\int_0^{NT} V^{\mathrm{int}}(t) dt \cr
  &= \frac{1}{N}\sum_{k=0}^{N-1} X^{-k} \left(\frac{1}{T}\int_0^T V^{\mathrm{int}}(t) dt\right) X^{k}. \cr
\end{align}

\end{thm}
The proof is given in Appendix \ref{sec:proofs}. The statement of the theorem makes use of a number $\kappa_1$.
It is chosen so that $\|H\|_1$ is finite; the details are given when
the norm is given in Appendix \ref{sec:proofs}.

Unpacking the theorem a bit in order to make contact with the discussion above, we see that
it states that there is a time-independent unitary operator $\mathcal{U}$ that
transforms the Floquet operator into the form $X e^{-iDT}$ with $[D,X]=0$ and local $D$,
up to corrections that are exponentially small in ${n_*} \sim 1/(\lambda T [\ln(1/\lambda T)]^3)$.
These ``error terms'' fall into two categories: time-independent terms that do not commute with $X$, which are
grouped into $E$; and time-dependent terms, which are grouped into $V(t)$. Both types of corrections
are exponentially-small in $n_*$. Since they are exponentially-small $\| E \|_{n_*}, \| V \|_{n_*}\sim \left(1/2\right)^{n_*}$,
these terms do not affect the evolution of the system until exponentially-long times, $t_* \sim e^{C n_*}$
(for some constant $C$). It is not possible to find a time-independent unitary transformation that exactly transforms
the Floquet operator into the form $X e^{-iDT}$ because the system must, eventually, heat up to infinite temperature
and the true Floquet eigenstates are infinite-temperature states, not the eigenstates of an operator of the form
$X e^{-iDT}$ with local $D$. In the interim, however, the approximate Floquet operator $X e^{-iDT}$
leads to Floquet time crystal behavior, as we will discuss in the next Section.

The proof of Theorem 1 constructs $\mathcal{U}$ and $D$
through a recursive procedure,
which combines elements of the proofs of pre-thermalization in driven and undriven systems given by Abanin et al. ~\onlinecite{Abanin2015a}.

In the case of pre-thermal undriven systems, the theorem we need has essentially already been given in Ref.~\onlinecite{Abanin2015a}, but we will restate the result in a form analogous with Theorem \ref{mainthm}, which entails some slightly different bounds (however, they are easily derivable using the techniques of Ref.~\onlinecite{Abanin2015a}).

\begin{thm}
Consider a time-independent Hamiltonian $H$ of the form
\begin{equation}
H = -u L + V,
\end{equation}
where $e^{2\pi i L} = 1$. We assume that $L$ can be written as a sum $L = \sum_i
L_i$ of terms acting only on single sites $i$. Define $\lambda \equiv \| V \|_1$, and assume that
\begin{equation}
\lambda/u \leq \gamma \kappa_1^2, \quad \gamma \approx 0.14.
\end{equation}
Then there exists a local unitary transformation $\mathcal{U}$ such that
\begin{equation}
\mathcal{U} H \mathcal{U}^{\dagger} = -u L + D + \hat{V}
\end{equation}
where $[L,D]=0$ and $\hat{V}$ satisfies
\begin{equation}
\| \hat{V} \|_{n_*} \leq \lambda  \left(\frac{1}{2}\right)^{n_*}\\
\end{equation}
where
\begin{equation}
n_* = \frac{\lambda_0/\lambda}{[1+\log(\lambda_0/\lambda)]^3}, \quad \lambda_0 = \frac{u \kappa_1^2}{144}.
\end{equation}
Furthermore,
\begin{equation}
\label{D_first_order_stat}
\| D - \langle V \rangle \|_{n_*} \leq \mu (\lambda^2/\lambda_0), \quad \mu \approx 2.9,
\end{equation}
\end{thm}

Here, we have defined, following Ref.~\onlinecite{Abanin2015a}, the symmetrized operator $\langle V\rangle$ according to
\begin{equation}
\langle V \rangle \equiv \int_0^{2\pi} \frac{d\theta}{2\pi}\, e^{iL\theta} \, V \, e^{-iL\theta}
\end{equation}
which, by construction, satisfies $[L,\langle V \rangle]=0$.

\section{Pre-thermalized Floquet time crystals}
\label{sec_floquet_time_crystal}

\subsection{Basic Picture}
\label{floquet_basic_picture}

The results of the previous section give us the tools that we need to construct
a model which is a Floquet time crystal in the pre-thermalized regime. Our
approach is reminiscent of Ref.~\onlinecite{vonKeyserlingk2016a}, where the
Floquet-MBL time crystals of Ref.~\onlinecite{Else2016b} were reinterpreted in terms of a spontaneously broken
``emergent'' $\mathbb{Z}_2$ symmetry. Here, ``emergent'' refers to the fact that
the symmetry is in some sense hidden --
its form depends on the parameters on the Hamiltonian in a manner that is not
\emph{a priori} known. Furthermore, it is not a symmetry of the
Hamiltonian, but is a symmetry of the Floquet operator.

In particular, suppose that we have a model where we can set $X = \prod_i
\sigma^x_i$. (Thus $N=2$). We then have $\Uf \approx \WUf =
\mathcal{U}^{\dagger} (X e^{-iDT}) \mathcal{U}$, where the quasi-local
Hamiltonian $D$ by construction respects the Ising symmetry generated by $X$.
This Ising symmetry corresponds to an \emph{approximate} ``emergent'' symmetry
$\mathcal{U} X \mathcal{U}^{\dagger}$ of $\Uf$ (``emergent'' for the reason stated above
and approximate because it an exact symmetry of $\WUf$, not $\Uf$, and therefore is
approximately conserved for times $t \ll t_*$.) Suppose that $D$ \emph{spontaneously breaks} the symmetry $X$ below some finite critical temperature $\tau_c$. For example, working in two dimensions or higher, we could have $D = -J \sum_{\langle i,j\rangle} \sigma_i^z \sigma_j^z$ plus additional smaller terms of strength which break integrability. We will be interested in the regime where the heating time $t_* \gg t_{\mathrm{pre-thermal}}$, where $t_{\mathrm{pre-thermal}}$ is the thermalization time of $D$.

Now consider the time evolution $\ket{\psi(t)}$, starting from a given short-range correlated state $\ket{\psi(0)}$. We also define the rotated states $\ket{\widetilde{\psi}(t)} = \mathcal{U} \ket{\psi(t)}$. At stroboscopic times $t = nT$, we find that $\ket{\widetilde{\psi}(nT)} =  (Xe^{-iDT})^n \ket{\widetilde{\psi}(0)}$.
Since $(Xe^{-iDT})^2 = e^{-2iDT}$, we see that at even multiples of the period, $t = 2nT$, the time evolution of $\ket{\widetilde{\psi}(t)}$ is described by the time-independent Hamiltonian $D$.
Thus, we expect that, after the time $t_{\mathrm{pre-thermal}}$, the system appears to be in a thermal state of $D$ at temperature $\tau$. Thus, $\ket{\widetilde{\psi}(2nT)}\bra{\widetilde{\psi}(2nT)} \approx \widetilde{\rho}$, where $\widetilde{\rho}$ is a thermal density matrix for $D$ at some temperature $\tau$, and the approximate equality means that the expectation values of local observables are approximately the same. Note that for $\tau < \tau_c$, the Ising symmetry of $D$ is spontaneously broken and $\widetilde{\rho}$ must either select a nonzero value for the order parameter $M_{2n} = \langle \sigma^z_i \rangle_{\widetilde{\rho}}$ or have long-range correlations. The latter case is impossible given our initial state, as long-range correlations cannot be generated in finite time. Then, at odd times $t = (2n+1)T$, we have 
\begin{align}
\ket{\widetilde{\psi}((2n+1)T)} \bra{\widetilde{\psi}((2n+1)T)} &\approx (X e^{-iDT}) \widetilde{\rho} (e^{iDT} X) \\ &= X \widetilde{\rho} X
\end{align}
 (since $\widetilde{\rho}$ commutes with $D$.) Therefore, at \emph{odd} times, the order parameter
\begin{equation}
M_{2n+1} = \langle \sigma^z_i \rangle_{X\widetilde{\rho}X} = -M_{2n}.
\end{equation}
Thus, the state of the system at odd times is different from the state at even times, and time translation by $T$ is spontaneously broken to time translation by $2T$.

\begin{figure*}
    \includegraphics[width=18cm]{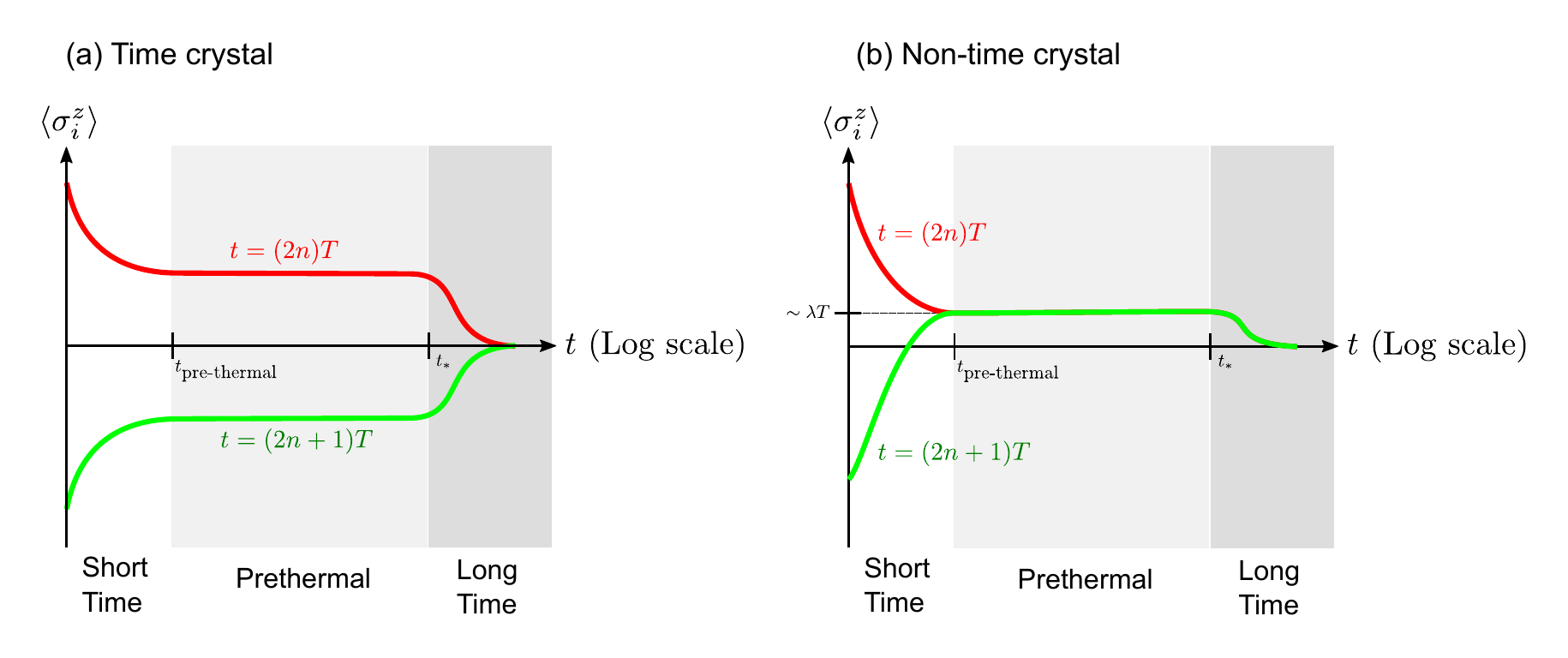}
\caption{\label{Diagrams_1} The expected time dependence of $\langle \sigma_i^z\rangle$ at stroboscopic times, starting from a state which is low-temperature with respect to $\mathcal{U} D \mathcal{U}^{\dagger}$ (for example, for a state with all spins polarized in the $z$ direction.), in (a) the pre-thermal time crystal phase, and (b) the non-time crystal pre-thermal phase.}
\end{figure*}

The above analysis took place in the frame rotated by $\mathcal{U}$. However, we
can also consider the expectation values of operators in the original frame, for
example $\bra{\psi(t)} \sigma_i^z \ket{\psi(t)} = \bra{\widetilde{\psi}(t)}
\mathcal{U}^{\dagger} \sigma_i^z \mathcal{U} \ket{\widetilde{\psi}(t)}$. The rotation
$\mathcal{U}$ is close to the identity in the regime where the heating time is
large\footnote{Specifically, it follows from the construction of $\mathcal{U}$
that $\mathcal{U} = 1 + O(\lambda T)$, and $\lambda T \ll 1$ is the regime where
the heating time is large.}, so $\sigma_i^z$ has large overlap with $\mathcal{U}^{\dagger} \sigma_i^z
\mathcal{U}$ and therefore will display fractional frequency oscillations. We
recall that the condition for fractional frequency oscillations in the
pre-thermalized regime is that (a) $D$ must spontaneously break the Ising
symmetry $X$ up to a finite critical temperature $\tau_c$; and\	 (b)  the energy
density with respect to $D$ of $\mathcal{U} \ket{\psi(0)}$ must correspond to a
temperature $\tau < \tau_c$. In Figure \ref{Diagrams_1}, we show the expected behavior at low temperatures $\tau$ and contrast it with the expected behavior in a system which is not a time crystal in the pre-thermal regime.

\subsection{Example: periodically-driven Ising spins}

Let us now consider a concrete model which realizes the behavior descrived above. We consider an Ising ferromagnet, with a longitudinal field applied to break the Ising symmetry explicitly, and driven at high frequency by a very strong transverse field. Thus, we take
\begin{equation}
H(t) = H_0(t) + V,
\end{equation}
where
\begin{align}
H_0(t) &= -\sum_i h^x(t) \sigma_i^x \\
V &= -J \sum_{\langle i,j\rangle} \sigma_i^z \sigma_j^z - h^z \sum_i \sigma_i^z,
\end{align}
and we choose the driving profile such that
\begin{equation}
\int_0^T h^x(t) dt = \frac{\pi}{2},
\end{equation}
ensuring that the ``unperturbed'' Floquet operator $U_0$ implements a $\pi$ pulse, $X = \prod_i \sigma_x^i$, and we can set $N = 2$. (If the driving does not exactly implement a $\pi$ pulse, this is not a significant problem since we can just incorporate the difference into $V$.) This implies that $h_x \sim 1/T$, and we assume that $h^z \lesssim J \ll 1/T$.

Then by the results of Section \ref{prethermalization_results} (with $J$ playing
the role of $\lambda$ here),
we find a quasi-local Hamiltonian $D = \overline{V} + \frac{1}{T}O((JT)^2)$, where
\begin{align}
\overline{V} &= \frac{1}{2T} \int_0^{2T} V_{\mathrm{int}}(t) dt.
\end{align}
In particular, in the case where the $\pi$ pulse acts instanteously, so that
\begin{equation}
h^x(t) = \frac{\pi}{2} \sum_{k=-\infty}^{\infty} \delta(t - kT),
\end{equation}
we find that
\begin{equation}
\label{ising_averaged}
\overline{V} = -J \sum_{\langle i, j\rangle} \sigma_i^z \sigma_j^z
\end{equation}
(this Hamiltonian is integrable, but in general the higher order corrections to $D$ will destroy integrability.)
More generally, if the delta function is smeared out so that the $\pi$ pulse acts over a time window $\delta$, the corrections from \eqnref{ising_averaged} will be at most of order $\sim J \delta/T$. Therefore, so long as $\delta \ll T$, then in two dimensions or higher, the Hamiltonian $D$ will indeed spontaneously break the Ising symmetry up to some finite temperature $\tau_c$, and we will observe the time-crystal behavior described above.

\subsection{Field Theory of the Pre-Thermal Floquet Time Crystal State}
\label{sec:FTC-field-theory}

The universal behavior of a pre-thermal Floquet time crystal state can be encapsulated in a field theory.
For the sake of concreteness, we derive this theory from the model analyzed in the previous section.
The Floquet operator can be written, up to nearly exponential accuracy, as:
\begin{equation}
 U_\text{f}   \approx \mathcal{U} (X e^{-iDT}) \mathcal{U}^{\dagger} 
\end{equation}
Consequently, the transition amplitude from an initial state $|{\psi_i}\rangle$ at time $t_0$
to a final state $|{\psi_f}\rangle$ at time ${t_0}+mT$
can be written in the following form, provided $t_{\mathrm{pre-thermal}} < t_0 < t_0 + mT < t_*$ :
\begin{align}
\label{eqn:transition-amplitude}
\langle{\psi_f}| \big({U_\text{f}}\big)^m |{\psi_i}\rangle
& = \langle{\psi_f}| \,\mathcal{U} (X e^{-iDT})^m \mathcal{U}^{\dagger}\,  |{\psi_i}\rangle\cr
& = \langle\tilde{\psi}_{f}| \, e^{-iDmT} \, |\tilde{\psi}_{i}\rangle
\end{align}
where $|\tilde{\psi}_{i}\rangle\equiv \mathcal{U}^{\dagger}  |{\psi_i}\rangle$ and
$|\tilde{\psi}_{f}\rangle\equiv X^m \,\mathcal{U}^{\dagger}  |{\psi_f}\rangle$;
recall that $X^m$ is $1$ or $X$ for, respectively, $m$ even or odd.

The second line of Eq. (\ref{eqn:transition-amplitude}) is just the transition amplitude for the quantum
transverse field Ising model in $(d+1)$-dimensional spacetime, with $d\geq 2$.
The model has nearest-neighbor interaction (\ref{ising_averaged}) together with higher-order terms that
are present in the full expression for $D$. Hence,
it can be represented by the standard functional integral
for the continuum limit of the Ising model:
\begin{multline}
\label{eqn:FI-real-time}
\langle\tilde{\psi}_{f}| \, e^{-iDmT} \, |\tilde{\psi}_{i}\rangle\, =\\ \int {\cal D}\varphi\,
e^{i\int {d^d}x \, dt \,  \left[\frac{1}{2}K{({\partial_t} \varphi)^2} - \frac{v^2}{2}K{(\nabla\varphi)^2} - U(\varphi)\right]}
\end{multline}
where $U(\varphi)$ has minima at $\varphi=\pm\varphi_0$ when the parameters in the Ising model place
it in the ordered phase. This functional integral is only valid for wavevectors that are
less that a wavevector cutoff: $|q| < \Lambda$,
where $\Lambda \ll 1/a$ and $a$ is the spatial lattice spacing.
Although the right-hand side of (\ref{eqn:FI-real-time})
has a continuous time variable, it is only
equal to the original peridiodically-driven problem for stroboscopic times $t=mT$ for
$m\in \mathbb{Z}$. Note the left-hand side of (\ref{eqn:FI-real-time}) is also well-defined for
arbitrary times, i.e. for continuous $m$, although it, too, only corresponds to the original
problem for integer $m$. Thus the continuous-time effective field
theory has a frequency cutoff $\Lambda_\omega$ that we are free to choose.
Although the functional integral only corresponds to the original problem for stroboscopic times, the
functional integral is well-defined for all times.
As a result of the factor of $X$ in $U_\text{f}$, the field $\varphi$ is related to the Ising spin according to
$\varphi(x,kT) \sim (-1)^k \, \sigma(x,kT)$. In other words, the field $\varphi$ in the functional integral
has the intepretation of the temporally-staggered magnetization density, just as, in the corresponding description of
an Ising anti-ferromagnet, this field would be the spatially-staggered magnetization. Discrete time-translation symmetry,
$t\rightarrow t + T$ has the following action: $\varphi \rightarrow -\varphi$. Thus, the symmetry-breaking phase,
in which $\varphi=\pm\varphi_0$, is a pre-thermal Floquet time crystal, in which TTSB occurs, as expected.

The rotated Floquet operator $\mathcal{U}^{\dagger} {U_\text{f}}\, \mathcal{U}$
has an approximate $\mathbb{Z}_2$ symmetry generated by the operator $X$ since
$\mathcal{U}^{\dagger} {U_\text{f}}\,  \mathcal{U} \approx X e^{-iDT}$ and $[D,X]=0$.
Hence, $\mathcal{U}^{\dagger} X \mathcal{U}$ commutes with the (unrotated) Floquet operator
${U_\text{f}}$. It is not a microscopic symmetry in the conventional sense, since
$\mathcal{U}^{\dagger} X\mathcal{U}$ does not commute
with the time-dependent Hamiltonian $H(t)$, except for special fine-tuned points in the Floquet time crystal phase.
However, since it commutes with the Floquet operator, it is a symmetry of the
continuum-limit field theory (\ref{eqn:FI-real-time}). (See Ref.~\onlinecite{vonKeyserlingk2016a}
for a discussion of Floquet time crystals in the MBL context that focuses on
such symmetries, sometimes called ``emergent symmetries''.)
Within the field theory (\ref{eqn:FI-real-time}),
this symmetry acts according to $\varphi\rightarrow -\varphi$, i.e. it acts in precisely the same way as
time-translation by a single period. Again, this is analogous to the case of an
Ising anti-ferromagnet, but with the time-translation taking the place of spatial translation. Thus, it is possible
to view the symmetry-breaking pattern as $\mathbb{Z}_\text{TTS} \times \mathbb{Z}_2 \rightarrow \mathbb{Z}$.
The unbroken $\mathbb{Z}$ symmetry is generated by the combination of time-translation by one period and
the action of $\mathcal{U}^{\dagger} X \mathcal{U}$.

However, there is an important difference
between a Floquet time crystal and an Ising antiferromagnet. In the latter case, it is possible to explicitly
break the the Ising symmetry without breaking translational symmetry
(e.g. with a uniform longitudinal magnetic field) and vice versa (e.g. with a spatially-oscillating exchange coupling).
In a Floquet time crystal, this is not possible because there is always a $\mathbb{Z}_2$ symmetry
$\mathcal{U}^{\dagger} X\mathcal{U}$ regardless of what small perturbation (compared to the drive
frequency) is added to the Hamiltonian. The only way to explicitly prevent the system from having a
$\mathbb{Z}_2$ symmetry is to explicitly break the time-translation symmetry. Suppose the
Floquet operator is $\mathcal{U} X e^{-iDT} \mathcal{U}^{\dagger}$. When a
weak perturbation with period $2T$ is added, the Floquet operator can be written in the approximate form
$\mathcal{U}' e^{-2i(D+Y)T} (\mathcal{U}')^{\dagger}$ where $Y$ is due to the doubled-period
weak perturbation, but it is not possible to guarantee that $[X,Y]=0$. Thus there is a symmetry generated
by an operator of the form $\mathcal{U}^{\dagger} X\mathcal{U}$
only if time-translation symmetry is present -- i.e. it is a consequence of time-translation symmetry
and pre-thermalization.

This functional integral is computed with boundary conditions on $\varphi$ at $t=t_0$ and ${t_0}+mT$.
Time-ordered correlation functions can be computed by inserting operators between the factors of $U_\text{f}$.
However, if we are interested in equal-time correlation functions (at stroboscopic times $t=kT$),
\begin{multline}
\langle\psi| \, \hat{O}(x,kT) \hat{O}(0,kT) \, |\psi\rangle\, \equiv\\
\langle\psi| \,\big({U_\text{f}}\big)^{-k} \, \hat{O}(x,0) \hat{O}(0,0) \, \big({U_\text{f}}\big)^{k}\, |\psi\rangle
\end{multline}
then we can make use of the fact that the system rapidly pre-thermalizes to
replace $\big({U_\text{f}}\big)^{k} |\psi\rangle$ by a thermal state:
\begin{multline}
\langle\psi| \big({U_\text{f}}\big)^{-k} \, \hat{O}(x,0) \hat{O}(0,0) \, \big({U_\text{f}}\big)^{k} |\psi\rangle \, =\\
\mathrm{tr}(e^{-\beta D} \hat{O}(x) \hat{O}(0))
\end{multline}
where $\beta$ is determined by $\mathrm{tr}(e^{-\beta D} D) = \langle\psi| D |\psi\rangle$.
The latter has an imaginary-time functional integral representation:
\begin{multline}
\label{eqn:FI-im-time}
\mathrm{tr}(e^{-\beta D} \hat{O}(x) \hat{O}(0)) \,=\\
\int {\cal D}\varphi\,
e^{-\!\int {d^d}x \, d\tau \,  \left[\frac{1}{2}K{({\partial_\tau} \varphi)^2} + \frac{v^2}{2}K{(\nabla\varphi)^2}
+ U(\varphi)\right]}
\end{multline}
This equation expresses equal-time correlation functions in a pre-thermal Floquet time crystal in terms of
the standard imaginary-time functional integral for the Ising model but with the understanding
that the field $\varphi$ in the functional integral is related to the Ising spins in the manner noted above.

In order to compute unequal-time correlation functions, it is convenient to use the Schwinger-Keldysh formalism
\cite{Schwinger61,Keldysh64} (see Ref. \onlinecite{Kamenev04} for a modern review).
This can be done by following the logic that led from the first line of Eq. (\ref{eqn:transition-amplitude}) to the second
and thence to Eq. (\ref{eqn:FI-real-time}). This will be presented in detail elsewhere~\cite{OpenSystems}.

We close this subsection by noting that the advantage of the field theory formulation
of a pre-thermal Floquet time crystal is the salience of the
similarity with the equilibrium Ising model; for instance, it is clear that the transition out of the Floquet time crystal
(e.g. as a function of the energy of the initial state) in the pre-thermal regime is an ordinary Ising phase transition.
The disadvantage is that it is difficult to connect it to measurable properties in a quantitative way because
the field $\varphi$ has a complicated relationship to the microscopic degrees of freedom.

\subsection{Relation to formal definitions of time crystals}
In the above discussion, we have implicitly been adopting an ``operational''
definition of time-crystal: it is a system in which, for physically reasonable
initial states, the system displays oscillations at a frequency other than the
drive frequency forever (or at least, in the pre-thermal case, for a nearly exponentially
long time.)
This is a perfectly reasonable definition of time crystal,
but it has the disadvantage of obscuring the analogies with spontaneous breaking
of other symmetries, which tends not to be defined in this way. (Although in
fact it could be; for example, an ``operational'' definition of spontaneously broken Ising
symmetry, say, would be a system in which the symmetry-breaking order parameter
does not decay with time for physically reasonable initial
states\cite{Pekker2014}.)
It was for
this reason that in Ref.~\onlinecite{Else2016b} we introduced a formal definition of time-translation
symmetry-breaking in MBL systems in terms of eigenstates (two equivalent formulations of which we called
TTSB-1 and TTSB-2.)

The definitions TTSB-1 and TTSB-2 of Ref.~\onlinecite{Else2016b} are natural generalizations of
the notion of ``eigenstate order'' used to define spontaneous breaking of other
symmetries in MBL \cite{Huse2013,Pekker2014}. On the other hand they, like the notion of eigenstate order
in general, are not really appropriate outside of the MBL context. In this
subsection, we will review the usual formal definitions of spontaneous symmetry breaking
in equilibrium. Then we will show how they can be extended in a natural way to
time-translation symmetries, and that these extended versions are satisfied by
the pre-thermal Floquet time crystals constructed above.

Let us first forget about time-translation symmetry, and consider a
time-independent Hamiltonian $H$ with
an Ising symmetry generated by $X$. Let $\rho$ be a steady state of the
Hamiltonian; that is, it is invariant under the time evolution generated by $H$.
(Here, we work in the thermodynamic limit, so by $\rho$ we really mean a
function which maps local observables to their expectation values; that is, we
define a state in the $C^{*}$-algebra sense \cite{Haag1996}.) Generically, we expect $\rho$ to
be essentially a thermal state.
If the symmetry is spontaneously broken, then $\rho$ can obey the cluster
decomposition (i.e.\ its correlations can be short-ranged), or it can be
invariant under the symmetry $X$, but \emph{not} both. That is, any state
invariant under the symmetry decomposes as $\rho = \frac{1}{2}(\rho_{\uparrow} +
\rho_{\downarrow})$,
where $\rho_{\uparrow}$ and $\rho_{\downarrow}$ have opposite values of the Ising order
parameter, and are mapped into each other under $X$. Thus, a formal definition
of spontaneously broken Ising symmetry can be given as follows. We call a
symmetry-invariant steady state $\rho$ state an \emph{extremal
symmetry-respecting state} if there do not exist states
$\rho_1$ and $\rho_2$ such that $\rho = p \rho_1 + (1-p)\rho_1$ for some $p \in
(0,1)$, where $\rho_1$ and $\rho_2$ are symmetry-invariant steady states. We say
the Ising symmetry is spontaneously broken if extremal symmetry-invariant steady
states do not satisfy the cluster decomposition. Similar statements can be made
for Floquet systems, where by ``steady state'' we fnow mean a state that returns
to itself after one time cycle.

We can now state the natural generalization to time-translation symmetry. For
time-translation symmetry, ``symmetry-invariant'' and ``steady state'' actually
mean the same thing. So we say that time-translation symmetry is spontaneously
broken if extremal steady states do not satisfy the cluster decomposition. This
is similar to our definition TTSB-2 from
Ref.~\onlinecite{Else2016b} (but not exactly the same, since TTSB-2 was expressed in terms of
eigenstates, rather than extremal steady states in an infinite system), so we call it TTSB-2$'$.
We note that TTSB-2$'$ implies that any
short-range correlated state $\rho$, i.e.\ a state $\rho$ which satisfies the cluster
decomposition, must not be an extremal steady state. Non-extremal states never
satisfy the cluster decomposition, so we conclude that short-range correlated
states must not be steady states at all, so they cannot simply return to
themselves after one time cycle.
(This is similar to, but again not identical with, TTSB-1 in
Ref.~\onlinecite{Else2016b}.)

We note that, for clean systems, the only steady state of the Floquet
operator $\Uf$ is believed to be the infinite temperature state\cite{DAlessio2014,Lazarides2014,Ponte2015a}
which always obeys the cluster property, and hence time translation symmetry is not broken spontaneously. This does not contradict our previous results, since we already saw that time translation symmetry is only spontaneously broken in the pre-thermal regime, not at infinitely long times. Instead, we should examine the steady states of the \emph{approximate} Floquet operator $\WUf$ which describes the dynamics in the pre-thermal regime. We recall that, after a unitary change of basis, $\WUf = X e^{-iDT}$, where $D$ commutes with $X$ and spontaneously breaks the Ising symmetry generated by $X$ (for temperatures $\tau < \tau_c$). Hence $\WUf^2 = e^{-2iDT}$. Any steady state $\rho$ of $\WUf$ must be a steady state of $\WUf^2$, which implies (if its energy density corresponds to a temperature $\tau < \tau_c$) that it must be of the form $\rho = t \rho_{SB} + (1-t)X \rho_{SB} X$, where $\rho_{SB}$ is an Ising symmetry-breaking state of temperature $\tau$ for the Hamiltonian $D$. Hence, we see (since $\rho_{SB}$ is invariant under $e^{-iDT}$) that $\WUf \rho \WUf^{\dagger} = t X \rho_{SB} X + (1-t) \rho_{SB}$. So if $\rho$ is a steady state of $\WUf$ and not just $\WUf^2$, we must have $t = 1/2$. But then the state $\rho$ clearly violates the cluster property. Hence, time translation is spontaneously broken.

\section{Spontaneously-broken continuous time-translation symmetry in the pre-thermal regime}
\label{sec_continuous_time_crystal}

\subsection{Basic Picture}
\label{continous_basic_picture}

The pre-thermalized Floquet time crystals discussed above have a natural analog in undriven systems with \emph{continuous} time translation symmetry. Suppose we have a time-independent Hamiltonian
\begin{equation}
\label{eqn:cTTSB-Hamiltonian}
H = -u L + V,
\end{equation}
where the eigenvalues of $L$ are integers; in other words, for time $T =
2\pi/u$, the condition $e^{in u L T} = 1$ holds for all $n\in\mathbb{Z}$. We
also assume that
$L$ is a sum of local terms of local strength $O(1)$;
and $V$ is a local Hamiltonian of local strength $\lambda \ll u$. Then by Theorem 3.1 of
Ref.~\onlinecite{Abanin2015a}, restated in Theorem 2 in Section \ref{prethermalization_results}),
there exists a local unitary $\mathcal{U}$ such that
$\mathcal{U} H \,\mathcal{U}^{\dagger} = -uL + D + \hat{V}$ such that $[D,L]=0$ and the local strength of
$\hat{V}$ is $\sim \lambda \, e^{-O([\log \lambda T]^3/[\lambda T])}$.
As noted in Theorem 2 in Section \ref{prethermalization_results}),
the first term in the explicit iterative construction of $D$ in Ref.~\onlinecite{Abanin2015a} is
$D = \langle{V}\rangle + \frac{1}{T}O(\lambda T)^2$, where
\begin{equation}
  \langle{V}\rangle \equiv \frac{1}{2\pi}\int_0^{2\pi} d\theta \, e^{iL\theta} V e^{-iL\theta}.
\end{equation}
As a result of this theorem, such a system
has an approximate U(1) symmetry generated by $\mathcal{U}^{\dagger} L \,\mathcal{U}$
that is explicitly broken only by nearly exponentially-small terms. 
Consequently, $\mathcal{U}^{\dagger} L \,\mathcal{U}$ is conserved by the dynamics of
$H$ for times $t \ll t_* = e^{O([-\log \lambda T]^3/[\lambda T])}$. We will call the Hamiltonian
$-u L+D$ the ``pre-thermal'' Hamiltonian, since it governs the dynamics of the system for times short compared
to $t_*$. We will assume that we
have added a constant to the Hamiltonian such that $L$ is positive-definite;
this will allow us to abuse terminology a little by referring to the expectation value of $L$ as the ``particle
number'', in order to make analogies with well-known properties of
Bose gases, in which the generator of the $\U(1)$ symmetry is the particle number
operator. In this vein, we will call $u$ the electric potential, in analogy with (negatively) charged superfluids.

We will further suppose that $D$ is neither integrable
nor many-body localized, so that
the dynamics of $D$ will cause an arbitrary initial state $|{\psi_0}\rangle$
with non-zero energy density and non-zero $\langle{\psi_0}| L|{\psi_0}\rangle$ to rapidly thermalize
on some short (compared to $t_*$) time scale $t_{\mathrm{pre-thermal}}\sim\lambda^{-1}$. The
resulting thermalized state can be characterized by the expectation values of $D$ and $L$,
both of which will be the same as in the initial state, since energy and particle number are
conserved. Equivalently, the thermalized state can be characterized by its
temperature $\beta$ (defined with respect to $D$) and effective chemical potential $\mu$.
In other words, all local correlation functions of local operators can be computed with
respect to the density matrix $\rho = e^{-\beta(D-\mu L)}$.
The chemical potential $\mu$ has been introduced to enforce the condition
$\text{tr}(\rho L) = \langle{\psi_0}| L|{\psi_0}\rangle$.

Now suppose that we choose $V$ such that $D$ spontaneously breaks the $U(1)$
symmetry in some range of temperature $1/\beta$ and chemical potential $\mu$.
Suppose, further, that we prepare the system in a short-range correlated initial state
$\ket{\psi_0}$ such that the energy density (and hence, its temperature) is sufficiently low, and the number
density sufficiently high, so that the corresponding thermalized state spontaneously
breaks the $\mathrm{U}(1)$ symmetry generated by $L$. Then, the preceding statement must be slightly
revised: all local correlation functions of local operators can be computed with
respect to the density matrix $\rho = e^{-\beta(D-\mu L - \epsilon X)}$ for some $X$
satisfying $[X,L]\neq 0$. The limit $\epsilon\rightarrow 0$ is taken after the thermodynamic limit is taken;
the direction of the infinitesimal symmetry-breaking field $X$ is determined by the initial state.
To avoid clutter, we will not explicitly write the $\epsilon X$in the next paragraph, but it is understood.

Consider an operator $\Phi$ that satisfies $[L,\Phi]=\Phi$.
(For example, if we interpret $L$ as the particle number, we can take $\Phi$ to be the particle creation operator.)
Its expectation value at time $t$ is given by
\begin{multline}
\label{eqn:c-TTSB-exp-val}
\langle{\psi_0}|e^{-i (-uL + D) t} \Phi e^{i (-uL + D) t}|{\psi_0}\rangle \\
= \text{tr}\!\left(\left[e^{-i (-uL+D) t} \Phi e^{i (-uL+D) t}\right]
e^{-{\beta}(D-\mu L)}\right)\\
=  e^{i (\mu-u)t}\,
\text{tr}\!\left(\left[e^{-i (-\mu L+D) t}\, \Phi \, e^{i (-\mu L+D) t}\right]
e^{-{\beta}(D-\mu L)}\right)\\
\end{multline}
According to the discussion in Appendix \ref{phase_winding}, which makes use of the
result of Watanabe and Oshikawa \cite{Watanabe15}, the trace on
the right-hand-side of the second equality must be independent of time. Hence, so long as
$\mathrm{Tr}(\Phi e^{-\beta(D - \mu L)}) \neq 0$ (which we assume to be true for some order parameter $\Phi$ in the symmetry-breaking phase), we find that the expectation value of $\Phi$ oscillates with frequency given by the ``effective
electrochemical potential'' $\mu-u$ due to the winding of the phase of $\Phi$.

If the dynamics were exactly governed by $-uL + D$, then the system would
oscillate with period $2\pi/(u-\mu)$ forever. As it is, these oscillations will be
observed until the exponentially late time $t_*$. At infinitely long times, the
system approaches a thermal state of the full Hamiltonian $-uL + D + \hat{V}$.
Since $\hat{V}$ is small, this is approximately the same as a thermal state of $-uL + D$.
However, because $\hat{V}$ is not exactly zero, the particle number is not
conserved and in equilibrium the system chooses the particle number that minimizes its
free energy, which corresponds to the ``electrochemical potential'' being zero,
$\mu - u = 0$. Since this corresponds to zero frequency of oscillations, it
follows that no oscillations are observed at infinite time.

The above discussion is essentially the logic that was discussed
in Refs.~\onlinecite{Wilczek13,Volovik2013,Watanabe15},
where it was pointed out that a superfluid
at non-zero chemical potential is a time crystal as a result of the well-known
time-dependence of the order parameter \cite{Pethick2008}.
However, there is an important difference: the U(1) symmetry is not a symmetry of
the Hamiltonian of the problem and, therefore, does not require
fine-tuning but, instead, emerges in the $u \rightarrow \infty$ limit, thereby
evading the criticism \cite{Nicolis2012,Castillo2014,Thies2014,Volovik2013,Watanabe15}
that the phase winds in the ground state only if the U(1) symmetry is exact.

\subsection{Example: XY Ferromagnet in a Large Perpendicular Field}

Consider the concrete example of a spin-1/2 system in three spatial dimensions,
with Hamiltonian
\begin{multline}
\label{eqn:XXZ}
  H = -h^z \sum_i S_i^z - h^x \sum_i S_i^x 
  \\- \sum_{i,j}\left[ J^x_{ij} S_i^{x} S_j^{x} + J^y_{ij} S_i^{y} S_j^{y} + J^z_{ij} S_i^{z} S_j^{z}\right],
\end{multline}
We take $L = S^z \equiv \sum_i S_i^z$, and the longitudinal magnetic field $h^z$ plays the
role of $u$ in the preceding section. We take
$J_{ij}$ and $J^z_{ij}$ to vanish except for nearest neighbors, for which $J^x_{ij}=J + \delta J$,
$J^y_{ij}=J_y +\delta J$, and $J^z_{ij}=J^z$. (We do not assume $\delta J \ll J$.)
The local scale of $V$ is given by $\lambda = \text{max}(J+\delta J,{h^x})$,
so that the condition $\lambda \ll T^{-1} \sim h^z$ is satisfied if $J+ \delta J, h^x  \ll h^z$.
In this case, $D$ is (to first order) the Hamiltonian of an $XY$ ferromagnet:
\begin{equation}
D = - \sum_{\langle i,j\rangle } \left[J (S_i^{x} S_j^{x} + S_i^{y} S_j^{y}) + J^z
S_i^{z} S_j^{z}\right] + \frac{1}{T} O(\lambda/h^z)^2.
\end{equation}

Then, starting from a short-range correlated state with appropriate values of energy and $\langle S^z\rangle$, we expect that time evolution governed by $D$ causes the system to ``pre-thermalize'' into a symmetry-breaking state with some value of the order parameter $\langle S_i^{+} \rangle = n_0 e^{i\phi}$. According to the preceding discussion, the order parameter will then rotate in time with angular frequency $\omega = \mu - h^z$ (where $\mu \lesssim \lambda$ is determined by the initial value of $\langle S^z \rangle$) for times short
compared to the thermalization time $t_*$.

Note, however, that we have assumed that the system is completely isolated. If the system is not isolated, then
the periodic rotation of the order parameter will cause the system to emit radiation, and this radiation will
cause the system to decay to its true ground state \cite{Bruno13a,Bruno13c}. 

\subsection{Field Theory of Pre-Thermal Continuous-TTSB Time Crystal}
\label{sec:cTTSB-field-theory}

For simplicity we will give only the imaginary-time field theory for equal-time correlation functions
deep within the pre-thermal regime; the Schwinger-Kelysh functional integral for
unequal-time correlation functions, with nearly exponentially-small thermalization effects taken into account,
will be discussed elsewhere \cite{OpenSystems}. Introducing the field
$\phi \sim ({S_x} + i {S_y}) e^{i (\mu-u)t}$, we apply Eq. (\ref{eqn:c-TTSB-exp-val})
to the XY ferromagnet of the previous section, thereby obtaining the effective action:
\begin{equation}
\label{eqn:cTTSB-im-time}
S_\text{eff}=
\int {d^d}x \, d\tau \,  \left[{\phi^*}\partial_\tau \phi - \mu{\phi^*}\phi + g({\phi^*}\phi)^2 + \ldots\right]
\end{equation}
The $\ldots$ represents higher-order terms. The U(1) symmetry generated by $S^z$
acts according $\phi\rightarrow e^{i\theta}\phi$. Time-translation symmetry acts according to
$\phi(t) \rightarrow  e^{i (\mu-u)a}\, \phi(t+a) $ for any $a$. Thus, when $\phi$ develops
an expectation value, both symmetries are broken and a combination of them is preserved
according to the symmetry-breaking pattern $\mathbb{R}_\text{TTS} \times U(1) \rightarrow \mathbb{R}$,
where the unbroken $\mathbb{R}$ is generated by a gauge transformation by $\theta$ and a time-translation
$t\rightarrow t +\frac{\theta}{\mu - u}$.

From the mathematical equivalence of Eq. (\ref{eqn:cTTSB-im-time}) to the effective field theory
of a neutral superfluid, we see that (1) in 2D, there is a quasi-long-range-ordered phase -- an
`algebraic time crystal' -- for initial state energies below a Kosterlitz-Thouless transition; 
(2) the TTSB phase transition in 3D is in the ordinary XY universality class in 3D; (3) the 3D time crystal
phase has Goldstone boson excitations. If we write
$\phi(x,t) = \sqrt{\left(\frac{\mu}{2g}+\delta\rho(x,t)\right)}\, e^{i\theta(x,t)}$,
and integrate out the gapped field $\delta\rho(x,t)$, then the effective action for the gapless Goldstone boson
$\theta(x,t)$ is of the form discussed in Ref.~\onlinecite{Castillo2014}.

\section{Pre-thermalized Floquet topological phases}
\label{sec_sxt}
We can also apply our general results of Section \ref{prethermalization_results} to Floquet
symmetry-protected (SPT) and symmetry-enriched (SET) topological phases, even
those which don't exist in stationary systems. (We will henceforth use the abbreviation SxT to refer to either SPT or SET phases.)

As was argued in Refs.~\onlinecite{Else16a,Potter16}, any such phase protected by symmetry $G$ is analogous to a topological phase of a \emph{stationary} system protected by symmetry $\mathbb{Z} \rtimes G$, where the extra $\mathbb{Z}$ corresponds to the time translation symmetry. Here the product is semi-direct for anti-unitary symmetries and direct for unitary symmetries. For simplicity, here we will consider only unitary symmetries. Similar arguments can be made for anti-unitary symmetries.

We will consider the class of phases which can still be realized when the $\mathbb{Z}$ is refined to $\mathbb{Z}_N$. That is, the analogous stationary phase can be protected by a unitary representation $W(\widetilde{g})$ of the group $\widetilde{G} = \mathbb{Z}_N \times G$.
 Then, in applying the general result of Section \ref{prethermalization_results}, we will choose $H_0(t)$ such that its time evolution over one time cycle is equal to $X \equiv W(\mathbb{T})$, where $\mathbb{T}$ is the generator of $\mathbb{Z}_N$.
Then it follows that, for a generic perturbation $V$ of small enough local strength $\lambda$, there exists a local unitary rotation $\mathcal{U}$ (commuting with all the symmetries of $\Uf$) such that $\Uf \approx \WUf$, where $\WUf = \mathcal{U} X e^{-iDT} \mathcal{U}^{\dagger}$, $D$ is a quasi-local Hamiltonian which commutes with $X$, and $\WUf$ well describes the dynamics until the almost exponentially large heating time $t_*$.

Now let us additionally assume (since we want to construct a Floquet-SxT protected by the symmetry $G$, plus time-translation) that the Floquet operator $\Uf$ is chosen such that it has the symmetry $G$. Specifically, this means that it is generated by a periodic time evolution $H(t)$ such that, for all $g \in G$, $W(g) H(t) W(g)^{-1}$,
By inspection of the explicit construction for $\mathcal{U}$ and $D$ (see Appendix \ref{sec:proofs}), it is easy to see that in this case $\mathcal{U}$ is a symmetry-respecting local unitary with respect to $W(g)$, and $D$ commutes with $W(g)$. That is, the rotation by $\mathcal{U}$ \emph{preserves} the existing symmetry $G$ as well as revealing a new $\mathbb{Z}_N$ symmetry generated by $X$ (which in the original frame was ``hidden'').

Therefore, we can choose $D$ to be a Hamiltonian whose ground state is in the \emph{stationary} SxT phase protected by $\mathbb{Z}_N \times G$. It follows (by the same arguments discussed in Ref.~\onlinecite{Else16a} for the MBL case) that the ground state $D$ will display the desired Floquet-SxT order under the time evolution generated by $\mathcal{U}^{\dagger} \Uf \mathcal{U} = Xe^{-iDT}$. Furthermore, since Floquet-SxT order is invariant under symmetry-respecting local unitaries, the ground state of $\mathcal{U} D \mathcal{U}^{\dagger}$ will display the desired Floquet-SxT order under $\Uf$.

We note, however, that topological order, in contrast to symmetry-breaking order, does not exist at nonzero temperature (in clean systems, for spatial dimensions $d < 4$). Thus, for initial state mean energies $\langle D\rangle$ that corresponds to temperatures $\beta^{-1}$ satisfying $0<\beta^{-1}\ll\Delta$, where $\Delta$ is the bulk energy gap, the system will exhibit exponentiall-small corrections $\sim e^{-\beta\Delta}$ to the quantized values
that would be observed in the ground state. This is no worse than the situation in thermal equilbirum
where, for instance, the Hall conductance is not precisely quantized in experiments, but has small
corrections $\sim e^{-\beta\Delta}$. However, preparing such an initial state will be more involved than
for a simple symmetry-breaking phase. For this reason it is more satisfactory to envision cooling the system by coupling to a thermal bath, as discussed in Section \ref{sec_open}, which is analogous to how topological phases
are observed in thermal equilibrium experiments -- by refrigeration.

\section{Open systems}
\label{sec_open}

So far, we have considered only isolated systems. In practice, of course, some coupling to the environment will always be present. One can also consider the effect of classical noise, for example some time-dependent randomness in the parameters of the drive, so that successive time steps do not implement exactly the same time evolution.
The Floquet-MBL time crystals of Ref.~\onlinecite{Else2016b} are not expected to remain
robust in such setups, since MBL will be destroyed. Since some amount of coupling to the environment is inevitable in realistic setups,
this limits the timescales over which one could expect to observe Floquet-MBL time crystals experimentally.

\begin{figure}
    \includegraphics[width=6cm]{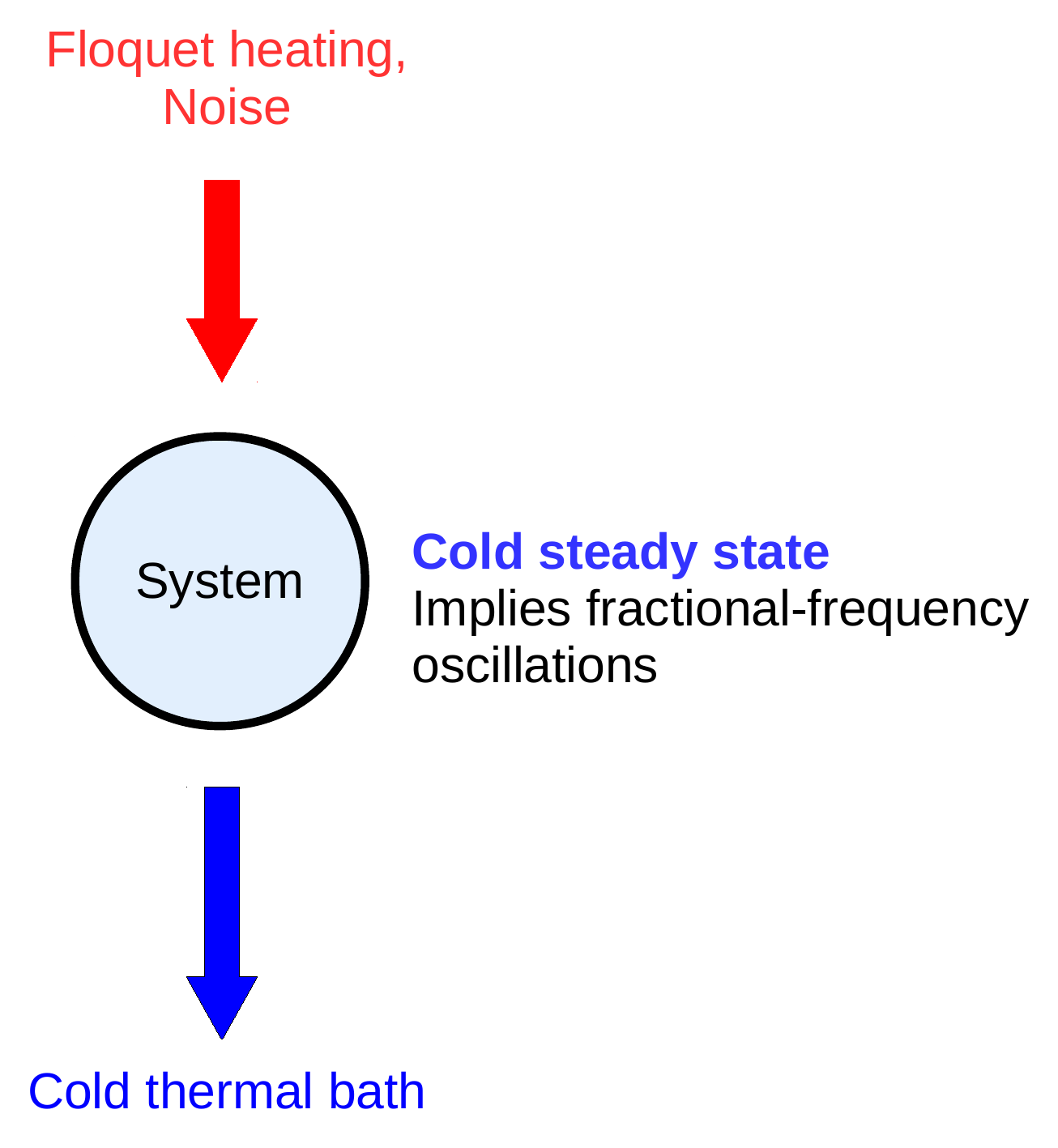}
    \caption{\label{balancing} So long as the energy inflow due to noise and
        periodic driving is balanced by the outflow to a cold thermal bath,
        giving a low-energy steady state, oscillations at a fraction of the
    drive frequency will be observed.}
\end{figure}

However, the situation could be quite different for the
pre-thermal time crystals of this work. A complete treatment is
beyond the scope of the present work, so in this section we will confine
ourselves to stating one very interesting hypothesis:
Floquet case time-crystals can actually be \emph{stabilized} in open systems
so that the oscillations actually continue \emph{forever} for \emph{any} initial state (in contrast to the case of isolated systems, in which, as discussed previously, the oscillations continue only up to some very long time, and only for some initial states). We will not attempt to establish this more rigorously, but simply discuss a plausible scenario by which this would occur. The idea, as depicted in Figure \ref{balancing}, is that the heating due to the periodic driving, as well as classical noise sources and other stray couplings to an environment, can be counteracted by cooling from a coupling to a sufficiently cold thermal bath. Provided that the resulting steady-state has sufficiently low ``energy'', we will argue that that oscillations at a fraction of the drive frequency will be observed in this steady state. Here ``energy'' means the expectation value of the effective Hamiltonian $D$ which describes the dynamics in the prethermal regime. We discuss this hypothesis further, and show that it indeed implies periodic oscillations, in Appendices \ref{appendix_open} and \ref{ssb_nonthermal}.
We also note that this argument does not apply to the continuous-time time crystals of Section \ref{sec_continuous_time_crystal}, since in that case low energy is not a sufficient condition to observe oscillations even in an isolated system; there is also a dependence on the chemical potential $\mu$.

\section{Discussion}
\label{sec_discussion}

In this paper, we have described how phases protected by time-translation
symmetry can be observed
in the pre-thermal regime of driven and undriven quantum systems. This greatly
increases the set of experimental systems in which such phases can be
observed, since, as opposed to previous proposals, we do not require many-body
localization to robustly prevent the system from heating to infinite temperature.
While many-body localization has been observed in experiments~\cite{Schreiber15,Smith2016,Choi2016},
the ideas put forward in this paper significantly reduce experimental requirements
as strong disorder is not required.

Our Theorem 1 implies that the time-translation-protected behavior (for example,
the fractional-frequency oscillations in the Floquet time crystal) can be observed to nearly exponentially-late times,
provided that the drive frequency is sufficently high. However, the rigorous
bound given in the theorem -- which requires a drive frequency $\sim 10^3$ times
larger than the local couplings in the time-dependent Hamiltonian -- may not be tight.
Therefore, it would be interesting to check numerically whether (in the Floquet
time crystal case, say) long-lived oscillations are observed
in systems with drive frequency only moderately larger than the local couplings. This may be challenging in small
systems, in which there isn't a large separation of energy scales between the local coupling strength and
the width of the many-body spectrum (which the frequency should certainly not exceed). In one-dimensional
systems, oscillations will not be observed to exponentially-long (in the drive frequency) times, but will have
a finite correlation time for any non-zero energy density initial state. However, there will be a universal quantum
critical regime in which the correlation time will be the inverse effective temperature.

Although naive application of Theorem 1 suggests that the ideal
situation is the one in which the drive frequency becomes infinitely large, in
practice very high-frequency driving will tend to excite high energy modes that
were ignored in constructing the model lattice Hamiltonian. For example, if the
model Hamiltonian describes electrons moving in a periodic potential in
the tight-binding approximation, high frequency driving would excite higher orbitals that were excluded. Thus,
the driving frequency $\Omega$ needs to be much greater than the local energy scales
of the degrees of freedom included in the model Hamiltonian (except for one
particular coupling, as discussed in Section \ref{sec_floquet_time_crystal}), but also much less than
the local energy scales of the degrees of freedom not included.
(One cannot simply include \emph{all} degrees of freedom in the model
Hamiltonian, because then the norm of local terms would be unbounded, and
Theorem 1 would not apply.)

In the case of undriven systems, we have shown that continuous time-translation symmetry
breaking can similarly occur on nearly exponentially-long time intervals even without any
fine-tuning of the Hamiltonian, provided that there is a large separation of scales in the
Hamiltonian. We show how in certain cases this can be described in terms of approximate
Goldstone bosons associated with the spontaneously-broken time-translation symmetry.

Our analysis relied on the construction of hidden approximate symmetries that
are present in a pre-thermal regime.
The analogous symmetries in MBL systems, where they
are exact, were elucidated in the interesting work of von Keyserlingk et al.~\cite{vonKeyserlingk2016a}.
In the time-translation protected phases discussed here,
the symmetry generated by the operator $\mathcal{U}^{\dagger} X \mathcal{U}$
is enslaved to time-translation symmetry since, in the absence of fine-tuning, such
a symmetry exists exists only if time-translation symmetry is present. (That is,
if we add fields to the Hamiltonian that are periodic with period $nT$ and not
period $T$, then the hidden symmetry no longer exists.)
Moreover, this symmetry is broken if and only if time-translation symmetry is broken.
(Similar statements hold in the MBL case\cite{vonKeyserlingk2016a}.)
In the Floquet time crystal case, the hidden symmetry generated by $\mathcal{U}^{\dagger} X \mathcal{U}$
acts on the order parameter at stroboscopic times in the same way as time-translation by $T$
(a single period of the drive), and therefore it does not constrain correlation functions any more than they already are
by time-translation symmetry. The same observation holds for the approximate symmetry
generated by $L_z$ in the undriven case.

However, there are systems in which time crystal behavior actually does ``piggyback'' off another
broken symmetry. This does require fine-tuning, since it is necessary to ensure that the system
posseses the ``primary'' symmetry, but such tuning may be physically natural
(e.g. helium atoms have a very long lifetime, leading to a U(1) symmetry).
The broken symmetry allows a many-body system to effectively become a few-body
system. Thus, time crystal behavior can occur in such systems for the same reason that oscillations
can persist in few-body systems. Oscillating Bose condensates (e.g. the AC Josephson effect and the
model of Ref.~\onlinecite{Sacha15}) can, thus, be viewed as fine-tuned time crystals.
They are not stable to arbitrary time-translation symmetry-respecting perturbations; a perturbation
that breaks the ``primary'' symmetry will cause the oscillations to decay.
Indeed, most few-body systems are actually many-body systems in which a spontaneously-broken
symmetry approximately decouples a few degrees of freedom. A pendulum is a system of $~10^{23}$ atoms
that can be treated as a single rigid body due to spontaneously-broken spatial translational symmetry:
its oscillations owe their persistence to this broken symmetry, which decouples the center-of-mass
position from the other degrees of freedom.

With the need for MBL obviated by pre-thermalization, we have opened up the
possibility of time-translation protected phases in open systems, in which MBL is impossible~\cite{Levi2015,Fischer2015,Nandkishore2014,Gopalakrishnan2014,Johri2015,Nandkishore2015b,Li2015,Nandkishore2016,Hyatt2016}. In fact, since the results of Appendix \ref{ssb_nonthermal} show that TTSB can occur
in non-thermal states, it is possible for the coupling to a cold bath to counteract the heating effect that would otherwise
bring an end to the pre-thermal state at time $t_*$. This raises the possibility
of time-translation protected phases that survive to infinite times in
non-equilibrium steady states; the construction of such states is an interesting avenue for future work.

\begin{acknowledgments}
We thank D.E.~Liu and T.~Karzig for helpful discussions. We thank W.W.~Ho for
    pointing out a typo in a previous version of the manuscript.
D.E. acknowledges support from the Microsoft Corporation.
Part of this work was completed at the Aspen Center for Physics,
which is supported by National Science Foundation grant PHY-1066293.
\end{acknowledgments}
\bigskip

{\it Note added:} After the submission of this paper, two experimental papers (J. Zhang {\it et al.}, arXiv:1609.08684 and S. Choi {\it et al.}, arXiv:1610.08057) have appeared with evidence
consistent with the observation of a Floquet time crystal. We note that the J. Zhang {\it et al.} paper
implements disorder by addressing each ion sequentially. A pre-thermal version of this experiment would
not need disorder, thereby sidestepping this bottleneck standing in the way of experiments on larger systems.
The Choi {\it et al.} paper occurs in a system that is unlikely to be many-body localize, and therefore
occurs during a slow approach to equilibrium. This is unlikely to correspond to a prethermal regime, but the approximate short-time form of the time evolution entailed in our Theorem 1 might still be relevant to understanding the results.

\appendix
\section{Rigorous proof of pre-thermalization results}
\label{sec:proofs}
\subsubsection{Definition of the norm}
\label{sec:norm}

Let's suppose, for the sake of concreteness, that we have a
spin system with a local time-dependent Hamiltonian of the form:
\begin{align}
H(t) &= \sum_{i,j} J_{i,j}^{\alpha\beta}(t) {S^\alpha_i}  {S^\beta_j} +
\sum_{i,j,k } K_{i,j,k}^{\alpha\beta\gamma}(t) {S^\alpha_i}{S^\beta_j}{S^\gamma_k}
+ \ldots\cr
& = \sum_p \sum_{p-\text{tuples}} A_{{i_1},\ldots,{i_p}}
\end{align}
Here $\alpha = x,y,z$ are the components of the spins, and $i,j,k$ are lattice sites.
In the first line, we have explicitly written the $2$-site and $3$-site terms; the $\ldots$ represents
terms up to $n$-site terms, for some finite $n$. It is assumed that these interactions
have finite range $r\geq n$ such that all of the sites in a $k$-site term are within distance $r$.
In the second line, we have re-expressed the Hamiltonian in a more generic form
in terms of $p$-site terms $A_{{i_1},\ldots,{i_p}}$ with ${i_1}\neq \ldots \neq {i_p}$.
To avoid clutter, we have not explicitly denoted the $t$-dependence of $A_{{i_1},\ldots,{i_p}}$.
We define the local instantaneous norm $\|A_{{i_1},\ldots,{i_p}}\|_n$ according to
\begin{equation}
\|A_{{i_1},\ldots,{i_p}}\|^\text{inst}_n \equiv e^{p \kappa_n} \|A_{{i_1},\ldots,{i_p}}\| 
\end{equation}
where $\|A_{{i_1},\ldots,{i_p}}\|$ is the operator norm of $A_{{i_1},\ldots,{i_p}}$
at a given instant of time $t$ and 
\begin{equation}
\label{kappan}
{\kappa_n} \equiv {\kappa_1}/[1+\ln n].
\end{equation}
 We make this
choice of $n$-dependence of $\kappa_n$,
following Ref.~\onlinecite{Abanin2015a} for reasons that will be clear later.
We then average the instantaneous norm over one cycle of the drive:
\begin{equation}
\|A_{{i_1},\ldots,{i_p}}\|_n \equiv \frac{1}{T} \int_0^T dt\,  \|A_{{i_1},\ldots,{i_p}}\|^\text{inst}_n 
\end{equation}
It is only in this step that we differ from Abanin et al. \cite{Abanin2015a},
who consider the supremum over $t$ rather than the average.
In analyzing the Floquet operator, i.e. the evolution due to $H$ at stroboscopic times,
it is the total effect of $H$, which is determined by its integral over a cycle, that concerns us.
Error terms that act over a very short time, even if they are relatively strong, have little effect on
the Floquet operator so long as their norm, as defined above, is small.
Finally, we define the global time-averaged norm of the Hamiltonian $H$:
\begin{equation}
\|H\|_n \equiv \sup_{j} \sum_p \sum_{p-\text{tuples}} \Bigl[\sum_k \delta_{j,i_k}\Bigr] \|A_{{i_1},\ldots,{i_p}}\|_n \,
\end{equation}
The term in square braces restricts the sum to $p$-tuples that contain the site $j$.

\subsubsection{More technical statement of Theorem \ref{mainthm}}
Theorem \ref{mainthm} stated above will follow from the following slightly more technical formulation. For notational simplicity we work in units with $T=1$.
\begin{thmbis}{mainthm}
 Consider a periodically-driven system with Floquet operator:
\begin{equation}
    U_\text{f} = \mathcal{T}\!\exp\left(-i\int_0^T [H_0(t) + V(t)] dt\right),
\end{equation}
where $X \equiv \mathcal{T} \exp\left(-i\int_0^T H_0(t)\right)$ satisfies $X^N =
1$ for some integer $N$, and we assume that $H_0$ can be written as a sum
$H_0(t) = \sum_i h_i(t)$ of terms acting on single sites $i$. Define $\lambda \equiv \| V \|_1$.
Then there exists a sequence of quasi-local $A_n$ such that, defining $\mathcal{U}_n = e^{-iA_n} \cdots e^{-iA_1}$, we have
\begin{equation}
\mathcal{U}_n \Uf\, \mathcal{U}_n^{\dagger} =
X \,\mathcal{T}\!\exp\left(-i\int_{0}^{1} [D_n + E_n + V_n(t)] dt\right),
\end{equation}
where $[D_n,X]=0$; $D_n, E_n$ are independent of time; and
\begin{align}
\label{bounds1}
\| V_n \|_n, \| E_n \|_n &\leq 2 K_n \lambda^n, \\
\| A_n \|_n &\leq (N+1) K_n \lambda^n, \label{bounds2} \\
\| D_n - D_{n-1} \|_n &\leq K_n \lambda^n, \label{bounds3}
\end{align}
where we have defined $\lambda \equiv \|V\|_1$, and
\begin{multline}
K_n = C^{n-1} \prod_{k=1}^{n-1} m(k), \quad C = 2(N+3)(N+4),
\\ m(n) = \frac{18}{\kappa_{n+1}(\kappa_n - \kappa_{n+1})}.
\end{multline}
These bounds hold provided that $n \leq n_*$, with
\begin{equation}
n_* = \frac{\lambda_0/\lambda}{[1+\log(\lambda_0/\lambda)]^3}, \quad \lambda_0 = (36C)^{-1}
\end{equation}
and provided that
\begin{equation}
\lambda < \frac{\mu}{N+3}, \quad \mu \approx 0.07.
\end{equation}
\end{thmbis}
Theorem \ref{mainthm} follows from Theorem \ref*{mainthm}$'$, because $n_*$ is chosen such that $n \leq n_*$ implies $C m(n) \leq \frac{1}{2\lambda}$. It then follows that $K_{n+1} \lambda^{n+1}/(K_n \lambda^n) = C m(n) \lambda \leq \frac{1}{2}$, and hence that $K_n \lambda^n \leq \lambda/2^{n-1}$. Moreover, we obtain \eqnref{D_first_order} by summing \eqnref{bounds3}, from which we see that $\| D_n - D_1 \|_n \leq \sum_{k=2}^{\infty} K_k \lambda^k \leq K_2 \lambda^2 \sum_{k=2}^{\infty} \left(\frac{1}{2}\right)^{k-2} = 2K_2 \lambda^2 = 2C m(1) \lambda^2 \approx 2.9 \lambda^2/\lambda_0$. (Here we use the fact that $\| \cdot \|_{n+1} \leq \| \cdot \|_n$.)

In the next subsections, we will give a proof of Theorem \ref*{mainthm}$'$.

\subsubsection{Iterative construction}
\label{sec:recusion}
The idea is to construct the $D_n,V_n,E_n,A_n$ discussed above iteratively. That is, suppose that at the $n$-th step, we have 
\begin{equation}
\mathcal{U}_n \Uf \,\mathcal{U}_n^{\dagger} \equiv \Uf^{(n)} =
X \,\mathcal{T}\!\exp\left(-i\int_{0}^{1} \mathcal{H}_n(t) dt\right),
\end{equation}
where $\mathcal{H}_n(t) = F_n + V_n(t)$, with $F_n =  \int_0^T \mathcal{H}_n(t) dt$ time-independent. We will choose to separate the time-independent piece $F_n$ according to $F_n = D_n + E_n$, where $D_n = \langle F_n \rangle$, and we have defined the symmetrization
\begin{equation}
\langle O \rangle = \frac{1}{N} \sum_{k=0}^{N-1} X^{-k} O X^k.
\end{equation}
In particular, this implies that $[D_n,X] = 0$ and $\langle D_n \rangle = D_n$, and therefore 
$\langle E_n \rangle = \langle F_n \rangle - \langle D_n \rangle = D_n - D_n = 0$.

We will now introduce a local unitary $\mathcal{A}_n = e^{-iA_n}$, which we use to rotate the Floquet operator $\Uf^{(n)}$, giving a new Floquet operator 
\begin{equation}
\Uf^{(n+1)} \equiv \mathcal{A}_n \Uf^{(n)} \mathcal{A}_n^{\dagger} = X \mathcal{T} \exp\left(-i\int_0^1 \mathcal{H}_{n+1}(t) dt\right).
\end{equation}
The ultimate goal, decomposing $\mathcal{H}_{n+1}(t) = D_{n+1} + E_{n+1} + V_{n+1}(t)$ as before, is to ensure that the residual error terms $E_{n+1}$ and $V_{n+1}$ are much smaller than $E_n$ and $V_n$. This goal is achieved in two separate steps. The first step ensures that $E_{n+1}$ is small (that is, the time-independent part of $\mathcal{H}_{n+1}(t)$ nearly commutes with $X$), and the second step ensures that $V_{n+1}$ is small.

\emph{Step One}.-- This step proceeds similarly to the recursion relation of Abanin et al \cite{Abanin2015a} for the \emph{time-independent} case (Section 5.4 of Ref. \onlinecite{Abanin2015a}). There the recursion relation was designed to make the Hamiltonian commute with its zero-th order version. This is analogous to our present goal of making the Floquet operator commute with $X$. Here, we adapt the analysis of Ref.~\onlinecite{Abanin2015a} to the Floquet case.

We observe that
\begin{align}
\Uf^{(n+1)} &= 
\mathcal{A}_n \Uf^{(n)} \mathcal{A}_n^{\dagger} \\
&= X \left[X^{\dagger} \mathcal{A}_n X \times \mathcal{T} \exp\left(-i\int_0^1 \cH_n(t) dt\right) \times \mathcal{A}_n^{\dagger}\right], \\
&= X \left[e^{-X^{\dagger} iA_n X} \times \mathcal{T} \exp\left(-i\int_0^1 \cH_n(t) dt\right) \times e^{iA_n}\right] \\
&= X \times \mathcal{T} \exp\left(-i\int_0^1 \cH_n^{\prime}(t) dt\right),
\end{align}
where
\begin{equation}
\label{hp}
\cH_n^{\prime}(t) = \begin{cases} \frac{1}{a} (-A_n) & 0 \leq t \leq a \\
    \frac{1}{1-2a} \cH_n\left( \frac{t-a}{1-2a} \right) & a\leq t \leq (1-a), \\
                              \frac{1}{a}(X^{\dagger} A_n X) & (1-a) \leq t \leq 1,
\end{cases}
\end{equation}
(for some constant $a \in [0,1/2]$ which can be chosen arbitrarily.)
Let us decompose $\cH_n^{\prime}(t) = D^{\prime}_n+ V^{\prime}_n(t)$, where $D^{\prime}_n = \frac{1}{T}\int_0^1 \cH^{\prime}_n(t)$. Our goal will be to ensure that the time-independent part $D_n^{\prime}$ commutes with $X$. It turns out this can actually be achieved exactly, and in particular we can choose $A_n$ such that $D_n^{\prime} = D_n$.

To this end, we first observe that
\begin{equation}
D^{\prime}_n = D_n + E_n + X^{\dagger} A_n X - A_n.
\end{equation}
We now claim that $D_n^{\prime} = D_n$ if we choose
\begin{equation}
\label{theA}
A_n := \frac{1}{N}\sum_{k=0}^{N-1} \sum_{p=0}^k E^{(p)}_n, \quad E^{(p)}_n = X^{-p} E X^{p}.
\end{equation}
To see this, note that, by construction,
\begin{align}
X^{\dagger} A_n X - A_n &= \frac{1}{N}\sum_{k=0}^{N-1} \sum_{p=0}^k [E^{(p+1)}_n - E^{(p)}_n] \\
                    &= \frac{1}{N}\sum_{k=0}^{N-1} [E^{(k+1)}_n - E_n] \\
                    &= -E_n + \langle E_n \rangle, \\
                    &= -E_n,
\end{align}
since $\langle E_n \rangle = 0$.

\emph{Step Two}.--
The next step is now to find a new time-dependent Hamiltonian $\cH_{n+1}(t)$ which gives the same unitary evolution as $\cH_{n}^{\prime}(t)$ over the time interval $[0,1]$, while making the time-dependent part smaller. That is, making the decomposition $\cH_{n+1}(t) = D_{n+1} + E_{n+1} + V_{n+1}(t)$ as before, the goal is to make $V_{n+1}$ small. In fact, this is precisely the problem already considered by Abanin et al\cite{Abanin2015a}, and we can use the procedure described in Section 4.1 of that paper.

One might worry whether Step Two undoes the good work done by Step One. That is, does making $V_{n+1}$ small come at the cost of making $E_{n+1}$ larger again? However, this turns out not to be a problem, as the bounds we derive below will make clear.

\subsubsection{Bounds on Error terms}
\label{sec:bounds}

Now we will derive bounds that quantify the success of the iterative procedure described in the previous subsection at making the residual error terms $E_n$ and $V_n$ small. Analysis proceeds in similar way to Abanin et al\cite{Abanin2015a}.
 We define
\begin{multline}
d(n) = \| D_n \|_n,  \quad v(n) = \| V_n \|_n, \quad v'(n) = \| V^{\prime}_n \|_{n}, \\ e(n) = \| E_n \|_n, \quad \delta d(n) = \| D_{n+1} - D_n \|_{n+1}, 
\end{multline}

First of all, from \eqnref{theA} we have a bound on $A_n$:
\begin{equation}
\label{Anbound}
\| A_n \|_n \leq \frac{N+1}{2} e(n)
\end{equation}
From \eqnref{hp} we observe that
\begin{equation}
\label{vp}
V_n^{\prime}(t) = \begin{cases} \frac{1}{a} (-A_n) - D_n & 0 \leq t \leq a \\
    \frac{1}{1-2a}\left[ 2a D_n + E_n + V_n\left(\frac{t-a}{1-2a}\right) \right] & a\leq t \leq (1-a), \\
                              \frac{1}{a}(X^{\dagger} A_n X) - D_n & (1-a) \leq t \leq 1,
\end{cases}
\end{equation}
and hence
\begin{align}
v'(n) &\leq 2 \| A_n \|_n + \|E_n\|_n + \| V_n \|_n + 4a \| D_n \|_n
\end{align}
Hence, we can send $a \to 0$ to give (using \eqnref{Anbound})
\begin{align}
\label{oslo}
v'(n) &\leq (N+2) e(n) + v(n).
\end{align}
%\begin{equation}
%v'(n) \leq \max\left\{ \frac{(N+1)}{2a} e(n), \,
%\frac{1}{1-2a} v(n) \right\}. \label{vpbound}
%\end{equation}
Then, as our construction of $\cH_{n+1}$ from $\cH_n^{\prime}$ is the one described in Section 4.1 of Abanin et al, we can use their bounds
\begin{align}
\label{abanin_bounds}
\| D_{n+1} + E_{n+1} - D_n \|_{n+1} &\leq \epsilon_n/2 \\
\label{copenhagen}
v(n+1) &\leq \epsilon_n
\end{align}
where 
\begin{equation}
\label{stockholm}
\epsilon_n =
 m(n) v'(n) \bigl(d(n) + 2 v'(n)),
\end{equation}
\begin{equation}
\label{mn}
m(n) = \frac{18}{(\kappa_{n+1} - \kappa_n)\kappa_{n+1}}.
\end{equation}
These bounds hold provided that 
\begin{equation}
\label{brussels}
3 v'(n) \leq \kappa_{n} - \kappa_{n+1}
\end{equation}
Since $D_{n+1} - D_n = \langle D_{n+1} + E_{n+1} - D_n \rangle$, we see that
\begin{equation}
\delta d(n) \leq \| D_{n+1} + E_{n+1} - D_n \|_{n+1} \leq \epsilon_n/2
\end{equation}
and
\begin{equation}
e(n+1) \leq \| D_{n+1} + E_{n+1} - D_n \|_{n+1} + \| D_{n+1} - D_n \|_{n+1} \leq \epsilon_n
\end{equation}

\subsubsection{Proof of Theorem \ref*{mainthm}$'$ by induction}
The idea now is to apply the bounds of the previous subsection recursively to give bounds expressed in terms of the original Floquet operator,
\begin{align}
    \Uf = \Uf^{(1)} &= \mathcal{T} \exp\left(-i\int_0^1 [H_0(t) + V(t)] \right) 
    \\ &= X \mathcal{T} \exp\left(-i\int_0^1 V_{\mathrm{int}}(t) dt\right),
\end{align}
and in particular the quantity $\lambda \equiv \| V_{\mathrm{int}} \|_1 = \| V \|_1$.
First of all, we write $\cH_1(t) \equiv V_{\mathrm{int}}(t) = F_1 + V_1(t)$, where $F_1 = \int_0^1 V_{\mathrm{int}}(t) dt$, and then separate $F_1 = D_1+ E_1$, where $D_1 = \langle F_1 \rangle$. We note that $\|F_1\|_1 \leq \lambda$, which implies that $v(1) \leq \| V_{\mathrm{int}} \|_1 + \|F_1\|_1 \leq 2\lambda$, and $d(1) \leq \lambda$. In turn this gives $e(1) \leq \|D_1\|_1 + \|F_1\|_1 \leq 2\lambda$.

Now we proceed by induction. Suppose that we have some $n$ such that, for all $1 \leq k \leq n$, we have
\begin{equation}
e(k), v(k) \leq 2K_k \lambda^k,
\end{equation}
and for all $1 \leq k < n$,
\begin{equation}
\delta d(k) \leq K_{k+1} \lambda^{k+1}
\end{equation}
where the coefficients $K_k$ satisfy $K_{k+1}/K_k \leq \frac{1}{2\lambda}$.
(The preceding discussion shows that this induction condition is satisfied for $n=1$ with $K_1=1$.)

Then from \eqnref{oslo} we find that
\begin{equation}
v^{\prime}(n) \leq 2c_N  K_n \lambda^{n}, \quad c_N = N+3,
\end{equation}
and hence
\begin{equation}
\label{berlin}
\epsilon_n \leq m(n) 2c_N K_n \lambda^n ( d(n) + 2c_N K_n \lambda^n).
\end{equation}
We note that the triangle inequality and the fact that $\|\cdot\|_n$ decreases with $n$ ensures that $d(n+1) - d(n) \leq \delta d(n)$. Hence we can bound $d(n)$ by
\begin{align}
d(n) &\leq d(1) + \sum_{k=1}^{n-1} \delta d(k)  \\
     &\leq \lambda  + \sum_{k=1}^{n-1} K_{k+1} \lambda^{k+1} \\
     &= \sum_{k=1}^n K_k \lambda^k \\
     &\leq \sum_{k=1}^n \lambda \left(\frac{1}{2}\right)^{k-1} \label{helsinki} \\
     &\leq 2 \lambda
\end{align}
In \eqnref{helsinki}, we used the inequality $K_{k+1}/K_k \leq 1/(2\lambda)$. This same inequality also ensures that $K_n \lambda^n \leq \lambda$, so inserting into \eqnref{berlin} gives
\begin{align}
\epsilon_n &\leq m(n) 2c_N K_n (2 + 2c_N)\lambda^{n+1}\cr
&\equiv 2C m(n) K_n \lambda^{n+1}\cr
& \equiv K_{n+1} \lambda^{n+1}.
\end{align}
Here we chose
\begin{equation}
K_{n+1} = C m(n) K_n, \quad C = 2c_N(1+c_N).
\end{equation}
Next we need to examine the conditions under which \eqnref{brussels} holds. Given the bounds on $v'(n)$ and using the inequality $K_n \lambda^n \leq \lambda (1/2)^{n-1}$, it is sufficient to demand that
\begin{equation}
3c_N (1/2)^{n-1} \lambda \leq \kappa_{n+1} - \kappa_n,
\end{equation}
or in other words
\begin{equation}
\label{paris}
\lambda \leq \frac{1}{3c_N} \max_{n \in \mathbb{N}} \left[2^{n-1} (\kappa_{n+1}
- \kappa_n) \right] = \frac{1}{3c_N}(\kappa_2 - \kappa_1) \approx
\frac{0.14\kappa_1}{N+3}.
\end{equation}
Provided that \eqnref{paris} holds, we then find that 
\begin{align}
\delta d(n), v(n+1)/2, e(n+1)/2 \leq  K_{n+1} \lambda^{n+1}.
\end{align}
Therefore, we can continue the induction provided that $K_{n+1}/K_n \leq \frac{1}{2\lambda}$. Since $K_{n+1}/K_n = C m(n)$, this is true provided that $n \leq n_*$. This completes the proof of Theorem \ref*{mainthm}$'$.

\section{Proof of phase-winding when a $\U(1)$ symmetry is spontaneously broken}
\label{phase_winding}

Here we intend to prove the claim made in Section \ref{continous_basic_picture} above that the expectation value
\begin{equation}
\mathrm{Tr}(\rho_{X} e^{itK} \Phi e^{-itK}) \equiv g_X(t)
\end{equation}
must be independent of time $t$, where we have defined $K \equiv D - \mu L$ and $\rho_X \equiv \lim_{\epsilon\to0^{+}} \frac{1}{\mathcal{Z}} e^{-\beta (K + \epsilon X)}$. The idea is to make a connection with results of Ref.~\onlinecite{Watanabe15}; however, these  were expressed in terms of \emph{two-point} correlation functions, and also did not have the $\epsilon X$ term in the definition of the density matrix. To make a connection, we assume that the symmetric density matrix $\rho = \frac{1}{\mathcal{Z}} e^{-\beta K}$ can be recovered by symmetrizing a symmetry-breaking state,
\begin{equation}
  \rho = \frac{1}{2\pi} \int_0^{2\pi} e^{-i\theta L} \rho_X e^{i\theta L}
  d\theta,
\end{equation}
and that the symmetry-breaking state $\rho_X$ is short-range correlated. Now we calculate the two-point correlation function (where $\Phi(x)$ and $\Phi(y)$ are two operators acting at different spatial locations $x$ and $y$)
\begin{widetext}
\begin{align}
f(t) &= \mathrm{Tr}[\rho e^{itK} \Phi(x) e^{itK} \Phi^{\dagger}(y)] \\
     &= \frac{1}{2\pi} \int_0^{2\pi} d\theta \mathrm{Tr} [e^{-i\theta L} \rho_X e^{i\theta L} e^{itK} \Phi(x) e^{-itK} \Phi(y)] \\
     &= \frac{1}{2\pi} \int_0^{2\pi} d\theta \mathrm{Tr} [\rho_X e^{itK} \{
     e^{i\theta L} \Phi(x) e^{-i \theta L}\} e^{-itK} \{ e^{i\theta L} \Phi^{\dagger}(y) e^{-i\theta L}\}] \\
     &= \mathrm{Tr} [\rho_X \{ e^{-itK} \Phi(x) e^{itK} \} \Phi^{\dagger}(y) \}] \\
     &= g_X(t) [g_X(0)]^{*},
\end{align}
\end{widetext}
where we used the fact that $L$ and $K$ commute and that $e^{i\theta L} \Phi e^{-i\theta L} = e^{i\theta} \Phi$. In the last line we sent $|x - y| \to \infty$ and used the assumption that $\rho_X$ has short-range correlations.

Now, the theorem of Ref.~\onlinecite{Watanabe15} rigorously proves that the function $f(t)$ must be independent of time. Hence, unless $g_X(0) = 0$, we conclude that $g_X(t)$ must be independent of time. (If $g_X(0) = 0$ but $g_X(t)$ is not independent of time then there must be some $t$ such that $g_X(t) \neq 0$. Then we can just relabel the time-coordinate so that $g_X(0) \neq 0$ and repeat the argument.)

\section{Open systems}
\label{appendix_open}
In this section, we will elaborate on our hypothesis for open systems introduced in Section \ref{sec_open} above, namely that in a large class of systems the steady state will have low energy.
First we need to clarify what we mean by ``energy'' and ``steady state'' in the
Floquet context. Let $H_S(t)$ be the time-evolution of the system alone (not
taking to account the coupling to the environment.) We define the Floquet
operator $\Uf = \mathcal{T} \exp\left(-i\int_0^T H_S(t) dt\right)$. Recall that
in the regime discussed in Section \ref{sec_floquet_time_crystal}, where $\lambda$ as defined there
satisfies $\lambda T \ll 1$, we can write $H_S(t) = \widetilde{H}_S(t) + V(t)$. Here $V(t)$ is a very weak residual perturbation, and $\widetilde{H}_S(t)$ is such that, if we define the approximate Floquet operator by $\widetilde{\Uf} = \mathcal{T} \exp\left(-i\int_0^T \widetilde{H}_S(t)\right)$, then it can be expressed, following a local unitary time-independent change of basis (which we will here set to $1$ for notational simplicity), as $\widetilde{\Uf} = X e^{-iDT}$, where $X^2 = 1$ and $D$ is a quasi-local Hamiltonian $D$ that commutes with $X$. In particular, we have $\widetilde{\Uf}^2 = e^{-2iDT}$. This implies that we can make a time-\emph{dependent} local unitary change of basis $W(t)$, periodic with period $2T$ and satisfying $W(0) = 1$, such that the transformed Hamiltonian, which is related to $\widetilde{H}_S(t)$ according to
\begin{equation}
    \widetilde{H}_S' = W H_S W^{\dagger} + i [\partial_t W] W^{\dagger},
\end{equation}
is time-independent and equal to $D$.
Therefore, in this new reference frame, it is clear that we should refer to the expectation value of $D$ as ``energy''. We emphasize that we have not gotten rid of the time-dependence completely: even in the new reference frame the residual driving term $V(t)$, as well as any couplings to the environment, will still be time-dependent. (Due to the time-dependent change of basis, the latter will gain a time-dependence even if it was originally time-independent.)

The steady state is now determined by some balance between the residual periodic
driving $V(t)$, the classical noise, and the coupling to the environment. We
leave a detailed analysis of this open system process for future
work\footnote{For one study of steady states of many-body Floquet systems
coupled to a bath, see Ref.~\onlinecite{Shirai2015}},
but we expect that in a suitable regime the energy-density of the steady state will be low. We will now explain why this implies oscillations (which are observed in the \emph{original} reference frame, not the rotating one defined above.)

Consider a short-range correlated steady state $\rho$ whose energy density with respect
to $D$ is small. Recall that in Section \ref{floquet_basic_picture} we argued that if $\rho$ is a
thermal state it must spontaneously break the symmetry generated by $X$, and it
follows that under $\widetilde{\Uf}$ it oscillates at twice the drive
frequency. Of course, for an open system the steady-state need not be thermal,
and time evolution of the open system is not exactly given by $\widetilde{\Uf}$.
However, as we prove in Appendix \ref{ssb_nonthermal}, even non-thermal states must fail to be
invariant under the symmetry $X$ if their energy density with respect to $D$ is
sufficiently small, provided that they satisfy a physically reasonable
``thermalizability'' condition. Moreover, if $\lambda T \ll 1$ (so that we can
approximate $\widetilde{\Uf} \approx X$), and the coupling
to the environment sufficiently weak, then the resulting state after one time
period is approximately given by $X \rho X^{\dagger}$, which by the preceding
discussion is \emph{not} the same as $\rho$. (We make this argument more precise
in Section \ref{ssb_nonthermal}.) Thus, provided that the energy of the
steady-state is sufficiently small, it does not return to itself after one time
period, and oscillations with period $2T$ will be observed.

Generic baths will destroy continuous-time time crystals. The difference with the discrete-time case is the existence of an extra variable characterizing thermal states of $D$; namely, the chemical potential $\mu$. This extra variable is needed because of the presence of the hidden $U(1)$ symmetry in the continuous-time regime. (There is no analogous variable when the hidden symmetry is \emph{discrete}). Thus, one certainly cannot make any statement that all low-energy states of $D$ oscillate, because, in particular, a \emph{thermal} state of $D$ in which the electrochemical potential $\mu-u=0$ does not oscillate. A coupling to a generic bath will not preserve the hidden $U(1)$ symmetry, and thus to the extent that the steady state of an open system process is close to a thermal state of $D$, we in fact expect it to have $\mu-u=0$, since this corresponds to minimizing the free energy.

In principle, one could fine-tune the bath so that it repects the symmetry. This would
allow the time crystal to survive, but is clearly contrived. One might wonder whether the
bath itself could also pre-thermalize: if we could consider the bath to be included in the Hamiltonian
(\ref{eqn:cTTSB-Hamiltonian}) then it could have an approximate U(1) symmetry along with the
rest of the system. This would require the local terms in the bath Hamiltonian to be much smaller than the
coupling $u$ in Eq. (\ref{eqn:cTTSB-Hamiltonian}). However, for most of the physically relevant baths that one would want to consider (for example, phonons), the local terms in the bath Hamiltonian are in fact unbounded.

\section{Spontaneous symmetry breaking for non-thermal states}
\label{ssb_nonthermal}
Let $D$ be a quasi-local Hamiltonian for which the thermal states spontaneously break an on-site $\mathbb{Z}_N$ symmetry generated by $X$ for energy densities $e < e_c$. More precisely, what we mean is the following, where we define the local distance between two states on a region $A$ according to
\begin{equation}
\| \rho_1 - \rho_2 \|_{A} = \| (\rho_1)_A - (\rho_2)_A \|_1
\end{equation}
where $\| \cdot \|_1$ is the trace norm, and $(\rho)_A = \mathrm{Tr}_{A^c} \rho$ is the reduced state of $\rho$ on $A$.
\begin{assumption}[Spontaneous symmetry-breaking]
    \label{assumption_ssb}
There exists some finite region $A$ and some $\gamma > 0$, such that, for any short-range correlated thermal state $\rho_\tau$ with energy density $e < e_c$, we have $\| \rho_\tau - X^k \rho_\tau X^{-k} \|_A \geq \gamma$ for all $0 < k < N$.
\end{assumption}

Now let $\rho$ be \emph{any} state (not necessarily thermal) such that the energy density $\epsilon \equiv \langle D \rangle_{\rho}/V < \epsilon_c$ (with $V$ the volume of the system.) We assume the following \emph{thermalizability} condition, which roughly states that $\rho$ can thermalize when time-evolved under $D$. More precisely:
\begin{assumption}[Thermalizability]
    \label{assumption_thermalizability}
There exist a time $t_1$ and a short-range correlated thermal state $\rho_\tau$ with the same energy density as $\rho$, such that $\| \rho(t_1) - \rho_\tau \|_A \leq \gamma/8$, where $\rho(t) = e^{-iDt_1} \rho e^{iDt_1}$.
\end{assumption}

From Assumptions \ref{assumption_ssb} and \ref{assumption_thermalizability} we derive the following lemma, which quantifies the sense in which the state $\rho$ must break the symmetry.
\begin{lem}
There exists a finite region $A'$ such that $\| \rho - X^k \rho X^{-k} \|_{A^{\prime}} \geq 3\gamma/4$.
\end{lem}
\begin{proof}
From the triangle inequality it follows that
\begin{widetext}
\begin{align}
& \| \rho(t_1) - X^k \rho(t_1) X^{-k} \|_A \\
& \geq \| \rho_\tau - X^k \rho_\tau X^{-k} \|_A - \| \rho(t_1) - X^k \rho(t_1) X^{-k} - (\rho_\tau - X^k \rho_\tau X^{-k}) \|_A \\
& \geq \gamma - 2\gamma/8\\
&= 3\gamma/4.
\end{align}
\end{widetext}

Using the characterization of the trace norm as
\begin{equation}
\| \rho \|_1 = \sup_{\hat{o} : \| \hat{o} \| = 1} |\langle \hat{o} \rangle_\rho|,
\end{equation}
it follows that there exists an operator $\hat{o}_A$ supported on $A$, with $\| \hat{o}_A \| = 1$, such that $|\langle X^{-k} \hat{o}_A X^k - \hat{o}_A\rangle_{\rho(t_1)}| \geq 3\gamma/4$. Now, since $D$ is quasi-local, it must obey a Lieb-Robinson bound \cite{Lieb1972,Nachtergaele2006}, which implies that there exists a local operator $\hat{O}_{A^{\prime}}$ supported on a finite region $A^{\prime}$ such that $\| \hat{o}(t_1) - \hat{O}_{A^{\prime}}\| \leq \gamma/8$, where $\hat{o}(t_1) = e^{iDt_1} \hat{o} e^{-iDt_1}$.
Hence we see that
\begin{align}
& |\langle X^{-k} \hat{O}_{A'} X^{k} - \hat{O}_{A'} \rangle_{\rho}| \\
&\geq -\gamma/4 + | \langle X^{-k} \hat{o}_A(t_1) X^k - \hat{o}_A(t_1) \rangle_{\rho} | \\
&= -\gamma/4 + | \langle X^{-k} \hat{o}_A X^k - \hat{o}_A \rangle_{\rho(t_1)} | \label{foogaz} \\
&\geq -\gamma/4 + 3\gamma/4. \\
&= \gamma/2.
\end{align}
To get to line \eqnref{foogaz}, we used the fact that $X$ and $D$ commute. The lemma follows.
\end{proof}

Now consider a system which in isolation would evolve under a time-dependent Hamiltonian $H(t)$, which is periodic with period $T$. We assume that $H(t)$ exhibits the pre-thermalization phenomena discussed in the main text. That is, we assume that the Floquet operator can be approximated according to $\Uf \approx \widetilde{\Uf} = X e^{-iDT}$, where $D$ is quasi-local and commutes with $X$, and where $\Uf$ is close to $\widetilde{\Uf}$ in the sense that
\begin{equation}
\| \Uf^{\dagger} O_{A'} \Uf - \widetilde{\Uf}^{\dagger} O_{A'} \widetilde{\Uf} \| \leq \frac{\gamma}{8} \|O_{A'} \|
\end{equation}
for any operator $O_{A'}$ supported on $A'$. 

Let $\rho_{\mathrm{open}}(t)$ be the reduced state of the system (tracing out
the bath) at time $t$, taking into account the system-bath coupling, and we
assume that $\rho_{\mathrm{open}}(0) \equiv \rho$ satisfies Assumption
\ref{assumption_thermalizability} above. We assume the coupling to the bath is sufficiently weak, in the following sense:
\begin{assumption}[Weak coupling]
\label{assumption_weak_coupling}
For any time $0 \leq t \leq T$, we have
$\| \rho_{\mathrm{open}}^{\mathrm{int}}(t) - \rho \|_{A'} \leq \gamma/8$.
\end{assumption}
Here we defined the interaction picture state $\rho^{\mathrm{int}}_{\mathrm{open}}(t) = U(0,t)^{-1} \rho_{\mathrm{open}}(t) U(0,t)$, where $U(0,t)$ is the time evolution generated by $H(t)$. If we were to set the coupling to the bath to zero then the state $\mathrm{\rho}_{\mathrm{open}}^{\mathrm{int}}(t)$ would be constant in time, so Assumption \ref{assumption_weak_coupling} corresponds to weak coupling. Finally, we will assume that the strength of $DT$ is small enough so that
\begin{assumption}
For any observable $O_{A'}$ supported on $A'$, we have
\begin{equation}
\| e^{-iDT} O_{A'} e^{iDT} - O_{A'} \| \leq \frac{\gamma}{8}\|O_{A'}\|
\end{equation}
\end{assumption}
This will always be true in the regime of interest, $\lambda T \ll 1$ (where
$\lambda$ is as defined in Section \ref{prethermalization_results}), because
$\|D\|_{n_*}$ is $O(\lambda)$ [see \eqnref{D_first_order} in Theorem
\ref{mainthm}].
 
From the above assumptions we can now derive our main result:
 \begin{thm}
\begin{equation}
\| \rho_{\mathrm{open}}(T) - \rho\|_{A'} \geq \gamma/8.
\end{equation}
\begin{proof}
 \begin{align}
& \| \rho_{\mathrm{open}}(T) - \rho \|_{A'} \\
&= \| \Uf \rho^{\mathrm{int}}_{\mathrm{open}}(T) \Uf^{\dagger} - \rho \|_{A'} \\
&\geq -\gamma/8 + \| \WUf \rho^{\mathrm{int}}_{\mathrm{open}}(T) \WUf^{\dagger} - \rho \|_{A'} 
\\
&= -\gamma/8 + \| e^{-iDT} \rho^{\mathrm{int}}_{\mathrm{open}}(T) e^{iDT} - X^{\dagger} \rho X \|_{A'} \\
&\geq -\gamma/8 - \gamma/8 + \| \rho^{\mathrm{int}}_{\mathrm{open}}(T) - X^{\dagger} \rho X \|_{A'} \\
&\geq -\gamma/8 - \gamma/8 - \gamma/8 + \| \rho - X^{\dagger} \rho X \|_{A'} \\
&\geq -\gamma/8 - \gamma/8 - \gamma/8 + \gamma/2. \\
&= \gamma/8.
\end{align}
\end{proof}
\end{thm}
In other words, the state of the open system at times $t=T$ and $t=0$ are locally distinguishable. That is, for the stated assumptions, the state of the system does not synchronize with the drive and time translation symmetry is spontaneously broken.


\begin{thebibliography}{100}%
\makeatletter
\providecommand \@ifxundefined [1]{%
 \@ifx{#1\undefined}
}%
\providecommand \@ifnum [1]{%
 \ifnum #1\expandafter \@firstoftwo
 \else \expandafter \@secondoftwo
 \fi
}%
\providecommand \@ifx [1]{%
 \ifx #1\expandafter \@firstoftwo
 \else \expandafter \@secondoftwo
 \fi
}%
\providecommand \natexlab [1]{#1}%
\providecommand \enquote  [1]{``#1''}%
\providecommand \bibnamefont  [1]{#1}%
\providecommand \bibfnamefont [1]{#1}%
\providecommand \citenamefont [1]{#1}%
\providecommand \href@noop [0]{\@secondoftwo}%
\providecommand \href [0]{\begingroup \@sanitize@url \@href}%
\providecommand \@href[1]{\@@startlink{#1}\@@href}%
\providecommand \@@href[1]{\endgroup#1\@@endlink}%
\providecommand \@sanitize@url [0]{\catcode `\\12\catcode `\$12\catcode
  `\&12\catcode `\#12\catcode `\^12\catcode `\_12\catcode `\%12\relax}%
\providecommand \@@startlink[1]{}%
\providecommand \@@endlink[0]{}%
\providecommand \url  [0]{\begingroup\@sanitize@url \@url }%
\providecommand \@url [1]{\endgroup\@href {#1}{\urlprefix }}%
\providecommand \urlprefix  [0]{URL }%
\providecommand \Eprint [0]{\href }%
\providecommand \doibase [0]{http://dx.doi.org/}%
\providecommand \selectlanguage [0]{\@gobble}%
\providecommand \bibinfo  [0]{\@secondoftwo}%
\providecommand \bibfield  [0]{\@secondoftwo}%
\providecommand \translation [1]{[#1]}%
\providecommand \BibitemOpen [0]{}%
\providecommand \bibitemStop [0]{}%
\providecommand \bibitemNoStop [0]{.\EOS\space}%
\providecommand \EOS [0]{\spacefactor3000\relax}%
\providecommand \BibitemShut  [1]{\csname bibitem#1\endcsname}%
\let\auto@bib@innerbib\@empty
%</preamble>
\bibitem [{\citenamefont {Gu}\ and\ \citenamefont {Wen}(2009)}]{Gu2009}%
  \BibitemOpen
  \bibfield  {author} {\bibinfo {author} {\bibfnamefont {Zheng-Cheng}\
  \bibnamefont {Gu}}\ and\ \bibinfo {author} {\bibfnamefont {Xiao-Gang}\
  \bibnamefont {Wen}},\ }\bibfield  {title} {\enquote {\bibinfo {title}
  {{Tensor-entanglement-filtering renormalization approach and
  symmetry-protected topological order}},}\ }\href {\doibase
  10.1103/PhysRevB.80.155131} {\bibfield  {journal} {\bibinfo  {journal} {Phys.
  Rev. B}\ }\textbf {\bibinfo {volume} {80}},\ \bibinfo {pages} {155131}
  (\bibinfo {year} {2009})},\ \Eprint {http://arxiv.org/abs/0903.1069}
  {arXiv:0903.1069} \BibitemShut {NoStop}%
\bibitem [{\citenamefont {Pollmann}\ \emph {et~al.}(2010)\citenamefont
  {Pollmann}, \citenamefont {Turner}, \citenamefont {Berg},\ and\ \citenamefont
  {Oshikawa}}]{Pollmann2010}%
  \BibitemOpen
  \bibfield  {author} {\bibinfo {author} {\bibfnamefont {Frank}\ \bibnamefont
  {Pollmann}}, \bibinfo {author} {\bibfnamefont {Ari~M.}\ \bibnamefont
  {Turner}}, \bibinfo {author} {\bibfnamefont {Erez}\ \bibnamefont {Berg}}, \
  and\ \bibinfo {author} {\bibfnamefont {Masaki}\ \bibnamefont {Oshikawa}},\
  }\bibfield  {title} {\enquote {\bibinfo {title} {{Entanglement spectrum of a
  topological phase in one dimension}},}\ }\href {\doibase
  10.1103/PhysRevB.81.064439} {\bibfield  {journal} {\bibinfo  {journal} {Phys.
  Rev. B}\ }\textbf {\bibinfo {volume} {81}},\ \bibinfo {pages} {064439}
  (\bibinfo {year} {2010})},\ \Eprint {http://arxiv.org/abs/0910.1811}
  {arXiv:0910.1811} \BibitemShut {NoStop}%
\bibitem [{\citenamefont {Pollmann}\ \emph {et~al.}(2012)\citenamefont
  {Pollmann}, \citenamefont {Berg}, \citenamefont {Turner},\ and\ \citenamefont
  {Oshikawa}}]{Pollmann2012}%
  \BibitemOpen
  \bibfield  {author} {\bibinfo {author} {\bibfnamefont {Frank}\ \bibnamefont
  {Pollmann}}, \bibinfo {author} {\bibfnamefont {Erez}\ \bibnamefont {Berg}},
  \bibinfo {author} {\bibfnamefont {Ari~M.}\ \bibnamefont {Turner}}, \ and\
  \bibinfo {author} {\bibfnamefont {Masaki}\ \bibnamefont {Oshikawa}},\
  }\bibfield  {title} {\enquote {\bibinfo {title} {{Symmetry protection of
  topological phases in one-dimensional quantum spin systems}},}\ }\href
  {\doibase 10.1103/PhysRevB.85.075125} {\bibfield  {journal} {\bibinfo
  {journal} {Phys. Rev. B}\ }\textbf {\bibinfo {volume} {85}},\ \bibinfo
  {pages} {075125} (\bibinfo {year} {2012})},\ \Eprint
  {http://arxiv.org/abs/0909.4059} {arXiv:0909.4059} \BibitemShut {NoStop}%
\bibitem [{\citenamefont {Fidkowski}\ and\ \citenamefont
  {Kitaev}(2010)}]{Fidkowski2010}%
  \BibitemOpen
  \bibfield  {author} {\bibinfo {author} {\bibfnamefont {Lukasz}\ \bibnamefont
  {Fidkowski}}\ and\ \bibinfo {author} {\bibfnamefont {Alexei}\ \bibnamefont
  {Kitaev}},\ }\bibfield  {title} {\enquote {\bibinfo {title} {{Effects of
  interactions on the topological classification of free fermion systems}},}\
  }\href {\doibase 10.1103/PhysRevB.81.134509} {\bibfield  {journal} {\bibinfo
  {journal} {Phys. Rev. B}\ }\textbf {\bibinfo {volume} {81}},\ \bibinfo
  {pages} {134509} (\bibinfo {year} {2010})},\ \Eprint
  {http://arxiv.org/abs/0904.2197} {arXiv:0904.2197} \BibitemShut {NoStop}%
\bibitem [{\citenamefont {Chen}\ \emph {et~al.}(2010)\citenamefont {Chen},
  \citenamefont {Gu},\ and\ \citenamefont {Wen}}]{Chen2010}%
  \BibitemOpen
  \bibfield  {author} {\bibinfo {author} {\bibfnamefont {Xie}\ \bibnamefont
  {Chen}}, \bibinfo {author} {\bibfnamefont {Zheng-Cheng}\ \bibnamefont {Gu}},
  \ and\ \bibinfo {author} {\bibfnamefont {Xiao-Gang}\ \bibnamefont {Wen}},\
  }\bibfield  {title} {\enquote {\bibinfo {title} {{Local unitary
  transformation, long-range quantum entanglement, wave function
  renormalization, and topological order}},}\ }\href {\doibase
  10.1103/PhysRevB.82.155138} {\bibfield  {journal} {\bibinfo  {journal} {Phys.
  Rev. B}\ }\textbf {\bibinfo {volume} {82}},\ \bibinfo {pages} {155138}
  (\bibinfo {year} {2010})},\ \Eprint {http://arxiv.org/abs/1004.3835}
  {arXiv:1004.3835} \BibitemShut {NoStop}%
\bibitem [{\citenamefont {Chen}\ \emph
  {et~al.}(2011{\natexlab{a}})\citenamefont {Chen}, \citenamefont {Gu},\ and\
  \citenamefont {Wen}}]{Chen2011}%
  \BibitemOpen
  \bibfield  {author} {\bibinfo {author} {\bibfnamefont {Xie}\ \bibnamefont
  {Chen}}, \bibinfo {author} {\bibfnamefont {Zheng-Cheng}\ \bibnamefont {Gu}},
  \ and\ \bibinfo {author} {\bibfnamefont {Xiao-Gang}\ \bibnamefont {Wen}},\
  }\bibfield  {title} {\enquote {\bibinfo {title} {{Classification of gapped
  symmetric phases in one-dimensional spin systems}},}\ }\href {\doibase
  10.1103/PhysRevB.83.035107} {\bibfield  {journal} {\bibinfo  {journal} {Phys.
  Rev. B}\ }\textbf {\bibinfo {volume} {83}},\ \bibinfo {pages} {035107}
  (\bibinfo {year} {2011}{\natexlab{a}})},\ \Eprint
  {http://arxiv.org/abs/1008.3745} {arXiv:1008.3745} \BibitemShut {NoStop}%
\bibitem [{\citenamefont {Schuch}\ \emph {et~al.}(2011)\citenamefont {Schuch},
  \citenamefont {P{\'{e}}rez-Garc{\'{i}}a},\ and\ \citenamefont
  {Cirac}}]{Schuch2011a}%
  \BibitemOpen
  \bibfield  {author} {\bibinfo {author} {\bibfnamefont {Norbert}\ \bibnamefont
  {Schuch}}, \bibinfo {author} {\bibfnamefont {David}\ \bibnamefont
  {P{\'{e}}rez-Garc{\'{i}}a}}, \ and\ \bibinfo {author} {\bibfnamefont
  {Ignacio}\ \bibnamefont {Cirac}},\ }\bibfield  {title} {\enquote {\bibinfo
  {title} {{Classifying quantum phases using matrix product states and
  projected entangled pair states}},}\ }\href {\doibase
  10.1103/PhysRevB.84.165139} {\bibfield  {journal} {\bibinfo  {journal} {Phys.
  Rev. B}\ }\textbf {\bibinfo {volume} {84}},\ \bibinfo {pages} {165139}
  (\bibinfo {year} {2011})},\ \Eprint {http://arxiv.org/abs/1010.3732}
  {arXiv:1010.3732} \BibitemShut {NoStop}%
\bibitem [{\citenamefont {Fidkowski}\ and\ \citenamefont
  {Kitaev}(2011)}]{Fidkowski2011}%
  \BibitemOpen
  \bibfield  {author} {\bibinfo {author} {\bibfnamefont {Lukasz}\ \bibnamefont
  {Fidkowski}}\ and\ \bibinfo {author} {\bibfnamefont {Alexei}\ \bibnamefont
  {Kitaev}},\ }\bibfield  {title} {\enquote {\bibinfo {title} {{Topological
  phases of fermions in one dimension}},}\ }\href {\doibase
  10.1103/PhysRevB.83.075103} {\bibfield  {journal} {\bibinfo  {journal} {Phys.
  Rev. B}\ }\textbf {\bibinfo {volume} {83}},\ \bibinfo {pages} {075103}
  (\bibinfo {year} {2011})},\ \Eprint {http://arxiv.org/abs/1008.4138}
  {arXiv:1008.4138} \BibitemShut {NoStop}%
\bibitem [{\citenamefont {Chen}\ \emph
  {et~al.}(2011{\natexlab{b}})\citenamefont {Chen}, \citenamefont {Liu},\ and\
  \citenamefont {Wen}}]{Chen2011b}%
  \BibitemOpen
  \bibfield  {author} {\bibinfo {author} {\bibfnamefont {Xie}\ \bibnamefont
  {Chen}}, \bibinfo {author} {\bibfnamefont {Zheng-Xin}\ \bibnamefont {Liu}}, \
  and\ \bibinfo {author} {\bibfnamefont {Xiao-Gang}\ \bibnamefont {Wen}},\
  }\bibfield  {title} {\enquote {\bibinfo {title} {{Two-dimensional
  symmetry-protected topological orders and their protected gapless edge
  excitations}},}\ }\href {\doibase 10.1103/PhysRevB.84.235141} {\bibfield
  {journal} {\bibinfo  {journal} {Phys. Rev. B}\ }\textbf {\bibinfo {volume}
  {84}},\ \bibinfo {pages} {235141} (\bibinfo {year} {2011}{\natexlab{b}})},\
  \Eprint {http://arxiv.org/abs/1106.4772} {arXiv:1106.4772} \BibitemShut
  {NoStop}%
\bibitem [{\citenamefont {Chen}\ \emph {et~al.}(2013)\citenamefont {Chen},
  \citenamefont {Gu}, \citenamefont {Liu},\ and\ \citenamefont
  {Wen}}]{Chen2013}%
  \BibitemOpen
  \bibfield  {author} {\bibinfo {author} {\bibfnamefont {Xie}\ \bibnamefont
  {Chen}}, \bibinfo {author} {\bibfnamefont {Zheng-Cheng}\ \bibnamefont {Gu}},
  \bibinfo {author} {\bibfnamefont {Zheng-Xin}\ \bibnamefont {Liu}}, \ and\
  \bibinfo {author} {\bibfnamefont {Xiao-Gang}\ \bibnamefont {Wen}},\
  }\bibfield  {title} {\enquote {\bibinfo {title} {{Symmetry protected
  topological orders and the group cohomology of their symmetry group}},}\
  }\href {\doibase 10.1103/PhysRevB.87.155114} {\bibfield  {journal} {\bibinfo
  {journal} {Phys. Rev. B}\ }\textbf {\bibinfo {volume} {87}},\ \bibinfo
  {pages} {155114} (\bibinfo {year} {2013})},\ \Eprint
  {http://arxiv.org/abs/1106.4772} {arXiv:1106.4772} \BibitemShut {NoStop}%
\bibitem [{\citenamefont {Levin}\ and\ \citenamefont {Gu}(2012)}]{Levin2012}%
  \BibitemOpen
  \bibfield  {author} {\bibinfo {author} {\bibfnamefont {Michael}\ \bibnamefont
  {Levin}}\ and\ \bibinfo {author} {\bibfnamefont {Zheng-Cheng}\ \bibnamefont
  {Gu}},\ }\bibfield  {title} {\enquote {\bibinfo {title} {{Braiding statistics
  approach to symmetry-protected topological phases}},}\ }\href {\doibase
  10.1103/PhysRevB.86.115109} {\bibfield  {journal} {\bibinfo  {journal} {Phys.
  Rev. B}\ }\textbf {\bibinfo {volume} {86}},\ \bibinfo {pages} {115109}
  (\bibinfo {year} {2012})},\ \Eprint {http://arxiv.org/abs/1202.3120}
  {arXiv:1202.3120} \BibitemShut {NoStop}%
\bibitem [{\citenamefont {Vishwanath}\ and\ \citenamefont
  {Senthil}(2013)}]{Vishwanath2013}%
  \BibitemOpen
  \bibfield  {author} {\bibinfo {author} {\bibfnamefont {Ashvin}\ \bibnamefont
  {Vishwanath}}\ and\ \bibinfo {author} {\bibfnamefont {T.}~\bibnamefont
  {Senthil}},\ }\bibfield  {title} {\enquote {\bibinfo {title} {{Physics of
  Three-Dimensional Bosonic Topological Insulators: Surface-Deconfined
  Criticality and Quantized Magnetoelectric Effect}},}\ }\href {\doibase
  10.1103/PhysRevX.3.011016} {\bibfield  {journal} {\bibinfo  {journal} {Phys.
  Rev. X}\ }\textbf {\bibinfo {volume} {3}},\ \bibinfo {pages} {011016}
  (\bibinfo {year} {2013})},\ \Eprint {http://arxiv.org/abs/1209.3058}
  {arXiv:1209.3058} \BibitemShut {NoStop}%
\bibitem [{\citenamefont {Wang}\ \emph {et~al.}(2014)\citenamefont {Wang},
  \citenamefont {Potter},\ and\ \citenamefont {Senthil}}]{Wang2014}%
  \BibitemOpen
  \bibfield  {author} {\bibinfo {author} {\bibfnamefont {Chong}\ \bibnamefont
  {Wang}}, \bibinfo {author} {\bibfnamefont {Andrew~C}\ \bibnamefont {Potter}},
  \ and\ \bibinfo {author} {\bibfnamefont {T}~\bibnamefont {Senthil}},\
  }\bibfield  {title} {\enquote {\bibinfo {title} {{Classification of
  interacting electronic topological insulators in three dimensions.}}}\ }\href
  {\doibase 10.1126/science.1243326} {\bibfield  {journal} {\bibinfo  {journal}
  {Science}\ }\textbf {\bibinfo {volume} {343}},\ \bibinfo {pages} {629--31}
  (\bibinfo {year} {2014})},\ \Eprint {http://arxiv.org/abs/1306.3238}
  {arXiv:1306.3238} \BibitemShut {NoStop}%
\bibitem [{\citenamefont {Kapustin}(2014)}]{Kapustin2014}%
  \BibitemOpen
  \bibfield  {author} {\bibinfo {author} {\bibfnamefont {Anton}\ \bibnamefont
  {Kapustin}},\ }\href {http://arxiv.org/abs/1403.1467} {\enquote {\bibinfo
  {title} {{Symmetry Protected Topological Phases, Anomalies, and Cobordisms:
  Beyond Group Cohomology}},}\ } (\bibinfo {year} {2014}),\ \Eprint
  {http://arxiv.org/abs/1403.1467} {arXiv:1403.1467} \BibitemShut {NoStop}%
\bibitem [{\citenamefont {Gu}\ and\ \citenamefont {Wen}(2014)}]{Gu2014}%
  \BibitemOpen
  \bibfield  {author} {\bibinfo {author} {\bibfnamefont {Zheng-Cheng}\
  \bibnamefont {Gu}}\ and\ \bibinfo {author} {\bibfnamefont {Xiao-Gang}\
  \bibnamefont {Wen}},\ }\bibfield  {title} {\enquote {\bibinfo {title}
  {{Symmetry-protected topological orders for interacting fermions: Fermionic
  topological nonlinear $\sigma$ models and a special group supercohomology
  theory}},}\ }\href {\doibase 10.1103/PhysRevB.90.115141} {\bibfield
  {journal} {\bibinfo  {journal} {Phys. Rev. B}\ }\textbf {\bibinfo {volume}
  {90}},\ \bibinfo {pages} {115141} (\bibinfo {year} {2014})},\ \Eprint
  {http://arxiv.org/abs/1201.2648} {arXiv:1201.2648} \BibitemShut {NoStop}%
\bibitem [{\citenamefont {Else}\ and\ \citenamefont {Nayak}(2014)}]{Else2014}%
  \BibitemOpen
  \bibfield  {author} {\bibinfo {author} {\bibfnamefont {Dominic~V.}\
  \bibnamefont {Else}}\ and\ \bibinfo {author} {\bibfnamefont {Chetan}\
  \bibnamefont {Nayak}},\ }\bibfield  {title} {\enquote {\bibinfo {title}
  {{Classifying symmetry-protected topological phases through the anomalous
  action of the symmetry on the edge}},}\ }\href {\doibase
  10.1103/PhysRevB.90.235137} {\bibfield  {journal} {\bibinfo  {journal} {Phys.
  Rev. B}\ }\textbf {\bibinfo {volume} {90}},\ \bibinfo {pages} {235137}
  (\bibinfo {year} {2014})},\ \Eprint {http://arxiv.org/abs/1409.5436}
  {arXiv:1409.5436} \BibitemShut {NoStop}%
\bibitem [{\citenamefont {Burnell}\ \emph {et~al.}(2014)\citenamefont
  {Burnell}, \citenamefont {Chen}, \citenamefont {Fidkowski},\ and\
  \citenamefont {Vishwanath}}]{Burnell2014}%
  \BibitemOpen
  \bibfield  {author} {\bibinfo {author} {\bibfnamefont {F.~J.}\ \bibnamefont
  {Burnell}}, \bibinfo {author} {\bibfnamefont {Xie}\ \bibnamefont {Chen}},
  \bibinfo {author} {\bibfnamefont {Lukasz}\ \bibnamefont {Fidkowski}}, \ and\
  \bibinfo {author} {\bibfnamefont {Ashvin}\ \bibnamefont {Vishwanath}},\
  }\bibfield  {title} {\enquote {\bibinfo {title} {{Exactly soluble model of a
  three-dimensional symmetry-protected topological phase of bosons with surface
  topological order}},}\ }\href {\doibase 10.1103/PhysRevB.90.245122}
  {\bibfield  {journal} {\bibinfo  {journal} {Phys. Rev. B}\ }\textbf {\bibinfo
  {volume} {90}},\ \bibinfo {pages} {245122} (\bibinfo {year} {2014})},\
  \Eprint {http://arxiv.org/abs/1302.7072} {arXiv:1302.7072} \BibitemShut
  {NoStop}%
\bibitem [{\citenamefont {Wang}\ \emph {et~al.}(2015)\citenamefont {Wang},
  \citenamefont {Gu},\ and\ \citenamefont {Wen}}]{Wang2015}%
  \BibitemOpen
  \bibfield  {author} {\bibinfo {author} {\bibfnamefont {Juven~C.}\
  \bibnamefont {Wang}}, \bibinfo {author} {\bibfnamefont {Zheng-Cheng}\
  \bibnamefont {Gu}}, \ and\ \bibinfo {author} {\bibfnamefont {Xiao-Gang}\
  \bibnamefont {Wen}},\ }\bibfield  {title} {\enquote {\bibinfo {title}
  {{Field-Theory Representation of Gauge-Gravity Symmetry-Protected Topological
  Invariants, Group Cohomology, and Beyond}},}\ }\href {\doibase
  10.1103/PhysRevLett.114.031601} {\bibfield  {journal} {\bibinfo  {journal}
  {Phys. Rev. Lett.}\ }\textbf {\bibinfo {volume} {114}},\ \bibinfo {pages}
  {031601} (\bibinfo {year} {2015})},\ \Eprint {http://arxiv.org/abs/1405.7689}
  {arXiv:1405.7689} \BibitemShut {NoStop}%
\bibitem [{\citenamefont {Cheng}\ \emph {et~al.}(2015)\citenamefont {Cheng},
  \citenamefont {Bi}, \citenamefont {You},\ and\ \citenamefont
  {Gu}}]{Cheng2015}%
  \BibitemOpen
  \bibfield  {author} {\bibinfo {author} {\bibfnamefont {Meng}\ \bibnamefont
  {Cheng}}, \bibinfo {author} {\bibfnamefont {Zhen}\ \bibnamefont {Bi}},
  \bibinfo {author} {\bibfnamefont {Yi-Zhuang}\ \bibnamefont {You}}, \ and\
  \bibinfo {author} {\bibfnamefont {Zheng-Cheng}\ \bibnamefont {Gu}},\ }\href
  {http://arxiv.org/abs/1501.01313v1$\backslash$nhttp://arxiv.org/abs/1501.01313}
  {\enquote {\bibinfo {title} {{Towards a Complete Classification of
  Symmetry-Protected Phases for Interacting Fermions in Two Dimensions}},}\ }
  (\bibinfo {year} {2015}),\ \Eprint {http://arxiv.org/abs/1501.01313}
  {arXiv:1501.01313} \BibitemShut {NoStop}%
\bibitem [{\citenamefont {Hasan}\ and\ \citenamefont {Kane}(2010)}]{Hasan10}%
  \BibitemOpen
  \bibfield  {author} {\bibinfo {author} {\bibfnamefont {M.~Z.}\ \bibnamefont
  {Hasan}}\ and\ \bibinfo {author} {\bibfnamefont {C.~L.}\ \bibnamefont
  {Kane}},\ }\bibfield  {title} {\enquote {\bibinfo {title}
  {\textit{Colloquium}: Topological insulators},}\ }\href {\doibase
  10.1103/RevModPhys.82.3045} {\bibfield  {journal} {\bibinfo  {journal} {Rev.
  Mod. Phys.}\ }\textbf {\bibinfo {volume} {82}},\ \bibinfo {pages}
  {3045--3067} (\bibinfo {year} {2010})}\BibitemShut {NoStop}%
\bibitem [{\citenamefont {Qi}\ and\ \citenamefont {Zhang}(2011)}]{Qi11}%
  \BibitemOpen
  \bibfield  {author} {\bibinfo {author} {\bibfnamefont {Xiao-Liang}\
  \bibnamefont {Qi}}\ and\ \bibinfo {author} {\bibfnamefont {Shou-Cheng}\
  \bibnamefont {Zhang}},\ }\bibfield  {title} {\enquote {\bibinfo {title}
  {Topological insulators and superconductors},}\ }\href {\doibase
  10.1103/RevModPhys.83.1057} {\bibfield  {journal} {\bibinfo  {journal} {Rev.
  Mod. Phys.}\ }\textbf {\bibinfo {volume} {83}},\ \bibinfo {pages}
  {1057--1110} (\bibinfo {year} {2011})}\BibitemShut {NoStop}%
\bibitem [{\citenamefont {Maciejko}\ \emph {et~al.}(2010)\citenamefont
  {Maciejko}, \citenamefont {Qi}, \citenamefont {Karch},\ and\ \citenamefont
  {Zhang}}]{Maciejko2010}%
  \BibitemOpen
  \bibfield  {author} {\bibinfo {author} {\bibfnamefont {Joseph}\ \bibnamefont
  {Maciejko}}, \bibinfo {author} {\bibfnamefont {Xiao-Liang}\ \bibnamefont
  {Qi}}, \bibinfo {author} {\bibfnamefont {Andreas}\ \bibnamefont {Karch}}, \
  and\ \bibinfo {author} {\bibfnamefont {Shou-Cheng}\ \bibnamefont {Zhang}},\
  }\bibfield  {title} {\enquote {\bibinfo {title} {{Fractional Topological
  Insulators in Three Dimensions}},}\ }\href {\doibase
  10.1103/PhysRevLett.105.246809} {\bibfield  {journal} {\bibinfo  {journal}
  {Phys. Rev. Lett.}\ }\textbf {\bibinfo {volume} {105}},\ \bibinfo {pages}
  {246809} (\bibinfo {year} {2010})},\ \Eprint {http://arxiv.org/abs/1004.3628}
  {arXiv:1004.3628} \BibitemShut {NoStop}%
\bibitem [{\citenamefont {Essin}\ and\ \citenamefont
  {Hermele}(2013)}]{Essin2013}%
  \BibitemOpen
  \bibfield  {author} {\bibinfo {author} {\bibfnamefont {Andrew~M.}\
  \bibnamefont {Essin}}\ and\ \bibinfo {author} {\bibfnamefont {Michael}\
  \bibnamefont {Hermele}},\ }\bibfield  {title} {\enquote {\bibinfo {title}
  {{Classifying fractionalization: Symmetry classification of gapped
  $\mathbb{Z}_2$ spin liquids in two dimensions}},}\ }\href {\doibase
  10.1103/PhysRevB.87.104406} {\bibfield  {journal} {\bibinfo  {journal} {Phys.
  Rev. B}\ }\textbf {\bibinfo {volume} {87}},\ \bibinfo {pages} {104406}
  (\bibinfo {year} {2013})},\ \Eprint {http://arxiv.org/abs/1212.0593}
  {arXiv:1212.0593} \BibitemShut {NoStop}%
\bibitem [{\citenamefont {Lu}\ and\ \citenamefont {Vishwanath}(2016)}]{Lu2013}%
  \BibitemOpen
  \bibfield  {author} {\bibinfo {author} {\bibfnamefont {Yuan-Ming}\
  \bibnamefont {Lu}}\ and\ \bibinfo {author} {\bibfnamefont {Ashvin}\
  \bibnamefont {Vishwanath}},\ }\bibfield  {title} {\enquote {\bibinfo {title}
  {{Classification and Properties of Symmetry Enriched Topological Phases: A
  Chern-Simons approach with applications to Z2 spin liquids}},}\ }\href
  {\doibase 10.1103/PhysRevB.93.155121} {\bibfield  {journal} {\bibinfo
  {journal} {Phys. Rev. B}\ }\textbf {\bibinfo {volume} {93}},\ \bibinfo
  {pages} {5} (\bibinfo {year} {2016})},\ \Eprint
  {http://arxiv.org/abs/1302.2634} {arXiv:1302.2634} \BibitemShut {NoStop}%
\bibitem [{\citenamefont {Mesaros}\ and\ \citenamefont
  {Ran}(2013)}]{Mesaros2013}%
  \BibitemOpen
  \bibfield  {author} {\bibinfo {author} {\bibfnamefont {Andrej}\ \bibnamefont
  {Mesaros}}\ and\ \bibinfo {author} {\bibfnamefont {Ying}\ \bibnamefont
  {Ran}},\ }\bibfield  {title} {\enquote {\bibinfo {title} {{Classification of
  symmetry enriched topological phases with exactly solvable models}},}\ }\href
  {\doibase 10.1103/PhysRevB.87.155115} {\bibfield  {journal} {\bibinfo
  {journal} {Phys. Rev. B}\ }\textbf {\bibinfo {volume} {87}},\ \bibinfo
  {pages} {155115} (\bibinfo {year} {2013})},\ \Eprint
  {http://arxiv.org/abs/1212.0835} {arXiv:1212.0835} \BibitemShut {NoStop}%
\bibitem [{\citenamefont {Hung}\ and\ \citenamefont {Wen}(2013)}]{Hung2013}%
  \BibitemOpen
  \bibfield  {author} {\bibinfo {author} {\bibfnamefont {Ling-Yan}\
  \bibnamefont {Hung}}\ and\ \bibinfo {author} {\bibfnamefont {Xiao-Gang}\
  \bibnamefont {Wen}},\ }\bibfield  {title} {\enquote {\bibinfo {title}
  {{Quantized topological terms in weak-coupling gauge theories with a global
  symmetry and their connection to symmetry-enriched topological phases}},}\
  }\href {\doibase 10.1103/PhysRevB.87.165107} {\bibfield  {journal} {\bibinfo
  {journal} {Phys. Rev. B}\ }\textbf {\bibinfo {volume} {87}},\ \bibinfo
  {pages} {165107} (\bibinfo {year} {2013})},\ \Eprint
  {http://arxiv.org/abs/1212.1827} {arXiv:1212.1827} \BibitemShut {NoStop}%
\bibitem [{\citenamefont {Barkeshli}\ \emph {et~al.}(2014)\citenamefont
  {Barkeshli}, \citenamefont {Bonderson}, \citenamefont {Cheng},\ and\
  \citenamefont {Wang}}]{Barkeshli2014}%
  \BibitemOpen
  \bibfield  {author} {\bibinfo {author} {\bibfnamefont {Maissam}\ \bibnamefont
  {Barkeshli}}, \bibinfo {author} {\bibfnamefont {Parsa}\ \bibnamefont
  {Bonderson}}, \bibinfo {author} {\bibfnamefont {Meng}\ \bibnamefont {Cheng}},
  \ and\ \bibinfo {author} {\bibfnamefont {Zhenghan}\ \bibnamefont {Wang}},\
  }\href {http://arxiv.org/abs/1410.4540} {\enquote {\bibinfo {title}
  {{Symmetry, Defects, and Gauging of Topological Phases}},}\ } (\bibinfo
  {year} {2014}),\ \Eprint {http://arxiv.org/abs/1410.4540} {arXiv:1410.4540}
  \BibitemShut {NoStop}%
\bibitem [{\citenamefont {Cheng}\ \emph {et~al.}(2016)\citenamefont {Cheng},
  \citenamefont {Zaletel}, \citenamefont {Barkeshli}, \citenamefont
  {Vishwanath},\ and\ \citenamefont {Bonderson}}]{Cheng2015a}%
  \BibitemOpen
  \bibfield  {author} {\bibinfo {author} {\bibfnamefont {Meng}\ \bibnamefont
  {Cheng}}, \bibinfo {author} {\bibfnamefont {Michael}\ \bibnamefont
  {Zaletel}}, \bibinfo {author} {\bibfnamefont {Maissam}\ \bibnamefont
  {Barkeshli}}, \bibinfo {author} {\bibfnamefont {Ashvin}\ \bibnamefont
  {Vishwanath}}, \ and\ \bibinfo {author} {\bibfnamefont {Parsa}\ \bibnamefont
  {Bonderson}},\ }\bibfield  {title} {\enquote {\bibinfo {title}
  {{Translational symmetry and microscopic constraints on symmetry-enriched
  topological phases: a view from the surface}},}\ }\href {\doibase
  10.1103/PhysRevX.6.041068} {\bibfield  {journal} {\bibinfo  {journal} {Phys.
  Rev. X}\ }\textbf {\bibinfo {volume} {6}} (\bibinfo {year} {2016}),\
  10.1103/PhysRevX.6.041068},\ \Eprint {http://arxiv.org/abs/1511.02263}
  {arXiv:1511.02263} \BibitemShut {NoStop}%
\bibitem [{\citenamefont {{Khemani}}\ \emph {et~al.}(2016)\citenamefont
  {{Khemani}}, \citenamefont {{Lazarides}}, \citenamefont {{Moessner}},\ and\
  \citenamefont {{Sondhi}}}]{Khemani15b}%
  \BibitemOpen
  \bibfield  {author} {\bibinfo {author} {\bibfnamefont {V.}~\bibnamefont
  {{Khemani}}}, \bibinfo {author} {\bibfnamefont {A.}~\bibnamefont
  {{Lazarides}}}, \bibinfo {author} {\bibfnamefont {R.}~\bibnamefont
  {{Moessner}}}, \ and\ \bibinfo {author} {\bibfnamefont {S.~L.}\ \bibnamefont
  {{Sondhi}}},\ }\bibfield  {title} {\enquote {\bibinfo {title} {{Phase
  structure of driven quantum systems}},}\ }\href@noop {} {\bibfield  {journal}
  {\bibinfo  {journal} {Phys. Rev. Lett.}\ }\textbf {\bibinfo {volume} {116}},\
  \bibinfo {pages} {250401} (\bibinfo {year} {2016})},\ \Eprint
  {http://arxiv.org/abs/1508.03344} {arXiv:1508.03344} \BibitemShut {NoStop}%
\bibitem [{\citenamefont {{von Keyserlingk}}\ and\ \citenamefont
  {{Sondhi}}(2016)}]{vonKeyserlingk16a}%
  \BibitemOpen
  \bibfield  {author} {\bibinfo {author} {\bibfnamefont {C.~W.}\ \bibnamefont
  {{von Keyserlingk}}}\ and\ \bibinfo {author} {\bibfnamefont {S.~L.}\
  \bibnamefont {{Sondhi}}},\ }\bibfield  {title} {\enquote {\bibinfo {title}
  {{Phase structure of one-dimensional interacting Floquet systems. I. Abelian
  symmetry-protected topological phases}},}\ }\href {\doibase
  10.1103/PhysRevB.93.245145} {\bibfield  {journal} {\bibinfo  {journal} {Phys.
  Rev. B}\ }\textbf {\bibinfo {volume} {93}},\ \bibinfo {pages} {245145}
  (\bibinfo {year} {2016})},\ \Eprint {http://arxiv.org/abs/1602.02157}
  {arXiv:1602.02157} \BibitemShut {NoStop}%
\bibitem [{\citenamefont {{Else}}\ and\ \citenamefont
  {{Nayak}}(2016)}]{Else16a}%
  \BibitemOpen
  \bibfield  {author} {\bibinfo {author} {\bibfnamefont {D.~V.}\ \bibnamefont
  {{Else}}}\ and\ \bibinfo {author} {\bibfnamefont {C.}~\bibnamefont
  {{Nayak}}},\ }\href {\doibase 10.1103/PhysRevB.93.201103} {\enquote {\bibinfo
  {title} {{Classification of topological phases in periodically driven
  interacting systems}},}\ } (\bibinfo {year} {2016}),\ \Eprint
  {http://arxiv.org/abs/1602.04804} {arXiv:1602.04804} \BibitemShut {NoStop}%
\bibitem [{\citenamefont {{Potter}}\ \emph {et~al.}()\citenamefont {{Potter}},
  \citenamefont {{Morimoto}},\ and\ \citenamefont {{Vishwanath}}}]{Potter16}%
  \BibitemOpen
  \bibfield  {author} {\bibinfo {author} {\bibfnamefont {A.~C.}\ \bibnamefont
  {{Potter}}}, \bibinfo {author} {\bibfnamefont {T.}~\bibnamefont
  {{Morimoto}}}, \ and\ \bibinfo {author} {\bibfnamefont {A.}~\bibnamefont
  {{Vishwanath}}},\ }\bibfield  {title} {\enquote {\bibinfo {title}
  {{Topological classification of interacting 1D Floquet phases}},}\ }\href
  {\doibase 10.1103/PhysRevX.6.041001} {\bibfield  {journal} {\bibinfo
  {journal} {Phys. Rev. X}\ }\textbf {\bibinfo {volume} {6}},\ \bibinfo {pages}
  {041001}},\ \Eprint {http://arxiv.org/abs/1602.05194} {arXiv:1602.05194}
  \BibitemShut {NoStop}%
\bibitem [{\citenamefont {{Roy}}\ and\ \citenamefont {{Harper}}()}]{Roy16}%
  \BibitemOpen
  \bibfield  {author} {\bibinfo {author} {\bibfnamefont {R.}~\bibnamefont
  {{Roy}}}\ and\ \bibinfo {author} {\bibfnamefont {F.}~\bibnamefont
  {{Harper}}},\ }\bibfield  {title} {\enquote {\bibinfo {title} {{ Abelian
  Floquet symmetry-protected topological phases in one dimension}},}\ }\href
  {\doibase 10.1103/PhysRevB.94.125105} {\bibfield  {journal} {\bibinfo
  {journal} {Phys. Rev. B}\ }\textbf {\bibinfo {volume} {94}},\ \bibinfo
  {pages} {125105}},\ \Eprint {http://arxiv.org/abs/1602.08089}
  {arXiv:1602.08089} \BibitemShut {NoStop}%
\bibitem [{\citenamefont {{Wilczek}}(2012)}]{Wilczek12}%
  \BibitemOpen
  \bibfield  {author} {\bibinfo {author} {\bibfnamefont {F.}~\bibnamefont
  {{Wilczek}}},\ }\bibfield  {title} {\enquote {\bibinfo {title} {{Quantum Time
  Crystals}},}\ }\href {\doibase 10.1103/PhysRevLett.109.160401} {\bibfield
  {journal} {\bibinfo  {journal} {Phys. Rev. Lett.}\ }\textbf {\bibinfo
  {volume} {109}},\ \bibinfo {eid} {160401} (\bibinfo {year} {2012})},\ \Eprint
  {http://arxiv.org/abs/1202.2539} {arXiv:1202.2539 [quant-ph]} \BibitemShut
  {NoStop}%
\bibitem [{\citenamefont {{Shapere}}\ and\ \citenamefont
  {{Wilczek}}(2012)}]{Shapere12}%
  \BibitemOpen
  \bibfield  {author} {\bibinfo {author} {\bibfnamefont {A.}~\bibnamefont
  {{Shapere}}}\ and\ \bibinfo {author} {\bibfnamefont {F.}~\bibnamefont
  {{Wilczek}}},\ }\bibfield  {title} {\enquote {\bibinfo {title} {{Classical
  Time Crystals}},}\ }\href {\doibase 10.1103/PhysRevLett.109.160402}
  {\bibfield  {journal} {\bibinfo  {journal} {Phys. Rev. Lett.}\ }\textbf
  {\bibinfo {volume} {109}},\ \bibinfo {eid} {160402} (\bibinfo {year}
  {2012})},\ \Eprint {http://arxiv.org/abs/1202.2537} {arXiv:1202.2537
  [cond-mat.other]} \BibitemShut {NoStop}%
\bibitem [{\citenamefont {Li}\ \emph {et~al.}(2012)\citenamefont {Li},
  \citenamefont {Gong}, \citenamefont {Yin}, \citenamefont {Quan},
  \citenamefont {Yin}, \citenamefont {Zhang}, \citenamefont {Duan},\ and\
  \citenamefont {Zhang}}]{Li13}%
  \BibitemOpen
  \bibfield  {author} {\bibinfo {author} {\bibfnamefont {Tongcang}\
  \bibnamefont {Li}}, \bibinfo {author} {\bibfnamefont {Zhe-Xuan}\ \bibnamefont
  {Gong}}, \bibinfo {author} {\bibfnamefont {Zhang-Qi}\ \bibnamefont {Yin}},
  \bibinfo {author} {\bibfnamefont {H.~T.}\ \bibnamefont {Quan}}, \bibinfo
  {author} {\bibfnamefont {Xiaobo}\ \bibnamefont {Yin}}, \bibinfo {author}
  {\bibfnamefont {Peng}\ \bibnamefont {Zhang}}, \bibinfo {author}
  {\bibfnamefont {L.-M.}\ \bibnamefont {Duan}}, \ and\ \bibinfo {author}
  {\bibfnamefont {Xiang}\ \bibnamefont {Zhang}},\ }\bibfield  {title} {\enquote
  {\bibinfo {title} {Space-time crystals of trapped ions},}\ }\href {\doibase
  10.1103/PhysRevLett.109.163001} {\bibfield  {journal} {\bibinfo  {journal}
  {Phys. Rev. Lett.}\ }\textbf {\bibinfo {volume} {109}},\ \bibinfo {pages}
  {163001} (\bibinfo {year} {2012})}\BibitemShut {NoStop}%
\bibitem [{\citenamefont {Bruno}(2013{\natexlab{a}})}]{Bruno13a}%
  \BibitemOpen
  \bibfield  {author} {\bibinfo {author} {\bibfnamefont {Patrick}\ \bibnamefont
  {Bruno}},\ }\bibfield  {title} {\enquote {\bibinfo {title} {Comment on
  ``quantum time crystals''},}\ }\href {\doibase
  10.1103/PhysRevLett.110.118901} {\bibfield  {journal} {\bibinfo  {journal}
  {Phys. Rev. Lett.}\ }\textbf {\bibinfo {volume} {110}},\ \bibinfo {pages}
  {118901} (\bibinfo {year} {2013}{\natexlab{a}})}\BibitemShut {NoStop}%
\bibitem [{\citenamefont {Bruno}(2013{\natexlab{b}})}]{Bruno13b}%
  \BibitemOpen
  \bibfield  {author} {\bibinfo {author} {\bibfnamefont {Patrick}\ \bibnamefont
  {Bruno}},\ }\bibfield  {title} {\enquote {\bibinfo {title} {Comment on
  ``space-time crystals of trapped ions''},}\ }\href {\doibase
  10.1103/PhysRevLett.111.029301} {\bibfield  {journal} {\bibinfo  {journal}
  {Phys. Rev. Lett.}\ }\textbf {\bibinfo {volume} {111}},\ \bibinfo {pages}
  {029301} (\bibinfo {year} {2013}{\natexlab{b}})}\BibitemShut {NoStop}%
\bibitem [{\citenamefont {Bruno}(2013{\natexlab{c}})}]{Bruno2013}%
  \BibitemOpen
  \bibfield  {author} {\bibinfo {author} {\bibfnamefont {Patrick}\ \bibnamefont
  {Bruno}},\ }\bibfield  {title} {\enquote {\bibinfo {title} {{Impossibility of
  Spontaneously Rotating Time Crystals: A No-Go Theorem}},}\ }\href {\doibase
  10.1103/PhysRevLett.111.070402} {\bibfield  {journal} {\bibinfo  {journal}
  {Phys. Rev. Lett.}\ }\textbf {\bibinfo {volume} {111}},\ \bibinfo {pages}
  {070402} (\bibinfo {year} {2013}{\natexlab{c}})},\ \Eprint
  {http://arxiv.org/abs/1306.6275} {arXiv:1306.6275} \BibitemShut {NoStop}%
\bibitem [{\citenamefont {Nozi\'eres}(2013)}]{Nozieres13}%
  \BibitemOpen
  \bibfield  {author} {\bibinfo {author} {\bibfnamefont {Philippe}\
  \bibnamefont {Nozi\'eres}},\ }\bibfield  {title} {\enquote {\bibinfo {title}
  {Time crystals: Can diamagnetic currents drive a charge density wave into
  rotation?}}\ }\href {http://stacks.iop.org/0295-5075/103/i=5/a=57008}
  {\bibfield  {journal} {\bibinfo  {journal} {Eur. Phys. Lett.}\ }\textbf
  {\bibinfo {volume} {103}},\ \bibinfo {pages} {57008} (\bibinfo {year}
  {2013})}\BibitemShut {NoStop}%
\bibitem [{\citenamefont {Volovik}(2013)}]{Volovik2013}%
  \BibitemOpen
  \bibfield  {author} {\bibinfo {author} {\bibfnamefont {G.~E.}\ \bibnamefont
  {Volovik}},\ }\bibfield  {title} {\enquote {\bibinfo {title} {{On the broken
  time translation symmetry in macroscopic systems: Precessing states and
  off-diagonal long-range order}},}\ }\href {\doibase
  10.1134/S0021364013210133} {\bibfield  {journal} {\bibinfo  {journal} {JETP
  Lett.}\ }\textbf {\bibinfo {volume} {98}},\ \bibinfo {pages} {491--495}
  (\bibinfo {year} {2013})},\ \Eprint {http://arxiv.org/abs/1309.1845}
  {arXiv:1309.1845} \BibitemShut {NoStop}%
\bibitem [{\citenamefont {{Watanabe}}\ and\ \citenamefont
  {{Oshikawa}}(2015)}]{Watanabe15}%
  \BibitemOpen
  \bibfield  {author} {\bibinfo {author} {\bibfnamefont {H.}~\bibnamefont
  {{Watanabe}}}\ and\ \bibinfo {author} {\bibfnamefont {M.}~\bibnamefont
  {{Oshikawa}}},\ }\bibfield  {title} {\enquote {\bibinfo {title} {{Absence of
  Quantum Time Crystals}},}\ }\href {\doibase 10.1103/PhysRevLett.114.251603}
  {\bibfield  {journal} {\bibinfo  {journal} {Phys. Rev. Lett.}\ }\textbf
  {\bibinfo {volume} {114}},\ \bibinfo {eid} {251603} (\bibinfo {year}
  {2015})},\ \Eprint {http://arxiv.org/abs/1410.2143} {arXiv:1410.2143
  [cond-mat.stat-mech]} \BibitemShut {NoStop}%
\bibitem [{\citenamefont {Else}\ \emph {et~al.}(2016)\citenamefont {Else},
  \citenamefont {Bauer},\ and\ \citenamefont {Nayak}}]{Else2016b}%
  \BibitemOpen
  \bibfield  {author} {\bibinfo {author} {\bibfnamefont {Dominic~V.}\
  \bibnamefont {Else}}, \bibinfo {author} {\bibfnamefont {Bela}\ \bibnamefont
  {Bauer}}, \ and\ \bibinfo {author} {\bibfnamefont {Chetan}\ \bibnamefont
  {Nayak}},\ }\bibfield  {title} {\enquote {\bibinfo {title} {{Floquet Time
  Crystals}},}\ }\href {http://arxiv.org/abs/1603.08001} {\bibfield  {journal}
  {\bibinfo  {journal} {Phys. Rev. Lett.}\ }\textbf {\bibinfo {volume} {117}},\
  \bibinfo {pages} {090402} (\bibinfo {year} {2016})},\ \Eprint
  {http://arxiv.org/abs/1603.08001} {arXiv:1603.08001} \BibitemShut {NoStop}%
\bibitem [{\citenamefont {{von Keyserlingk}}\ and\ \citenamefont
  {{Sondhi}}()}]{vonKeyserlingk16b}%
  \BibitemOpen
  \bibfield  {author} {\bibinfo {author} {\bibfnamefont {C.~W.}\ \bibnamefont
  {{von Keyserlingk}}}\ and\ \bibinfo {author} {\bibfnamefont {S.~L.}\
  \bibnamefont {{Sondhi}}},\ }\bibfield  {title} {\enquote {\bibinfo {title}
  {Phase structure of one-dimensional interacting floquet systems. ii.
  symmetry-broken phases},}\ }\href {\doibase 10.1103/PhysRevB.93.245146}
  {\bibfield  {journal} {\bibinfo  {journal} {Phys. Rev. B}\ }\textbf {\bibinfo
  {volume} {93}},\ \bibinfo {pages} {245146}},\ \Eprint
  {http://arxiv.org/abs/1602.06949} {arXiv:1602.06949} \BibitemShut {NoStop}%
\bibitem [{\citenamefont {von Keyserlingk}\ \emph {et~al.}(2016)\citenamefont
  {von Keyserlingk}, \citenamefont {Khemani},\ and\ \citenamefont
  {Sondhi}}]{vonKeyserlingk2016a}%
  \BibitemOpen
  \bibfield  {author} {\bibinfo {author} {\bibfnamefont {C.~W.}\ \bibnamefont
  {von Keyserlingk}}, \bibinfo {author} {\bibfnamefont {Vedika}\ \bibnamefont
  {Khemani}}, \ and\ \bibinfo {author} {\bibfnamefont {S.~L.}\ \bibnamefont
  {Sondhi}},\ }\bibfield  {title} {\enquote {\bibinfo {title} {Absolute
  stability and spatiotemporal long-range order in floquet systems},}\ }\href
  {http://arxiv.org/abs/1605.00639} {\bibfield  {journal} {\bibinfo  {journal}
  {Phys. Rev. B}\ }\textbf {\bibinfo {volume} {94}},\ \bibinfo {pages} {085112}
  (\bibinfo {year} {2016})},\ \Eprint {http://arxiv.org/abs/1605.00639}
  {arXiv:1605.00639} \BibitemShut {NoStop}%
\bibitem [{\citenamefont {D'Alessio}\ and\ \citenamefont
  {Rigol}(2014)}]{DAlessio2014}%
  \BibitemOpen
  \bibfield  {author} {\bibinfo {author} {\bibfnamefont {Luca}\ \bibnamefont
  {D'Alessio}}\ and\ \bibinfo {author} {\bibfnamefont {Marcos}\ \bibnamefont
  {Rigol}},\ }\bibfield  {title} {\enquote {\bibinfo {title} {{Long-time
  Behavior of Isolated Periodically Driven Interacting Lattice Systems}},}\
  }\href {\doibase 10.1103/PhysRevX.4.041048} {\bibfield  {journal} {\bibinfo
  {journal} {Phys. Rev. X}\ }\textbf {\bibinfo {volume} {4}},\ \bibinfo {pages}
  {041048} (\bibinfo {year} {2014})},\ \Eprint {http://arxiv.org/abs/1402.5141}
  {arXiv:1402.5141} \BibitemShut {NoStop}%
\bibitem [{\citenamefont {Lazarides}\ \emph {et~al.}(2014)\citenamefont
  {Lazarides}, \citenamefont {Das},\ and\ \citenamefont
  {Moessner}}]{Lazarides2014}%
  \BibitemOpen
  \bibfield  {author} {\bibinfo {author} {\bibfnamefont {Achilleas}\
  \bibnamefont {Lazarides}}, \bibinfo {author} {\bibfnamefont {Arnab}\
  \bibnamefont {Das}}, \ and\ \bibinfo {author} {\bibfnamefont {Roderich}\
  \bibnamefont {Moessner}},\ }\bibfield  {title} {\enquote {\bibinfo {title}
  {{Equilibrium states of generic quantum systems subject to periodic
  driving}},}\ }\href {\doibase 10.1103/PhysRevE.90.012110} {\bibfield
  {journal} {\bibinfo  {journal} {Phys. Rev. E}\ }\textbf {\bibinfo {volume}
  {90}},\ \bibinfo {pages} {012110} (\bibinfo {year} {2014})},\ \Eprint
  {http://arxiv.org/abs/1403.2946} {arXiv:1403.2946} \BibitemShut {NoStop}%
\bibitem [{\citenamefont {Ponte}\ \emph {et~al.}(2015)\citenamefont {Ponte},
  \citenamefont {Chandran}, \citenamefont {Papi{\'{c}}},\ and\ \citenamefont
  {Abanin}}]{Ponte2015a}%
  \BibitemOpen
  \bibfield  {author} {\bibinfo {author} {\bibfnamefont {Pedro}\ \bibnamefont
  {Ponte}}, \bibinfo {author} {\bibfnamefont {Anushya}\ \bibnamefont
  {Chandran}}, \bibinfo {author} {\bibfnamefont {Z.}~\bibnamefont
  {Papi{\'{c}}}}, \ and\ \bibinfo {author} {\bibfnamefont {Dmitry~A.}\
  \bibnamefont {Abanin}},\ }\bibfield  {title} {\enquote {\bibinfo {title}
  {{Periodically driven ergodic and many-body localized quantum systems}},}\
  }\href {\doibase 10.1016/j.aop.2014.11.008} {\bibfield  {journal} {\bibinfo
  {journal} {Ann. Phys.}\ }\textbf {\bibinfo {volume} {353}},\ \bibinfo {pages}
  {196--204} (\bibinfo {year} {2015})},\ \Eprint
  {http://arxiv.org/abs/1403.6480} {arXiv:1403.6480} \BibitemShut {NoStop}%
\bibitem [{\citenamefont {Abanin}\ \emph {et~al.}(2016)\citenamefont {Abanin},
  \citenamefont {Roeck},\ and\ \citenamefont {Huveneers}}]{Abanin2014}%
  \BibitemOpen
  \bibfield  {author} {\bibinfo {author} {\bibfnamefont {Dmitry~A.}\
  \bibnamefont {Abanin}}, \bibinfo {author} {\bibfnamefont {Wojciech~De}\
  \bibnamefont {Roeck}}, \ and\ \bibinfo {author} {\bibfnamefont
  {Fran{\c{c}}ois}\ \bibnamefont {Huveneers}},\ }\bibfield  {title} {\enquote
  {\bibinfo {title} {{A theory of many-body localization in periodically driven
  systems}},}\ }\href {\doibase 10.1016/j.aop.2016.03.010} {\bibfield
  {journal} {\bibinfo  {journal} {Ann. Phys.}\ }\textbf {\bibinfo {volume}
  {372}},\ \bibinfo {pages} {5} (\bibinfo {year} {2016})},\ \Eprint
  {http://arxiv.org/abs/1412.4752} {arXiv:1412.4752} \BibitemShut {NoStop}%
\bibitem [{\citenamefont {{Ponte}}\ \emph
  {et~al.}(2015{\natexlab{a}})\citenamefont {{Ponte}}, \citenamefont
  {{Chandran}}, \citenamefont {{Papi{\'c}}},\ and\ \citenamefont
  {{Abanin}}}]{Ponte15a}%
  \BibitemOpen
  \bibfield  {author} {\bibinfo {author} {\bibfnamefont {P.}~\bibnamefont
  {{Ponte}}}, \bibinfo {author} {\bibfnamefont {A.}~\bibnamefont {{Chandran}}},
  \bibinfo {author} {\bibfnamefont {Z.}~\bibnamefont {{Papi{\'c}}}}, \ and\
  \bibinfo {author} {\bibfnamefont {D.~A.}\ \bibnamefont {{Abanin}}},\
  }\bibfield  {title} {\enquote {\bibinfo {title} {{Periodically driven ergodic
  and many-body localized quantum systems}},}\ }\href {\doibase
  10.1016/j.aop.2014.11.008} {\bibfield  {journal} {\bibinfo  {journal} {Ann.
  Phys.}\ }\textbf {\bibinfo {volume} {353}},\ \bibinfo {pages} {196--204}
  (\bibinfo {year} {2015}{\natexlab{a}})},\ \Eprint
  {http://arxiv.org/abs/1403.6480} {arXiv:1403.6480 [cond-mat.dis-nn]}
  \BibitemShut {NoStop}%
\bibitem [{\citenamefont {{Ponte}}\ \emph
  {et~al.}(2015{\natexlab{b}})\citenamefont {{Ponte}}, \citenamefont
  {{Papi{\'c}}}, \citenamefont {{Huveneers}},\ and\ \citenamefont
  {{Abanin}}}]{Ponte15b}%
  \BibitemOpen
  \bibfield  {author} {\bibinfo {author} {\bibfnamefont {P.}~\bibnamefont
  {{Ponte}}}, \bibinfo {author} {\bibfnamefont {Z.}~\bibnamefont
  {{Papi{\'c}}}}, \bibinfo {author} {\bibfnamefont {F.}~\bibnamefont
  {{Huveneers}}}, \ and\ \bibinfo {author} {\bibfnamefont {D.~A.}\ \bibnamefont
  {{Abanin}}},\ }\bibfield  {title} {\enquote {\bibinfo {title} {{Many-Body
  Localization in Periodically Driven Systems}},}\ }\href {\doibase
  10.1103/PhysRevLett.114.140401} {\bibfield  {journal} {\bibinfo  {journal}
  {Phys. Rev. Lett.}\ }\textbf {\bibinfo {volume} {114}},\ \bibinfo {eid}
  {140401} (\bibinfo {year} {2015}{\natexlab{b}})},\ \Eprint
  {http://arxiv.org/abs/1410.8518} {arXiv:1410.8518 [cond-mat.dis-nn]}
  \BibitemShut {NoStop}%
\bibitem [{\citenamefont {{Lazarides}}\ \emph {et~al.}(2015)\citenamefont
  {{Lazarides}}, \citenamefont {{Das}},\ and\ \citenamefont
  {{Moessner}}}]{Lazarides15}%
  \BibitemOpen
  \bibfield  {author} {\bibinfo {author} {\bibfnamefont {A.}~\bibnamefont
  {{Lazarides}}}, \bibinfo {author} {\bibfnamefont {A.}~\bibnamefont {{Das}}},
  \ and\ \bibinfo {author} {\bibfnamefont {R.}~\bibnamefont {{Moessner}}},\
  }\bibfield  {title} {\enquote {\bibinfo {title} {{Fate of Many-Body
  Localization Under Periodic Driving}},}\ }\href {\doibase
  10.1103/PhysRevLett.115.030402} {\bibfield  {journal} {\bibinfo  {journal}
  {Phys. Rev. Lett.}\ }\textbf {\bibinfo {volume} {115}},\ \bibinfo {eid}
  {030402} (\bibinfo {year} {2015})},\ \Eprint {http://arxiv.org/abs/1410.3455}
  {arXiv:1410.3455 [cond-mat.stat-mech]} \BibitemShut {NoStop}%
\bibitem [{\citenamefont {{Iadecola}}\ \emph {et~al.}(2015)\citenamefont
  {{Iadecola}}, \citenamefont {{Santos}},\ and\ \citenamefont
  {{Chamon}}}]{Iadecola15}%
  \BibitemOpen
  \bibfield  {author} {\bibinfo {author} {\bibfnamefont {T.}~\bibnamefont
  {{Iadecola}}}, \bibinfo {author} {\bibfnamefont {L.~H.}\ \bibnamefont
  {{Santos}}}, \ and\ \bibinfo {author} {\bibfnamefont {C.}~\bibnamefont
  {{Chamon}}},\ }\bibfield  {title} {\enquote {\bibinfo {title} {{Stroboscopic
  symmetry-protected topological phases}},}\ }\href {\doibase
  10.1103/PhysRevB.92.125107} {\bibfield  {journal} {\bibinfo  {journal}
  {\prb}\ }\textbf {\bibinfo {volume} {92}},\ \bibinfo {eid} {125107} (\bibinfo
  {year} {2015})},\ \Eprint {http://arxiv.org/abs/1503.07871} {arXiv:1503.07871
  [cond-mat.str-el]} \BibitemShut {NoStop}%
\bibitem [{\citenamefont {{Basko}}\ \emph {et~al.}(2006)\citenamefont
  {{Basko}}, \citenamefont {{Aleiner}},\ and\ \citenamefont
  {{Altshuler}}}]{Basko06a}%
  \BibitemOpen
  \bibfield  {author} {\bibinfo {author} {\bibfnamefont {D.~M.}\ \bibnamefont
  {{Basko}}}, \bibinfo {author} {\bibfnamefont {I.~L.}\ \bibnamefont
  {{Aleiner}}}, \ and\ \bibinfo {author} {\bibfnamefont {B.~L.}\ \bibnamefont
  {{Altshuler}}},\ }\bibfield  {title} {\enquote {\bibinfo {title} {{Metal
  insulator transition in a weakly interacting many-electron system with
  localized single-particle states}},}\ }\href {\doibase
  10.1016/j.aop.2005.11.014} {\bibfield  {journal} {\bibinfo  {journal} {Annals
  of Physics}\ }\textbf {\bibinfo {volume} {321}},\ \bibinfo {pages}
  {1126--1205} (\bibinfo {year} {2006})},\ \Eprint
  {http://arxiv.org/abs/cond-mat/0506617} {arXiv:cond-mat/0506617} \BibitemShut
  {NoStop}%
\bibitem [{\citenamefont {{Basko}}\ \emph {et~al.}()\citenamefont {{Basko}},
  \citenamefont {{Aleiner}},\ and\ \citenamefont {{Altshuler}}}]{Basko06b}%
  \BibitemOpen
  \bibfield  {author} {\bibinfo {author} {\bibfnamefont {D.~M.}\ \bibnamefont
  {{Basko}}}, \bibinfo {author} {\bibfnamefont {I.~L.}\ \bibnamefont
  {{Aleiner}}}, \ and\ \bibinfo {author} {\bibfnamefont {B.~L.}\ \bibnamefont
  {{Altshuler}}},\ }\href@noop {} {\enquote {\bibinfo {title} {{On the problem
  of many-body localization}},}\ }\Eprint
  {http://arxiv.org/abs/cond-mat/0602510} {arXiv:cond-mat/0602510} \BibitemShut
  {NoStop}%
\bibitem [{\citenamefont {{Oganesyan}}\ and\ \citenamefont
  {{Huse}}(2007)}]{Oganesyan07}%
  \BibitemOpen
  \bibfield  {author} {\bibinfo {author} {\bibfnamefont {V.}~\bibnamefont
  {{Oganesyan}}}\ and\ \bibinfo {author} {\bibfnamefont {D.~A.}\ \bibnamefont
  {{Huse}}},\ }\bibfield  {title} {\enquote {\bibinfo {title} {{Localization of
  interacting fermions at high temperature}},}\ }\href {\doibase
  10.1103/PhysRevB.75.155111} {\bibfield  {journal} {\bibinfo  {journal}
  {\prb}\ }\textbf {\bibinfo {volume} {75}},\ \bibinfo {pages} {155111}
  (\bibinfo {year} {2007})}\BibitemShut {NoStop}%
\bibitem [{\citenamefont {\u{Z}nidari\u{c}}\ \emph {et~al.}(2008)\citenamefont
  {\u{Z}nidari\u{c}}, \citenamefont {Prosen},\ and\ \citenamefont
  {Prelov\u{s}ek}}]{Znidaric08}%
  \BibitemOpen
  \bibfield  {author} {\bibinfo {author} {\bibfnamefont {Marko}\ \bibnamefont
  {\u{Z}nidari\u{c}}}, \bibinfo {author} {\bibfnamefont {Toma\u{z}}\
  \bibnamefont {Prosen}}, \ and\ \bibinfo {author} {\bibfnamefont {Peter}\
  \bibnamefont {Prelov\u{s}ek}},\ }\bibfield  {title} {\enquote {\bibinfo
  {title} {Many-body localization in the {Heisenberg $XXZ$} magnet in a random
  field},}\ }\href {\doibase 10.1103/PhysRevB.77.064426} {\bibfield  {journal}
  {\bibinfo  {journal} {Phys. Rev. B}\ }\textbf {\bibinfo {volume} {77}},\
  \bibinfo {pages} {064426} (\bibinfo {year} {2008})}\BibitemShut {NoStop}%
\bibitem [{\citenamefont {Pal}\ and\ \citenamefont {Huse}(2010)}]{Pal10}%
  \BibitemOpen
  \bibfield  {author} {\bibinfo {author} {\bibfnamefont {Arijeet}\ \bibnamefont
  {Pal}}\ and\ \bibinfo {author} {\bibfnamefont {David~A.}\ \bibnamefont
  {Huse}},\ }\bibfield  {title} {\enquote {\bibinfo {title} {Many-body
  localization phase transition},}\ }\href {\doibase
  10.1103/PhysRevB.82.174411} {\bibfield  {journal} {\bibinfo  {journal} {Phys.
  Rev. B}\ }\textbf {\bibinfo {volume} {82}},\ \bibinfo {pages} {174411}
  (\bibinfo {year} {2010})}\BibitemShut {NoStop}%
\bibitem [{\citenamefont {{Bardarson}}\ \emph {et~al.}(2012)\citenamefont
  {{Bardarson}}, \citenamefont {{Pollmann}},\ and\ \citenamefont
  {{Moore}}}]{Bardarson12}%
  \BibitemOpen
  \bibfield  {author} {\bibinfo {author} {\bibfnamefont {J.~H.}\ \bibnamefont
  {{Bardarson}}}, \bibinfo {author} {\bibfnamefont {F.}~\bibnamefont
  {{Pollmann}}}, \ and\ \bibinfo {author} {\bibfnamefont {J.~E.}\ \bibnamefont
  {{Moore}}},\ }\bibfield  {title} {\enquote {\bibinfo {title} {{Unbounded
  Growth of Entanglement in Models of Many-Body Localization}},}\ }\href
  {\doibase 10.1103/PhysRevLett.109.017202} {\bibfield  {journal} {\bibinfo
  {journal} {\prl}\ }\textbf {\bibinfo {volume} {109}},\ \bibinfo {pages}
  {017202} (\bibinfo {year} {2012})}\BibitemShut {NoStop}%
\bibitem [{\citenamefont {{Bauer}}\ and\ \citenamefont
  {{Nayak}}(2013)}]{Bauer13}%
  \BibitemOpen
  \bibfield  {author} {\bibinfo {author} {\bibfnamefont {B.}~\bibnamefont
  {{Bauer}}}\ and\ \bibinfo {author} {\bibfnamefont {C.}~\bibnamefont
  {{Nayak}}},\ }\bibfield  {title} {\enquote {\bibinfo {title} {{Area laws in a
  many-body localized state and its implications for topological order}},}\
  }\href {\doibase 10.1088/1742-5468/2013/09/P09005} {\bibfield  {journal}
  {\bibinfo  {journal} {J. Stat. Mech: Theor. Exp.}\ }\textbf {\bibinfo
  {volume} {9}},\ \bibinfo {eid} {09005} (\bibinfo {year} {2013})},\ \Eprint
  {http://arxiv.org/abs/1306.5753} {arXiv:1306.5753 [cond-mat.dis-nn]}
  \BibitemShut {NoStop}%
\bibitem [{\citenamefont {Serbyn}\ \emph {et~al.}(2013)\citenamefont {Serbyn},
  \citenamefont {Papi\'{c}},\ and\ \citenamefont {Abanin}}]{Serbyn13a}%
  \BibitemOpen
  \bibfield  {author} {\bibinfo {author} {\bibfnamefont {M.}~\bibnamefont
  {Serbyn}}, \bibinfo {author} {\bibfnamefont {Z.}~\bibnamefont {Papi\'{c}}}, \
  and\ \bibinfo {author} {\bibfnamefont {D.~A.}\ \bibnamefont {Abanin}},\
  }\bibfield  {title} {\enquote {\bibinfo {title} {Universal slow growth of
  entanglement in interacting strongly disordered systems},}\ }\href {\doibase
  10.1103/PhysRevLett.110.260601} {\bibfield  {journal} {\bibinfo  {journal}
  {\prl}\ }\textbf {\bibinfo {volume} {110}},\ \bibinfo {pages} {260601}
  (\bibinfo {year} {2013})},\ \Eprint {http://arxiv.org/abs/1304.4605}
  {arXiv:1304.4605} \BibitemShut {NoStop}%
\bibitem [{\citenamefont {{Serbyn}}\ \emph {et~al.}(2013)\citenamefont
  {{Serbyn}}, \citenamefont {{Papi{\'c}}},\ and\ \citenamefont
  {{Abanin}}}]{Serbyn13b}%
  \BibitemOpen
  \bibfield  {author} {\bibinfo {author} {\bibfnamefont {M.}~\bibnamefont
  {{Serbyn}}}, \bibinfo {author} {\bibfnamefont {Z.}~\bibnamefont
  {{Papi{\'c}}}}, \ and\ \bibinfo {author} {\bibfnamefont {D.~A.}\ \bibnamefont
  {{Abanin}}},\ }\bibfield  {title} {\enquote {\bibinfo {title} {{Local
  Conservation Laws and the Structure of the Many-Body Localized States}},}\
  }\href {\doibase 10.1103/PhysRevLett.111.127201} {\bibfield  {journal}
  {\bibinfo  {journal} {Phys. Rev. Lett.}\ }\textbf {\bibinfo {volume} {111}},\
  \bibinfo {eid} {127201} (\bibinfo {year} {2013})},\ \Eprint
  {http://arxiv.org/abs/1305.5554} {arXiv:1305.5554 [cond-mat.dis-nn]}
  \BibitemShut {NoStop}%
\bibitem [{\citenamefont {Huse}\ \emph {et~al.}(2014)\citenamefont {Huse},
  \citenamefont {Nandkishore},\ and\ \citenamefont {Oganesyan}}]{Huse14}%
  \BibitemOpen
  \bibfield  {author} {\bibinfo {author} {\bibfnamefont {David~A.}\
  \bibnamefont {Huse}}, \bibinfo {author} {\bibfnamefont {Rahul}\ \bibnamefont
  {Nandkishore}}, \ and\ \bibinfo {author} {\bibfnamefont {Vadim}\ \bibnamefont
  {Oganesyan}},\ }\bibfield  {title} {\enquote {\bibinfo {title} {Phenomenology
  of fully many-body-localized systems},}\ }\href {\doibase
  10.1103/PhysRevB.90.174202} {\bibfield  {journal} {\bibinfo  {journal} {Phys.
  Rev. B}\ }\textbf {\bibinfo {volume} {90}},\ \bibinfo {pages} {174202}
  (\bibinfo {year} {2014})},\ \Eprint {http://arxiv.org/abs/1305.4915}
  {arXiv:1305.4915} \BibitemShut {NoStop}%
\bibitem [{\citenamefont {Levi}\ \emph {et~al.}(2016)\citenamefont {Levi},
  \citenamefont {Heyl}, \citenamefont {Lesanovsky},\ and\ \citenamefont
  {Garrahan}}]{Levi2015}%
  \BibitemOpen
  \bibfield  {author} {\bibinfo {author} {\bibfnamefont {Emanuele}\
  \bibnamefont {Levi}}, \bibinfo {author} {\bibfnamefont {Markus}\ \bibnamefont
  {Heyl}}, \bibinfo {author} {\bibfnamefont {Igor}\ \bibnamefont {Lesanovsky}},
  \ and\ \bibinfo {author} {\bibfnamefont {Juan~P.}\ \bibnamefont {Garrahan}},\
  }\bibfield  {title} {\enquote {\bibinfo {title} {Robustness of many-body
  localization in the presence of dissipation},}\ }\href {\doibase
  10.1103/PhysRevLett.116.237203} {\bibfield  {journal} {\bibinfo  {journal}
  {Phys. Rev. Lett.}\ }\textbf {\bibinfo {volume} {116}},\ \bibinfo {pages}
  {237203} (\bibinfo {year} {2016})},\ \Eprint
  {http://arxiv.org/abs/1510.04634} {arXiv:1510.04634} \BibitemShut {NoStop}%
\bibitem [{\citenamefont {{Fischer}}\ \emph {et~al.}(2016)\citenamefont
  {{Fischer}}, \citenamefont {{Maksymenko}},\ and\ \citenamefont
  {{Altman}}}]{Fischer2015}%
  \BibitemOpen
  \bibfield  {author} {\bibinfo {author} {\bibfnamefont {M.~H.}\ \bibnamefont
  {{Fischer}}}, \bibinfo {author} {\bibfnamefont {M.}~\bibnamefont
  {{Maksymenko}}}, \ and\ \bibinfo {author} {\bibfnamefont {E.}~\bibnamefont
  {{Altman}}},\ }\bibfield  {title} {\enquote {\bibinfo {title} {{Dynamics of a
  Many-Body-Localized System Coupled to a Bath}},}\ }\href {\doibase
  10.1103/PhysRevLett.116.160401} {\bibfield  {journal} {\bibinfo  {journal}
  {Physical Review Letters}\ }\textbf {\bibinfo {volume} {116}},\ \bibinfo
  {eid} {160401} (\bibinfo {year} {2016})},\ \Eprint
  {http://arxiv.org/abs/1512.02669} {arXiv:1512.02669 [cond-mat.str-el]}
  \BibitemShut {NoStop}%
\bibitem [{\citenamefont {Nandkishore}\ \emph {et~al.}(2014)\citenamefont
  {Nandkishore}, \citenamefont {Gopalakrishnan},\ and\ \citenamefont
  {Huse}}]{Nandkishore2014}%
  \BibitemOpen
  \bibfield  {author} {\bibinfo {author} {\bibfnamefont {Rahul}\ \bibnamefont
  {Nandkishore}}, \bibinfo {author} {\bibfnamefont {Sarang}\ \bibnamefont
  {Gopalakrishnan}}, \ and\ \bibinfo {author} {\bibfnamefont {David~A.}\
  \bibnamefont {Huse}},\ }\bibfield  {title} {\enquote {\bibinfo {title}
  {Spectral features of a many-body-localized system weakly coupled to a
  bath},}\ }\href {\doibase 10.1103/PhysRevB.90.064203} {\bibfield  {journal}
  {\bibinfo  {journal} {Phys. Rev. B}\ }\textbf {\bibinfo {volume} {90}},\
  \bibinfo {pages} {064203} (\bibinfo {year} {2014})}\BibitemShut {NoStop}%
\bibitem [{\citenamefont {Gopalakrishnan}\ and\ \citenamefont
  {Nandkishore}(2014)}]{Gopalakrishnan2014}%
  \BibitemOpen
  \bibfield  {author} {\bibinfo {author} {\bibfnamefont {Sarang}\ \bibnamefont
  {Gopalakrishnan}}\ and\ \bibinfo {author} {\bibfnamefont {Rahul}\
  \bibnamefont {Nandkishore}},\ }\bibfield  {title} {\enquote {\bibinfo {title}
  {Mean-field theory of nearly many-body localized metals},}\ }\href {\doibase
  10.1103/PhysRevB.90.224203} {\bibfield  {journal} {\bibinfo  {journal} {Phys.
  Rev. B}\ }\textbf {\bibinfo {volume} {90}},\ \bibinfo {pages} {224203}
  (\bibinfo {year} {2014})}\BibitemShut {NoStop}%
\bibitem [{\citenamefont {Johri}\ \emph {et~al.}(2015)\citenamefont {Johri},
  \citenamefont {Nandkishore},\ and\ \citenamefont {Bhatt}}]{Johri2015}%
  \BibitemOpen
  \bibfield  {author} {\bibinfo {author} {\bibfnamefont {Sonika}\ \bibnamefont
  {Johri}}, \bibinfo {author} {\bibfnamefont {Rahul}\ \bibnamefont
  {Nandkishore}}, \ and\ \bibinfo {author} {\bibfnamefont {R.~N.}\ \bibnamefont
  {Bhatt}},\ }\bibfield  {title} {\enquote {\bibinfo {title} {Many-body
  localization in imperfectly isolated quantum systems},}\ }\href {\doibase
  10.1103/PhysRevLett.114.117401} {\bibfield  {journal} {\bibinfo  {journal}
  {Phys. Rev. Lett.}\ }\textbf {\bibinfo {volume} {114}},\ \bibinfo {pages}
  {117401} (\bibinfo {year} {2015})}\BibitemShut {NoStop}%
\bibitem [{\citenamefont {Nandkishore}(2015)}]{Nandkishore2015b}%
  \BibitemOpen
  \bibfield  {author} {\bibinfo {author} {\bibfnamefont {R.}~\bibnamefont
  {Nandkishore}},\ }\bibfield  {title} {\enquote {\bibinfo {title} {Many-body
  localization proximity effect},}\ }\href {\doibase
  10.1103/PhysRevB.92.245141} {\bibfield  {journal} {\bibinfo  {journal} {Phys.
  Rev. B}\ }\textbf {\bibinfo {volume} {92}},\ \bibinfo {pages} {245141}
  (\bibinfo {year} {2015})},\ \Eprint {http://arxiv.org/abs/1506.05468}
  {arXiv:1506.05468} \BibitemShut {NoStop}%
\bibitem [{\citenamefont {Li}\ \emph {et~al.}(2015)\citenamefont {Li},
  \citenamefont {Ganeshan}, \citenamefont {Pixley},\ and\ \citenamefont {{Das
  Sarma}}}]{Li2015}%
  \BibitemOpen
  \bibfield  {author} {\bibinfo {author} {\bibfnamefont {X.}~\bibnamefont
  {Li}}, \bibinfo {author} {\bibfnamefont {S.}~\bibnamefont {Ganeshan}},
  \bibinfo {author} {\bibfnamefont {J.~H.}\ \bibnamefont {Pixley}}, \ and\
  \bibinfo {author} {\bibfnamefont {S.}~\bibnamefont {{Das Sarma}}},\
  }\bibfield  {title} {\enquote {\bibinfo {title} {Many body localization and
  quantum non-ergodicity in a model with a single-particle mobility edge},}\
  }\href {\doibase 10.1103/PhysRevLett.115.186601} {\bibfield  {journal}
  {\bibinfo  {journal} {Phys. Rev. Lett.}\ }\textbf {\bibinfo {volume} {115}},\
  \bibinfo {pages} {186601} (\bibinfo {year} {2015})},\ \Eprint
  {http://arxiv.org/abs/1504.00016} {arXiv:1504.00016} \BibitemShut {NoStop}%
\bibitem [{\citenamefont {{Nandkishore}}\ and\ \citenamefont
  {{Gopalakrishnan}}()}]{Nandkishore2016}%
  \BibitemOpen
  \bibfield  {author} {\bibinfo {author} {\bibfnamefont {R.}~\bibnamefont
  {{Nandkishore}}}\ and\ \bibinfo {author} {\bibfnamefont {S.}~\bibnamefont
  {{Gopalakrishnan}}},\ }\href@noop {} {\enquote {\bibinfo {title} {{General
  theory of many body localized systems coupled to baths}},}\ }\Eprint
  {http://arxiv.org/abs/1606.08465} {arXiv:1606.08465} \BibitemShut {NoStop}%
\bibitem [{\citenamefont {{Hyatt}}\ \emph {et~al.}()\citenamefont {{Hyatt}},
  \citenamefont {{Garrison}}, \citenamefont {{Potter}},\ and\ \citenamefont
  {{Bauer}}}]{Hyatt2016}%
  \BibitemOpen
  \bibfield  {author} {\bibinfo {author} {\bibfnamefont {K.}~\bibnamefont
  {{Hyatt}}}, \bibinfo {author} {\bibfnamefont {J.~R.}\ \bibnamefont
  {{Garrison}}}, \bibinfo {author} {\bibfnamefont {A.~C.}\ \bibnamefont
  {{Potter}}}, \ and\ \bibinfo {author} {\bibfnamefont {B.}~\bibnamefont
  {{Bauer}}},\ }\bibfield  {title} {\enquote {\bibinfo {title} {{Many-body
  localization in the presence of a small bath}},}\ }\href {\doibase
  10.1103/PhysRevB.95.035132} {\bibfield  {journal} {\bibinfo  {journal} {Phys.
  Rev. B}\ }\textbf {\bibinfo {volume} {95}},\ \bibinfo {pages} {035132}},\
  \Eprint {http://arxiv.org/abs/1601.07184} {arXiv:1601.07184} \BibitemShut
  {NoStop}%
\bibitem [{\citenamefont {Abanin}\ \emph
  {et~al.}(2015{\natexlab{a}})\citenamefont {Abanin}, \citenamefont {{De
  Roeck}},\ and\ \citenamefont {Huveneers}}]{Abanin2015}%
  \BibitemOpen
  \bibfield  {author} {\bibinfo {author} {\bibfnamefont {Dmitry~A.}\
  \bibnamefont {Abanin}}, \bibinfo {author} {\bibfnamefont {Wojciech}\
  \bibnamefont {{De Roeck}}}, \ and\ \bibinfo {author} {\bibfnamefont
  {Fran{\c{c}}ois}\ \bibnamefont {Huveneers}},\ }\bibfield  {title} {\enquote
  {\bibinfo {title} {{Exponentially Slow Heating in Periodically Driven
  Many-Body Systems}},}\ }\href {\doibase 10.1103/PhysRevLett.115.256803}
  {\bibfield  {journal} {\bibinfo  {journal} {Phys. Rev. Lett.}\ }\textbf
  {\bibinfo {volume} {115}},\ \bibinfo {pages} {256803} (\bibinfo {year}
  {2015}{\natexlab{a}})},\ \Eprint {http://arxiv.org/abs/1507.01474}
  {arXiv:1507.01474} \BibitemShut {NoStop}%
\bibitem [{\citenamefont {Abanin}\ \emph {et~al.}(2017)\citenamefont {Abanin},
  \citenamefont {{De Roeck}},\ and\ \citenamefont {Ho}}]{Abanin2015a}%
  \BibitemOpen
  \bibfield  {author} {\bibinfo {author} {\bibfnamefont {D~A}\ \bibnamefont
  {Abanin}}, \bibinfo {author} {\bibfnamefont {W}~\bibnamefont {{De Roeck}}}, \
  and\ \bibinfo {author} {\bibfnamefont {W~W}\ \bibnamefont {Ho}},\ }\bibfield
  {title} {\enquote {\bibinfo {title} {{Effective Hamiltonians,
  prethermalization and slow energy absorption in periodically driven many-body
  systems}},}\ }\href {http://arxiv.org/abs/1510.03405} {\bibfield  {journal}
  {\bibinfo  {journal} {Phys. Rev. B}\ }\textbf {\bibinfo {volume} {95}},\
  \bibinfo {pages} {8} (\bibinfo {year} {2017})},\ \Eprint
  {http://arxiv.org/abs/1510.03405} {arXiv:1510.03405} \BibitemShut {NoStop}%
\bibitem [{\citenamefont {Abanin}\ \emph
  {et~al.}(2015{\natexlab{b}})\citenamefont {Abanin}, \citenamefont {{De
  Roeck}}, \citenamefont {Huveneers},\ and\ \citenamefont {Ho}}]{Abanin2015b}%
  \BibitemOpen
  \bibfield  {author} {\bibinfo {author} {\bibfnamefont {Dmitry}\ \bibnamefont
  {Abanin}}, \bibinfo {author} {\bibfnamefont {Wojciech}\ \bibnamefont {{De
  Roeck}}}, \bibinfo {author} {\bibfnamefont {Francois}\ \bibnamefont
  {Huveneers}}, \ and\ \bibinfo {author} {\bibfnamefont {Wen~Wei}\ \bibnamefont
  {Ho}},\ }\href {http://arxiv.org/abs/1509.05386} {\enquote {\bibinfo {title}
  {{A rigorous theory of many-body prethermalization for periodically driven
  and closed quantum systems}},}\ } (\bibinfo {year} {2015}{\natexlab{b}}),\
  \Eprint {http://arxiv.org/abs/1509.05386} {arXiv:1509.05386} \BibitemShut
  {NoStop}%
\bibitem [{\citenamefont {Kuwahara}\ \emph {et~al.}(2016)\citenamefont
  {Kuwahara}, \citenamefont {Mori},\ and\ \citenamefont
  {Saito}}]{Kuwahara2015}%
  \BibitemOpen
  \bibfield  {author} {\bibinfo {author} {\bibfnamefont {Tomotaka}\
  \bibnamefont {Kuwahara}}, \bibinfo {author} {\bibfnamefont {Takashi}\
  \bibnamefont {Mori}}, \ and\ \bibinfo {author} {\bibfnamefont {Keiji}\
  \bibnamefont {Saito}},\ }\bibfield  {title} {\enquote {\bibinfo {title}
  {{Floquet–Magnus theory and generic transient dynamics in periodically
  driven many-body quantum systems}},}\ }\href {\doibase
  10.1016/j.aop.2016.01.012} {\bibfield  {journal} {\bibinfo  {journal} {Ann.
  Phys.}\ }\textbf {\bibinfo {volume} {367}},\ \bibinfo {pages} {96--124}
  (\bibinfo {year} {2016})},\ \Eprint {http://arxiv.org/abs/1508.05797}
  {arXiv:1508.05797} \BibitemShut {NoStop}%
\bibitem [{\citenamefont {Mori}\ \emph {et~al.}(2016)\citenamefont {Mori},
  \citenamefont {Kuwahara},\ and\ \citenamefont {Saito}}]{Mori2016}%
  \BibitemOpen
  \bibfield  {author} {\bibinfo {author} {\bibfnamefont {Takashi}\ \bibnamefont
  {Mori}}, \bibinfo {author} {\bibfnamefont {Tomotaka}\ \bibnamefont
  {Kuwahara}}, \ and\ \bibinfo {author} {\bibfnamefont {Keiji}\ \bibnamefont
  {Saito}},\ }\bibfield  {title} {\enquote {\bibinfo {title} {{Rigorous Bound
  on Energy Absorption and Generic Relaxation in Periodically Driven Quantum
  Systems}},}\ }\href {\doibase 10.1103/PhysRevLett.116.120401} {\bibfield
  {journal} {\bibinfo  {journal} {Physical Review Letters}\ }\textbf {\bibinfo
  {volume} {116}},\ \bibinfo {pages} {120401} (\bibinfo {year} {2016})},\
  \Eprint {http://arxiv.org/abs/1509.03968} {arXiv:1509.03968} \BibitemShut
  {NoStop}%
\bibitem [{\citenamefont {Bukov}\ \emph {et~al.}(2015)\citenamefont {Bukov},
  \citenamefont {Gopalakrishnan}, \citenamefont {Knap},\ and\ \citenamefont
  {Demler}}]{bukov2015}%
  \BibitemOpen
  \bibfield  {author} {\bibinfo {author} {\bibfnamefont {Marin}\ \bibnamefont
  {Bukov}}, \bibinfo {author} {\bibfnamefont {Sarang}\ \bibnamefont
  {Gopalakrishnan}}, \bibinfo {author} {\bibfnamefont {Michael}\ \bibnamefont
  {Knap}}, \ and\ \bibinfo {author} {\bibfnamefont {Eugene}\ \bibnamefont
  {Demler}},\ }\bibfield  {title} {\enquote {\bibinfo {title} {Prethermal
  floquet steady states and instabilities in the periodically driven, weakly
  interacting bose-hubbard model},}\ }\href {\doibase
  10.1103/PhysRevLett.115.205301} {\bibfield  {journal} {\bibinfo  {journal}
  {Phys. Rev. Lett.}\ }\textbf {\bibinfo {volume} {115}},\ \bibinfo {pages}
  {205301} (\bibinfo {year} {2015})}\BibitemShut {NoStop}%
\bibitem [{\citenamefont {Canovi}\ \emph {et~al.}(2016)\citenamefont {Canovi},
  \citenamefont {Kollar},\ and\ \citenamefont {Eckstein}}]{canovi2016}%
  \BibitemOpen
  \bibfield  {author} {\bibinfo {author} {\bibfnamefont {Elena}\ \bibnamefont
  {Canovi}}, \bibinfo {author} {\bibfnamefont {Marcus}\ \bibnamefont {Kollar}},
  \ and\ \bibinfo {author} {\bibfnamefont {Martin}\ \bibnamefont {Eckstein}},\
  }\bibfield  {title} {\enquote {\bibinfo {title} {Stroboscopic
  prethermalization in weakly interacting periodically driven systems},}\
  }\href {\doibase 10.1103/PhysRevE.93.012130} {\bibfield  {journal} {\bibinfo
  {journal} {Phys. Rev. E}\ }\textbf {\bibinfo {volume} {93}},\ \bibinfo
  {pages} {012130} (\bibinfo {year} {2016})}\BibitemShut {NoStop}%
\bibitem [{\citenamefont {Bukov}\ \emph {et~al.}(2016)\citenamefont {Bukov},
  \citenamefont {Heyl}, \citenamefont {Huse},\ and\ \citenamefont
  {Polkovnikov}}]{bukov2016}%
  \BibitemOpen
  \bibfield  {author} {\bibinfo {author} {\bibfnamefont {Marin}\ \bibnamefont
  {Bukov}}, \bibinfo {author} {\bibfnamefont {Markus}\ \bibnamefont {Heyl}},
  \bibinfo {author} {\bibfnamefont {David~A.}\ \bibnamefont {Huse}}, \ and\
  \bibinfo {author} {\bibfnamefont {Anatoli}\ \bibnamefont {Polkovnikov}},\
  }\bibfield  {title} {\enquote {\bibinfo {title} {Heating and many-body
  resonances in a periodically driven two-band system},}\ }\href {\doibase
  10.1103/PhysRevB.93.155132} {\bibfield  {journal} {\bibinfo  {journal} {Phys.
  Rev. B}\ }\textbf {\bibinfo {volume} {93}},\ \bibinfo {pages} {155132}
  (\bibinfo {year} {2016})}\BibitemShut {NoStop}%
\bibitem [{\citenamefont {Schwinger}(1961)}]{Schwinger61}%
  \BibitemOpen
  \bibfield  {author} {\bibinfo {author} {\bibfnamefont {J.}~\bibnamefont
  {Schwinger}},\ }\bibfield  {title} {\enquote {\bibinfo {title} {Brownian
  motion of a quantum oscillator},}\ }\href {\doibase 10.1063/1.1703727}
  {\bibfield  {journal} {\bibinfo  {journal} {J. Math. Phys.}\ }\textbf
  {\bibinfo {volume} {2}},\ \bibinfo {pages} {407} (\bibinfo {year}
  {1961})}\BibitemShut {NoStop}%
\bibitem [{\citenamefont {Keldysh}(1964)}]{Keldysh64}%
  \BibitemOpen
  \bibfield  {author} {\bibinfo {author} {\bibfnamefont {L.~V.}\ \bibnamefont
  {Keldysh}},\ }\bibfield  {title} {\enquote {\bibinfo {title} {Diagram
  technique fon non-equilibrium processes},}\ }\href@noop {} {\bibfield
  {journal} {\bibinfo  {journal} {Zh. Eksp. Teor. Fiz.}\ }\textbf {\bibinfo
  {volume} {47}},\ \bibinfo {pages} {1515} (\bibinfo {year} {1964})},\ \bibinfo
  {note} {[Sov. Phys. JETP {\bf 20}, 1018 (1965)]}\BibitemShut {NoStop}%
\bibitem [{\citenamefont {{Kamenev}}()}]{Kamenev04}%
  \BibitemOpen
  \bibfield  {author} {\bibinfo {author} {\bibfnamefont {A.}~\bibnamefont
  {{Kamenev}}},\ }\href@noop {} {\enquote {\bibinfo {title} {{Many-body theory
  of non-equilibrium systems}},}\ }\bibinfo {note} {Eprint
  arXiv:cond-mat/0412296}\BibitemShut {NoStop}%
\bibitem [{\citenamefont {Else}\ \emph {et~al.}()\citenamefont {Else},
  \citenamefont {Bauer},\ and\ \citenamefont {Nayak}}]{OpenSystems}%
  \BibitemOpen
  \bibfield  {author} {\bibinfo {author} {\bibfnamefont {D.~V.}\ \bibnamefont
  {Else}}, \bibinfo {author} {\bibfnamefont {B.}~\bibnamefont {Bauer}}, \ and\
  \bibinfo {author} {\bibfnamefont {C.}~\bibnamefont {Nayak}},\ }\href@noop {}
  {}\bibinfo {note} {To appear}\BibitemShut {NoStop}%
\bibitem [{\citenamefont {Pekker}\ \emph {et~al.}(2014)\citenamefont {Pekker},
  \citenamefont {Refael}, \citenamefont {Altman}, \citenamefont {Demler},\ and\
  \citenamefont {Oganesyan}}]{Pekker2014}%
  \BibitemOpen
  \bibfield  {author} {\bibinfo {author} {\bibfnamefont {David}\ \bibnamefont
  {Pekker}}, \bibinfo {author} {\bibfnamefont {Gil}\ \bibnamefont {Refael}},
  \bibinfo {author} {\bibfnamefont {Ehud}\ \bibnamefont {Altman}}, \bibinfo
  {author} {\bibfnamefont {Eugene}\ \bibnamefont {Demler}}, \ and\ \bibinfo
  {author} {\bibfnamefont {Vadim}\ \bibnamefont {Oganesyan}},\ }\bibfield
  {title} {\enquote {\bibinfo {title} {{Hilbert-Glass Transition: New
  Universality of Temperature-Tuned Many-Body Dynamical Quantum
  Criticality}},}\ }\href {\doibase 10.1103/PhysRevX.4.011052} {\bibfield
  {journal} {\bibinfo  {journal} {Phys. Rev. X}\ }\textbf {\bibinfo {volume}
  {4}},\ \bibinfo {pages} {011052} (\bibinfo {year} {2014})},\ \Eprint
  {http://arxiv.org/abs/1307.3253} {arXiv:1307.3253} \BibitemShut {NoStop}%
\bibitem [{\citenamefont {Huse}\ \emph {et~al.}(2013)\citenamefont {Huse},
  \citenamefont {Nandkishore}, \citenamefont {Oganesyan}, \citenamefont {Pal},\
  and\ \citenamefont {Sondhi}}]{Huse2013}%
  \BibitemOpen
  \bibfield  {author} {\bibinfo {author} {\bibfnamefont {David~A.}\
  \bibnamefont {Huse}}, \bibinfo {author} {\bibfnamefont {Rahul}\ \bibnamefont
  {Nandkishore}}, \bibinfo {author} {\bibfnamefont {Vadim}\ \bibnamefont
  {Oganesyan}}, \bibinfo {author} {\bibfnamefont {Arijeet}\ \bibnamefont
  {Pal}}, \ and\ \bibinfo {author} {\bibfnamefont {S.~L.}\ \bibnamefont
  {Sondhi}},\ }\bibfield  {title} {\enquote {\bibinfo {title}
  {{Localization-protected quantum order}},}\ }\href {\doibase
  10.1103/PhysRevB.88.014206} {\bibfield  {journal} {\bibinfo  {journal} {Phys.
  Rev. B}\ }\textbf {\bibinfo {volume} {88}},\ \bibinfo {pages} {014206}
  (\bibinfo {year} {2013})},\ \Eprint {http://arxiv.org/abs/1304.1158}
  {arXiv:1304.1158} \BibitemShut {NoStop}%
\bibitem [{\citenamefont {Haag}(1996)}]{Haag1996}%
  \BibitemOpen
  \bibfield  {author} {\bibinfo {author} {\bibfnamefont {Rudolf}\ \bibnamefont
  {Haag}},\ }\href@noop {} {\emph {\bibinfo {title} {{Local Quantum Physics:
  Fields, Particles, Algebras}}}}\ (\bibinfo  {publisher} {Springer},\ \bibinfo
  {year} {1996})\BibitemShut {NoStop}%
\bibitem [{\citenamefont {{Wilczek}}(2013)}]{Wilczek13}%
  \BibitemOpen
  \bibfield  {author} {\bibinfo {author} {\bibfnamefont {F.}~\bibnamefont
  {{Wilczek}}},\ }\bibfield  {title} {\enquote {\bibinfo {title}
  {{Superfluidity and Space-Time Translation Symmetry Breaking}},}\ }\href
  {\doibase 10.1103/PhysRevLett.111.250402} {\bibfield  {journal} {\bibinfo
  {journal} {Physical Review Letters}\ }\textbf {\bibinfo {volume} {111}},\
  \bibinfo {eid} {250402} (\bibinfo {year} {2013})},\ \Eprint
  {http://arxiv.org/abs/1308.5949} {arXiv:1308.5949 [cond-mat.supr-con]}
  \BibitemShut {NoStop}%
\bibitem [{\citenamefont {Pethick}\ and\ \citenamefont {{H.
  Smith}}(2008)}]{Pethick2008}%
  \BibitemOpen
  \bibfield  {author} {\bibinfo {author} {\bibfnamefont {C.~J.}\ \bibnamefont
  {Pethick}}\ and\ \bibinfo {author} {\bibnamefont {{H. Smith}}},\ }\href@noop
  {} {\emph {\bibinfo {title} {{Bose--Einstein Condensation in Dilute
  Gases}}}}\ (\bibinfo  {publisher} {Cambridge University Press},\ \bibinfo
  {year} {2008})\BibitemShut {NoStop}%
\bibitem [{\citenamefont {Nicolis}\ and\ \citenamefont
  {Piazza}(2012)}]{Nicolis2012}%
  \BibitemOpen
  \bibfield  {author} {\bibinfo {author} {\bibfnamefont {Alberto}\ \bibnamefont
  {Nicolis}}\ and\ \bibinfo {author} {\bibfnamefont {Federico}\ \bibnamefont
  {Piazza}},\ }\bibfield  {title} {\enquote {\bibinfo {title} {{Spontaneous
  symmetry probing}},}\ }\href {\doibase 10.1007/JHEP06(2012)025} {\bibfield
  {journal} {\bibinfo  {journal} {J. High Energy Phys.}\ }\textbf {\bibinfo
  {volume} {2012}} (\bibinfo {year} {2012}),\ 10.1007/JHEP06(2012)025},\
  \Eprint {http://arxiv.org/abs/1112.5174} {arXiv:1112.5174 [hep-th]}
  \BibitemShut {NoStop}%
\bibitem [{\citenamefont {Castillo}\ \emph {et~al.}(2014)\citenamefont
  {Castillo}, \citenamefont {Koch},\ and\ \citenamefont
  {Palma}}]{Castillo2014}%
  \BibitemOpen
  \bibfield  {author} {\bibinfo {author} {\bibfnamefont {Esteban}\ \bibnamefont
  {Castillo}}, \bibinfo {author} {\bibfnamefont {Benjamin}\ \bibnamefont
  {Koch}}, \ and\ \bibinfo {author} {\bibfnamefont {Gonzalo}\ \bibnamefont
  {Palma}},\ }\href {http://arxiv.org/abs/1410.2261} {\enquote {\bibinfo
  {title} {{On the dynamics of fluctuations in time crystals}},}\ } (\bibinfo
  {year} {2014}),\ \Eprint {http://arxiv.org/abs/1410.2261} {arXiv:1410.2261}
  \BibitemShut {NoStop}%
\bibitem [{\citenamefont {Thies}(2014)}]{Thies2014}%
  \BibitemOpen
  \bibfield  {author} {\bibinfo {author} {\bibfnamefont {Michael}\ \bibnamefont
  {Thies}},\ }\href {http://arxiv.org/abs/1411.4236} {\enquote {\bibinfo
  {title} {{Semiclassical time crystal in the chiral Gross-Neveu model}},}\ }
  (\bibinfo {year} {2014}),\ \Eprint {http://arxiv.org/abs/1411.4236}
  {arXiv:1411.4236} \BibitemShut {NoStop}%
\bibitem [{\citenamefont {Bruno}(2013{\natexlab{d}})}]{Bruno13c}%
  \BibitemOpen
  \bibfield  {author} {\bibinfo {author} {\bibfnamefont {Patrick}\ \bibnamefont
  {Bruno}},\ }\bibfield  {title} {\enquote {\bibinfo {title} {{Impossibility of
  Spontaneously Rotating Time Crystals: A No-Go Theorem}},}\ }\href {\doibase
  10.1103/PhysRevLett.111.070402} {\bibfield  {journal} {\bibinfo  {journal}
  {Phys. Rev. Lett.}\ }\textbf {\bibinfo {volume} {111}},\ \bibinfo {pages}
  {070402} (\bibinfo {year} {2013}{\natexlab{d}})},\ \Eprint
  {http://arxiv.org/abs/1306.6275} {arXiv:1306.6275} \BibitemShut {NoStop}%
\bibitem [{\citenamefont {{Schreiber}}\ \emph {et~al.}(2015)\citenamefont
  {{Schreiber}}, \citenamefont {{Hodgman}}, \citenamefont {{Bordia}},
  \citenamefont {{L{\"u}schen}}, \citenamefont {{Fischer}}, \citenamefont
  {{Vosk}}, \citenamefont {{Altman}}, \citenamefont {{Schneider}},\ and\
  \citenamefont {{Bloch}}}]{Schreiber15}%
  \BibitemOpen
  \bibfield  {author} {\bibinfo {author} {\bibfnamefont {M.}~\bibnamefont
  {{Schreiber}}}, \bibinfo {author} {\bibfnamefont {S.~S.}\ \bibnamefont
  {{Hodgman}}}, \bibinfo {author} {\bibfnamefont {P.}~\bibnamefont {{Bordia}}},
  \bibinfo {author} {\bibfnamefont {H.~P.}\ \bibnamefont {{L{\"u}schen}}},
  \bibinfo {author} {\bibfnamefont {M.~H.}\ \bibnamefont {{Fischer}}}, \bibinfo
  {author} {\bibfnamefont {R.}~\bibnamefont {{Vosk}}}, \bibinfo {author}
  {\bibfnamefont {E.}~\bibnamefont {{Altman}}}, \bibinfo {author}
  {\bibfnamefont {U.}~\bibnamefont {{Schneider}}}, \ and\ \bibinfo {author}
  {\bibfnamefont {I.}~\bibnamefont {{Bloch}}},\ }\bibfield  {title} {\enquote
  {\bibinfo {title} {{Observation of many-body localization of interacting
  fermions in a quasirandom optical lattice}},}\ }\href {\doibase
  10.1126/science.aaa7432} {\bibfield  {journal} {\bibinfo  {journal}
  {Science}\ }\textbf {\bibinfo {volume} {349}},\ \bibinfo {pages} {842--845}
  (\bibinfo {year} {2015})},\ \Eprint {http://arxiv.org/abs/1501.05661}
  {arXiv:1501.05661 [cond-mat.quant-gas]} \BibitemShut {NoStop}%
\bibitem [{\citenamefont {Smith}\ \emph {et~al.}(2016)\citenamefont {Smith},
  \citenamefont {Lee}, \citenamefont {Richerme}, \citenamefont {Neyenhuis},
  \citenamefont {Hess}, \citenamefont {Hauke}, \citenamefont {Heyl},
  \citenamefont {Huse},\ and\ \citenamefont {Monroe}}]{Smith2016}%
  \BibitemOpen
  \bibfield  {author} {\bibinfo {author} {\bibfnamefont {J.}~\bibnamefont
  {Smith}}, \bibinfo {author} {\bibfnamefont {A.}~\bibnamefont {Lee}}, \bibinfo
  {author} {\bibfnamefont {P.}~\bibnamefont {Richerme}}, \bibinfo {author}
  {\bibfnamefont {B.}~\bibnamefont {Neyenhuis}}, \bibinfo {author}
  {\bibfnamefont {P.~W.}\ \bibnamefont {Hess}}, \bibinfo {author}
  {\bibfnamefont {P.}~\bibnamefont {Hauke}}, \bibinfo {author} {\bibfnamefont
  {M.}~\bibnamefont {Heyl}}, \bibinfo {author} {\bibfnamefont {D.~A.}\
  \bibnamefont {Huse}}, \ and\ \bibinfo {author} {\bibfnamefont
  {C.}~\bibnamefont {Monroe}},\ }\bibfield  {title} {\enquote {\bibinfo {title}
  {Many-body localization in a quantum simulator with programmable random
  disorder},}\ }\href {\doibase 10.1038/nphys3783} {\bibfield  {journal}
  {\bibinfo  {journal} {Nat Phys}\ }\textbf {\bibinfo {volume} {12}},\ \bibinfo
  {pages} {907} (\bibinfo {year} {2016})}\BibitemShut {NoStop}%
\bibitem [{\citenamefont {Choi}\ \emph {et~al.}(2016)\citenamefont {Choi},
  \citenamefont {Hild}, \citenamefont {Zeiher}, \citenamefont {Schau{\ss}},
  \citenamefont {Rubio-Abadal}, \citenamefont {Yefsah}, \citenamefont
  {Khemani}, \citenamefont {Huse}, \citenamefont {Bloch},\ and\ \citenamefont
  {Gross}}]{Choi2016}%
  \BibitemOpen
  \bibfield  {author} {\bibinfo {author} {\bibfnamefont {Jae-yoon}\
  \bibnamefont {Choi}}, \bibinfo {author} {\bibfnamefont {Sebastian}\
  \bibnamefont {Hild}}, \bibinfo {author} {\bibfnamefont {Johannes}\
  \bibnamefont {Zeiher}}, \bibinfo {author} {\bibfnamefont {Peter}\
  \bibnamefont {Schau{\ss}}}, \bibinfo {author} {\bibfnamefont {Antonio}\
  \bibnamefont {Rubio-Abadal}}, \bibinfo {author} {\bibfnamefont {Tarik}\
  \bibnamefont {Yefsah}}, \bibinfo {author} {\bibfnamefont {Vedika}\
  \bibnamefont {Khemani}}, \bibinfo {author} {\bibfnamefont {David~A.}\
  \bibnamefont {Huse}}, \bibinfo {author} {\bibfnamefont {Immanuel}\
  \bibnamefont {Bloch}}, \ and\ \bibinfo {author} {\bibfnamefont {Christian}\
  \bibnamefont {Gross}},\ }\bibfield  {title} {\enquote {\bibinfo {title}
  {Exploring the many-body localization transition in two dimensions},}\ }\href
  {\doibase 10.1126/science.aaf8834} {\bibfield  {journal} {\bibinfo  {journal}
  {Science}\ }\textbf {\bibinfo {volume} {352}},\ \bibinfo {pages} {1547--1552}
  (\bibinfo {year} {2016})}\BibitemShut {NoStop}%
\bibitem [{\citenamefont {{Sacha}}(2015)}]{Sacha15}%
  \BibitemOpen
  \bibfield  {author} {\bibinfo {author} {\bibfnamefont {K.}~\bibnamefont
  {{Sacha}}},\ }\bibfield  {title} {\enquote {\bibinfo {title} {{Modeling
  spontaneous breaking of time-translation symmetry}},}\ }\href {\doibase
  10.1103/PhysRevA.91.033617} {\bibfield  {journal} {\bibinfo  {journal}
  {\pra}\ }\textbf {\bibinfo {volume} {91}},\ \bibinfo {eid} {033617} (\bibinfo
  {year} {2015})},\ \Eprint {http://arxiv.org/abs/1410.3638} {arXiv:1410.3638
  [cond-mat.quant-gas]} \BibitemShut {NoStop}%
\bibitem [{\citenamefont {Shirai}\ \emph {et~al.}(2016)\citenamefont {Shirai},
  \citenamefont {Thingna}, \citenamefont {Mori}, \citenamefont {Denisov},
  \citenamefont {H{\"{a}}nggi},\ and\ \citenamefont {Miyashita}}]{Shirai2015}%
  \BibitemOpen
  \bibfield  {author} {\bibinfo {author} {\bibfnamefont {Tatsuhiko}\
  \bibnamefont {Shirai}}, \bibinfo {author} {\bibfnamefont {Juzar}\
  \bibnamefont {Thingna}}, \bibinfo {author} {\bibfnamefont {Takashi}\
  \bibnamefont {Mori}}, \bibinfo {author} {\bibfnamefont {Sergey}\ \bibnamefont
  {Denisov}}, \bibinfo {author} {\bibfnamefont {Peter}\ \bibnamefont
  {H{\"{a}}nggi}}, \ and\ \bibinfo {author} {\bibfnamefont {Seiji}\
  \bibnamefont {Miyashita}},\ }\bibfield  {title} {\enquote {\bibinfo {title}
  {{Effective Floquet--Gibbs states for dissipative quantum systems}},}\ }\href
  {\doibase 10.1088/1367-2630/18/5/053008} {\bibfield  {journal} {\bibinfo
  {journal} {New J. Phys.}\ }\textbf {\bibinfo {volume} {18}},\ \bibinfo
  {pages} {053008} (\bibinfo {year} {2016})},\ \Eprint
  {http://arxiv.org/abs/1511.06864} {arXiv:1511.06864} \BibitemShut {NoStop}%
\bibitem [{\citenamefont {Lieb}\ and\ \citenamefont
  {Robinson}(1972)}]{Lieb1972}%
  \BibitemOpen
  \bibfield  {author} {\bibinfo {author} {\bibfnamefont {Elliott~H.}\
  \bibnamefont {Lieb}}\ and\ \bibinfo {author} {\bibfnamefont {Derek~W.}\
  \bibnamefont {Robinson}},\ }\bibfield  {title} {\enquote {\bibinfo {title}
  {{The finite group velocity of quantum spin systems}},}\ }\href {\doibase
  10.1007/BF01645779} {\bibfield  {journal} {\bibinfo  {journal} {Commun. Math.
  Phys.}\ }\textbf {\bibinfo {volume} {28}},\ \bibinfo {pages} {251--257}
  (\bibinfo {year} {1972})}\BibitemShut {NoStop}%
\bibitem [{\citenamefont {Nachtergaele}\ and\ \citenamefont
  {Sims}(2006)}]{Nachtergaele2006}%
  \BibitemOpen
  \bibfield  {author} {\bibinfo {author} {\bibfnamefont {Bruno}\ \bibnamefont
  {Nachtergaele}}\ and\ \bibinfo {author} {\bibfnamefont {Robert}\ \bibnamefont
  {Sims}},\ }\bibfield  {title} {\enquote {\bibinfo {title} {{Lieb-Robinson
  bounds and the exponential clustering theorem}},}\ }\href {\doibase
  10.1007/s00220-006-1556-1} {\bibfield  {journal} {\bibinfo  {journal}
  {Commun. Math. Phys.}\ }\textbf {\bibinfo {volume} {265}},\ \bibinfo {pages}
  {119--130} (\bibinfo {year} {2006})},\ \Eprint
  {http://arxiv.org/abs/math-ph/0506030} {arXiv:math-ph/0506030 [math-ph]}
  \BibitemShut {NoStop}%
\end{thebibliography}
\end{document}